\documentclass[11pt]{article}

    \usepackage{amsmath, amssymb,amsfonts,xspace,graphicx,relsize,bm,mathtools,xcolor}
    \usepackage{mathrsfs}
    \usepackage[pagebackref]{hyperref}
    \usepackage{enumerate}
    \usepackage{bm}

    \usepackage{amsthm,cleveref}
    \usepackage{lipsum} 
    \newtheorem{theorem}{Theorem}
    \newtheorem{result}{Result}

    \def\01{\{0,1\}}

    \newcommand{\eps}{\varepsilon}

    \newcommand{\Tr}{\mbox{\rm Tr}}
    \newcommand{\sym}[1]{\mathsf{Sym}_{#1}}
    \newcommand{\mat}[1]{\mathsf{Mat}_{#1}}

    \newcommand{\diag}{\mbox{\rm diag}}
    
    \newcommand{\norm}[1]{\mbox{$\|{#1}\|$}}
    \DeclareMathOperator*{\E}{\mathbb{E}}
    \newcommand{\bx}{\boldsymbol{x}}
    \newcommand{\bg}{\boldsymbol{g}}
    \newcommand{\by}{\boldsymbol{y}}
    \newcommand{\bz}{\boldsymbol{z}}
    \newcommand{\bh}{\boldsymbol{h}}
    \newcommand{\bs}{\boldsymbol{s}}
    \newcommand{\bu}{\boldsymbol{u}}
    \newcommand{\bb}{\boldsymbol{b}}

    \def\X{\mathcal{X}}
    
    \def\H{\mathcal{H}}
    \newcommand{\Exp}{\mathbb{E}}

    \newcommand{\PRG}{\ensuremath{\mathsf{PRG}}}
    \newcommand{\PTF}{\ensuremath{\mathsf{PTF}}}

    \newcommand{\U}{\ensuremath{\mathcal{U}}}
    \newcommand{\G}{\ensuremath{\mathcal{G}}}
    
    \newcommand{\Hi}{\ensuremath{\mathcal{H}}}

    \newcommand{\calE}{\ensuremath{\mathcal{E}}}
    
    \newcommand{\De}{\ensuremath{\mathcal{D}}}
    \newcommand{\R}{\ensuremath{\mathbb{R}}}
    
    \newcommand{\Oh}{\ensuremath{\mathcal{O}}}
    \newcommand{\id}{\ensuremath{\mathbb{I}}}
    \newcommand{\ophi}{\ensuremath{\overline{g}}}

    \DeclareMathOperator{\poly}{poly}
    
    \newcommand{\Fe}{\ensuremath{\mathcal{F}}}
    \newtheorem{definition}[theorem]{Definition}
    \newtheorem{fact}[theorem]{Fact}
    \newtheorem{lemma}[theorem]{Lemma}
    \newtheorem{proposition}[theorem]{Proposition}
    \newtheorem{corollary}[theorem]{Corollary}
    
    \newtheorem{claim}[theorem]{Claim}
    \newcommand{\pmset}[1]{\{-1,1\}^{#1}} 

    \def\01{\{0,1\}}

    \newcommand{\NS}{\mathsf{NS}}

    \newcommand{\GSA}{\mathsf{GSA}}
    \newcommand{\PSD}{\mathsf{PSD}}
    \newcommand{\NSD}{\mathsf{NSD}}
    \newcommand{\AS}{\mathsf{AS}}
    \newcommand{\vol}{\mathsf{vol}}

    \DeclareMathOperator{\sign}{sign}
    \DeclareMathOperator{\out}{out}
    
    \usepackage[noline, boxed]{algorithm2e}

    \newcommand {\minusspace} {\: \! \!}
    \newcommand{\reals}{{\mathbb R}}
    \newcommand {\br} [1] {\ensuremath{ \left( #1 \right) }}
    \newcommand {\normsub} [2] {\ensuremath{ \norm{#1}_{#2} }}
    \newcommand {\onenorm} [1] {\normsub{#1}{1}}
    \newcommand {\twonorm} [1] {\normsub{#1}{2}}
    \newcommand {\abs} [1] {\ensuremath{ \left| #1 \right| }}
    \newcommand {\expec} [1] {\ensuremath{\mathbb{E}\left[#1\right]}}
    \newcommand {\set} [1] {\ensuremath{ \left\lbrace #1 \right\rbrace }}

    \newcommand {\prob} [1] {\Fn{\mathrm{Pr}}{#1}}
    
    \newcommand {\Fn} [2] {\ensuremath{ #1 \minusspace \Br{ #2 } }}
    \newcommand {\Br} [1] {\ensuremath{ \left[ #1 \right] }}
    \newcommand{\braket}[1] {\ensuremath {\langle #1 \rangle}}

    \usepackage[margin=1in]{geometry}
    \hypersetup{
    	colorlinks,
    	linkcolor={red!100!black},
    	citecolor={blue!100!black},
    }

\begin{document}
    	
    	\title{Positive spectrahedra: \\ Invariance principles and Pseudorandom generators}

    \author{
    Srinivasan Arunachalam\\[2mm]
     IBM Quantum.\\
     \small IBM T.J. Watson Research Center\\ \small Yorktown Heights, USA\\
    \small \texttt{Srinivasan.Arunachalam@ibm.com}
    \and
    Penghui Yao\\[2mm]
     State Key Laboratory for\\ Novel
    Software Technology,\\ \small Nanjing
    University\\
    \small \texttt{pyao@nju.edu.cn}
    }

    \date{\today}

    	\maketitle
    \begin{abstract}
     In a recent work, O'Donnell, Servedio and Tan (STOC 2019) gave explicit pseudorandom generators ($\PRG$s) for arbitrary $m$-facet polytopes in $n$ variables with seed length poly-logarithmic in $m,n$, concluding a sequence of works in the last decade, that was started by Diakonikolas,  Gopalan,  Jaiswal,  Servedio, Viola (SICOMP 2010) and Meka, Zuckerman (SICOMP 2013) for fooling linear and polynomial threshold functions, respectively. In this work, we consider a natural extension of  $\PRG$s for intersections of positive spectrahedra. A positive spectrahedron  is a Boolean function $f(x)=[x_1A^1+\cdots +x_nA^n \preceq B]$ where the $A^i$s are $k\times k$ positive semidefinite matrices. We construct explicit $\PRG$s  that $\delta$-fool  ``regular" width-$M$ positive spectrahedra (i.e., when none of the $A^i$s are dominant) over the Boolean space  with seed length~$\poly(\log k,\log n, M, 1/\delta)$.
    \vspace{1.5mm}

    Our main technical contributions are the following: We first prove an invariance principle for positive spectrahedra via the well-known Lindeberg method. As far as we are aware such a  generalization of the Lindeberg method was unknown. Second, we prove an upper bound on noise sensitivity and a Littlewood-Offord theorem for positive spectrahedra. Using these results, we give applications for constructing $\PRG$s for positive spectrahedra, learning theory, discrepancy sets for positive spectrahedra (over the Boolean cube) and $\PRG$s for  intersections of structured polynomial threshold~functions.
    \end{abstract}

    \pagenumbering{gobble}

    \newpage
     \setcounter{tocdepth}{2}
    \renewcommand{\baselinestretch}{0.95}\normalsize
    \tableofcontents
    \renewcommand{\baselinestretch}{1.0}\normalsize
    \pagenumbering{gobble}
    \clearpage

    \pagenumbering{arabic}
    \setcounter{page}{1}

    \section{Introduction}
    Constructing explicit pseudorandom generators $(\PRG)$ for a class of interesting Boolean functions has received tremendous attention in the last few decades.  One particular class of functions that has seen a flurry of works is the class of halfspaces. A \emph{halfspace} is a Boolean function $f:\pmset{n}\rightarrow \set{0,1}$ that can be expressed as $f(x)=\sign(a_1x_1+\cdots+a_nx_n-b)$ for some real  values $a_1,\ldots,a_n,b\in~\R$. Halfspaces arise naturally in many  areas of theoretical computer science including
    machine learning, communication complexity, circuit complexity and pseudorandomness.  A successful line of work~\cite{servedio2006every,diakonikolas2010bounding,meka2013pseudorandom,kothari2015almost, gopalan2018pseudorandomness} resulted in $\PRG$s that $\varepsilon$-fool halfspaces with seed length poly-logarithmic in $(n/\varepsilon)$ over the Boolean space.

    Given the success in designing $\PRG$s for single halfspaces (or linear threshold function), two alternate lines of work received a lot of attention, \emph{polynomial threshold functions} and \emph{intersections} of halfspaces. A degree-$d$ polynomial threshold function ($\PTF$) is simply a function $f(x)=\sign(p(x))$ where $p$ is a degree-$d$ polynomial. In this direction, there have been a sequence of works~\cite{diakonikolas2010bounded,diakonikolas2010bounding,kane2010k,kane2011Gaussian,DBLP:conf/coco/Kane11,DBLP:conf/focs/Kane11,kane2014pseudorandom,o2020fooling} that produced $\PRG$s with seed length exponential in $d$ over the Boolean space and quasi-polynomial in $d$ over the Gaussian space. Alternatively, another line of work considered \emph{intersections} of halfspaces (i.e., a polytope).  In this direction, a sequence of works~\cite{gopalan2010fooling,harsha2013invariance,servedio2017fooling,chattopadhyay2019simple,o2019fooling} produced a $\PRG$ for $m$-facet polytopes in $n$ variables with seed length poly-logarithmic in $m,n$.

    In this work, we  initiate the construction of $\PRG$s for spectrahedra: a natural generalization of halfspaces, polytopes and $\PTF$s in one framework.
    A \emph{spectrahedron} $S\subseteq\R^n$ is a feasible region of a \emph{semidefinite program}. Namely,
    $$
    S=\set{x\in\R^n:\sum_ix_iA^i\preceq B}
    $$
    for some $k\times k$ symmetric matrices $A^1,\ldots, A^n, B$, where $\preceq$ is the standard L\"owner ordering.\footnote{In this ordering, we say $A\preceq B$ if $B-A$ is positive semidefinite, i.e., all the eigenvalues of $B-A$ are~non-negative.} We say~$S$ is a \emph{positive spectrahedron} if either all $A^i$s are positive semidefinite  ($\PSD$) or all $A^i$s are negative semidefinite. spectrahedra are important basic objects in polynomial optimization and algebraic geometry~\cite{doi:10.1137/1.9781611972290,scheiderer2018spectrahedral}. Mathematically, spectrahedra have rich and complicated structures and include well-known geometric objects like polytopes, cylinders, polyhedrons, elliptopes.  Computationally, semidefinite programming has found many applications in theoretical computer science in the field of optimization~\cite{arora2007combinatorial}, approximation theory~\cite{DBLP:journals/jacm/GoemansW95,gartner2012approximation}, algorithms~\cite{arora2005fast,10.1145/3357713.3384338}, SoS hierarchy~\cite{barak2019nearly}, extension complexity~\cite{fiorini2015exponential,lee2015lower}. The class of semidefinite programs that consists of only $\PSD$ matrices is an important class of SDPs, termed as {\em positive semidefinite programs}, which has been used to characterize various quantum interactive proof systems~\cite{ 5438601,10.1145/2049697.2049704,GXU:2013}. Their   computational complexity has also received a lot of attention in the past decade~\cite{6108207,10.1145/2312005.2312026,doi:10.1137/1.9781611974331.ch127,10.1145/3357713.3384338}.
    But in several ways, our understanding of spectrahedra is at an early stage and seriously lags behind our understanding of polytopes.  Many basic geometric properties of spetrahedrons, such as dimensions, numbers of connected components, matrix ranks~\cite{cynthia} are not well understood, even basic properties such as proving the membership of spetrahedrons for some geometric objects is highly non-trivial~\cite{Nie2008}.

    Our main result in this work is $\PRG$s for regular positive spectrahedra with seed length poly-logarithmic in $n$ and $k$, which we define in Section~\ref{sec:thisworkoverview}. Before stating our main  results, we briefly discuss the techniques developed by prior works to construct $\PRG$s for polytopes before discussing the challenges we need to~handle here.

    \subsection{Prior work and conceptual challenges}
    \subsubsection{Prior work}
    One of the earliest works that considered fooling  threshold functions was by Meka-Zuckerman~\cite{meka2013pseudorandom} and~\cite{diakonikolas2010bounded}. A powerful technique that Meka-Zuckerman introduced was a general recipe to construct $\PRG$s for functions $f$ via \emph{invariance principles}. Roughly speaking, an invariance principle for a function $f:\pmset{n}\rightarrow \01$ states that, the expected value of $f(\U^n)$ (where the input is uniformly random in $\pmset{n}$) is close to the expected value of~$f(\G^n)$ (where the input is a standard $\G^n=\mathcal{N}(0,1)^n$ Gaussian). Invariance theorems are generalizations of the classic Berry-Esseen central limit theorem, proven using the well-known Lindeberg method~\cite{lindeberg}. The versatile framework of~\cite{meka2013pseudorandom} allows one to use invariance principles along with a few more ingredients to construct $\PRG$s, so the technical challenge is in establishing invariance principles.

    Using this framework, Harsha, Klivans and Meka~\cite{harsha2013invariance} proved  an \emph{invariance principle for regular polytopes} (i.e., when the coefficients in (all) the halfspaces are ``regular"). The main novelty in their work was  the \emph{poly-logarithmic} (in the input parameters) error dependence. In order to prove this, they first proved a general invariance principle for smooth functions (over polytopes).  Subsequently they instantiate their invariance principle for the so-called {\em Bentkus mollifier}~\cite{bentkus1990smooth},\footnote{The Bentkus \emph{mollifier} is a function which provides a ``smooth" continuous approximation to the the discrete multivariate indicator function (also referred to as \emph{orthant functions}). We discuss this further below.} crucially relying on the fact that the mollifier has derivatives that scale poly-logarithmic in the input size. Finally in order to go from invariance principles (for the mollifier)  to fooling regular polytopes, they need to prove an \emph{anti-concentration of polytopes} in the Gaussian space. For this, they use (as a black-box) a well-known result of Nazarov~\cite{nazarov2003maximal,klivans2008learning}, which bounds the \emph{Gaussian surface area} $(\GSA)$ of polytopes. Putting together the invariance principle for smooth functions, Bentkus mollifier and Nazarov's bound on $\GSA$,~\cite{harsha2013invariance} obtained their main results for regular polytopes. We discuss this proof idea in more detail in Section~\ref{sec:introhkm}.

    Subsequently, Servedio and Tan~\cite{servedio2017fooling} improved the results of~\cite{harsha2013invariance} by considering ``low-weight" polytopes, which removes the regularity condition (albeit, with the seed length of the $\PRG$ in~\cite{servedio2017fooling} depending on the weight). Finally, O'Donnell, Servedio and Tan~\cite{o2019fooling} showed how to fool arbitrary polytopes. In~\cite{o2019fooling} they bypass the entire Gaussian  space (in fact it is a \emph{necessity} to avoid this Gaussian space since standard invariance principles do not hold for non-regular polytopes) and proved a ``Boolean-invariance principle" for the Bentkus mollifier. Although they bypass the Gaussian intermediate (which is standard in invariance principles), their proof techniques still use the Lindeberg method. Additionally, a crucial tool introduced by them was the Boolean anti-concentration of polytopes, since they can no longer use the $\GSA$ bound of Nazarov which used by~\cite{harsha2013invariance,servedio2017fooling,chattopadhyay2019simple} for \emph{Gaussian} anti-concentration.

    \subsubsection{PRGs for spectrahedra: Conceptual challenges}
    There are two straightforward approaches to constructing $\PRG$s for positive spectrahedra. The first is to write a spectrahedron as a linear program. Naturally one can approximate a positive-semidefinite constraint $X\succeq 0$ of a $k\times k$ symmetric matrix with exponentially many constraints $z^T \hspace{0.5mm}X\hspace{0.5mm} z\geq 0$ for $z\in~\R^k$. However the results of~\cite{harsha2013invariance,o2019fooling} would be moot here since the seed-lengths of their $\PRG$s are poly-logarithmic in the number of constraints, which is polynomial in the dimension $k$, while our goal it to have a seed length \emph{poly-logarithmic} in $k$.  The second approach is to use  Sylvester’s criterion to write out $k$ polynomials of degree at most $k$ (corresponding to the $k$ determinantal representation of the $k$ minors) and one could use $\PRG$s for polynomial threshold functions ($\PTF$).    However, finding optimal $\PRG$s for $\PTF$s has remains open and the best-known $\PRG$s we have for degree-$k$  $\PTF$s over the Boolean space depends \emph{exponentially} in $k$~\cite{meka2013pseudorandom}.

    This naturally motivates us to use the ``eigenstructure" of $X\succeq 0$ crucially in understanding spectrahedra.  The next line of approach is to use the existing invariance-principle framework of~\cite{meka2013pseudorandom}  which we overviewed in the previous section, but this opens up a few challenges:
    \begin{enumerate}

        \item \textbf{Invariance principles:} Since a spectrahedron naturally deals with eigenvalues of matrices, it is unclear if we could use known invariance principles for spectrahedra. In fact, we are not even aware of a generalization of the Lindeberg-type argument to show an invariance principle for \emph{spectral functions} (i.e., functions that act on the eigenspectra of~matrices).

        \item \textbf{Geometric properties:} Prior works of~\cite{klivans2008learning,harsha2013invariance,servedio2017fooling,chattopadhyay2019simple} crucially used the work of Nazarov~\cite{nazarov2003maximal} which bounds the Gaussian surface area of polytopes in order to prove their anti-concentration. However, spectrahedra are very poorly understood, and even more basic questions about their average sensitivity, noise sensitivity, surface area are~unknown.

        \item \textbf{Anti-concentration:} An important  technique for constructing $\PRG$s using invariance principles requires one to prove \emph{anti-concentration}, i.e., when moving from the smooth mollifiers to the orthant functions a crucial ingredient is anti-concentration. It is far from clear if spectrahedra enjoy such nice properties in either Boolean spaces or Gaussian spaces.
    \end{enumerate}
    As far as we are aware, none of these questions have been considered for any class of spectrahedra except polytopes. Our main contribution is to make significant progress in all these questions for the class of  positive spectrahedra.

    \subsection{Our main result}
    \label{sec:thisworkoverview}
    In order to state our main result we first define $\PRG$s and $(\tau,M)$-regular spectrahedra. A \emph{pseudorandom generator} is a function $G: \pmset{r}\rightarrow \pmset{n}$ and is said to $\varepsilon$-fool  a class of functions $\Fe\subseteq \{f:\pmset{n}\rightarrow \01\}$ with \emph{seed length} $r$ if it satisfies the following: for every $f\in \Fe$, we~have
    $$
    \abs{\Pr_{\bx\sim \U_n}[f(\bx)=1]-\Pr_{\by\sim \U_r}[f\br{G(\by)}=1]}\leq \varepsilon,
    $$
    where $\U_n$ (resp.~$\U_r$) corresponds to uniform distribution over $\pmset{n}$ (resp.~$\pmset{r}$). We next define the class of regular positive spectrahedra. Given $\tau, M>0$, we say a sequence of $k\times k$ positive semidefinite matrices $\br{A^1,\ldots, A^n}$ is \emph{$\br{\tau, M}$-regular} if
      \begin{equation}\label{eqn:regularity}
      \id\preceq\sum_{i=1}^n \br{A^i}^2\preceq M\cdot\id \hspace{2mm}~\mbox{and}~\hspace{2mm}  A^i\preceq\tau\cdot\id \text{ for every } i\in [n].
    \end{equation}
    This regularity assumption is a very natural assumption, it says that the \emph{width} of a semidefinite program defined by these matrices is bounded.  We remark that our regularity condition naturally extends (and is in fact \emph{less} restrictive) the regularity condition that was used in prior works on fooling halfspaces and polytopes~\cite{gopalan2010fooling,diakonikolas2010bounded,meka2013pseudorandom,harsha2013invariance}. In Section~\ref{intro:booleananti} we discuss more about why this notion of regularity is necessary and sufficient for our proof techniques.

    A \emph{spectrahedron} $S\subseteq\reals^n$ is a feasible region  of the convex set  $S=\set{x\in\R^n:\sum_ix_iA^i\preceq B}$.\footnote{For simplicity in exposition, we assume here that $\|B\|\leq  M$ (our main theorems depend on the norm of~$B$).} We say $S$ is a \emph{positive spectreheron} if either all $A^i$s are positive semidefinite ($\PSD)$ or all $A^i$s are negative semidefinite. We say $S$ is a \emph{$(\tau,M)$-regular positive spectrahedron} if $(A^1,\ldots,A^n)$ are $(\tau,M)$ regular. It is also natural to consider an \emph{intersection} of positive spectrahedra $S_1,\ldots,S_t$. However, without loss of generality one can assume that  $t=2$ since one can ``pack" all the $S_i$s with $\PSD$ matrices into a larger block diagonal matrix with dimension $t\cdot k$ and similarly all the negative semidefinite matrices, so we can always assume we are working with an intersection of two positive spectrahedra.\footnote{Crucially we remark that the seed length of our $\PRG$ has dependence only logarithmic in $k$, so even with an intersection of $t$ positive spectrahedra, the dependence would be logarithmic in $t$ as well.}
    For simplicity, in the introduction we assume that we are working with a single regular positive spectrahedron here and state our main theorem.
    \begin{result} [PRG for positive spectrahedra]
    \label{thm:informalprg}
    There exists a $\PRG$ $G:\01^r\rightarrow \pmset{n}$~with seed~length
    $$
    r=O(\log n\cdot \log k\cdot M\cdot 1/\delta)
    $$
    that $\delta$-fools $(\tau,M)$-regular positive spectrahedra for $\tau\leq \poly(\delta/(M\cdot \log k))$.
    \end{result}
    Typically,  handling the ``regular case" is the first step towards  obtaining optimal results in pseudorandom generators for geometric objects and we have accomplished that here for the first time. To prove this theorem, we follow the well-known three-step approach and prove the following:
    \begin{enumerate}
        \item An invariance principle for the Bentkus mollifier of {\em arbitrary} regular spectrahedra.

        \item  Boolean and Gaussian anti-concentration  for \emph{positive} regular spectrahedra.

        \item An invariance principle for \emph{positive} regular spectrahedra
    \end{enumerate}
    Before proving these statements, we first overview the~\cite{harsha2013invariance,o2019fooling} approach to proving invariance principles (since our high-level ideas are inspired by their works).

    \subsection{Sketch of the~\cite{harsha2013invariance} invariance principle for polytopes}
    \label{sec:introhkm}
    First recall that a polytope is the feasible region of the set $\{x\in \R^n:Wx\leq b\}$ for a fixed $W\in \R^{n\times n},b\in \R^n$.\footnote{For simplicity, we assume that the number of constraints and variables are equal. Their analysis is more general.} We say a polytope is $\tau$-regular if each row $W^i$ satisfies $\|W^i\|_2=1$ and $\|W^i\|_4\leq \tau$. At a high-level the~\cite{harsha2013invariance} invariance principle states the following:
    \begin{align}
    \label{eq:mainresulthkmintro1}
    \abs{\Pr_{\bx\sim \U_n}[W\bx\leq b] - \Pr_{\bg\sim \G^n}[W\bg\leq b]} \leq \poly(\log n,\tau).
    \end{align}
    To show this, they first express the \emph{orthant} function above (which we denote $\Oh:\R^n~\rightarrow~\01)$, as     $[W\bx\leq b]=[W^1\bx\leq b_1]\cdots [W^n\bx\leq b_n]$. Given this structure, they now use the well-known Lindeberg  method~\cite{lindeberg} (see~\cite{o2014analysis,taolindeberg} for a detailed exposition) to move from the uniform distribution over a Boolean space to the Gaussian space. To establish  Eq.~\eqref{eq:mainresulthkmintro1}, they follow a three-step approach: (1) First, they prove a version of  Eq.~\eqref{eq:mainresulthkmintro1} for \emph{smooth} functions $\widetilde{\Oh}:\R^n\rightarrow \R$ (i.e., functions who have bounded multivariate derivatives). In particular, they use the Lindeberg method to show that the expected value of $\widetilde{\Oh}(W\bx)$ for $x\sim \U_n$, is ``close" to the expected value of $\widetilde{\Oh}(W\bg)$ for $\bg\sim\G^n$. To understand this closeness, they write out $\widetilde{\Oh}(W \bz)$ using the standard multivariate Taylor expansion and bound the  distance between $\widetilde{\Oh}(W\bx)$ and $\widetilde{\Oh}(W\bg)$ by the higher-order derivatives of the smooth function $\widetilde{\Oh}$. (2) Second, they observe that a result of Bentkus~\cite{bentkus1990smooth} provides exactly an approximator $\widetilde{\Oh}:\R^n\rightarrow \R$ (which we refer to as the \emph{Bentkus mollifier}) which serves as a \emph{smooth} approximation to the $\01$-valued orthant function $\Oh(x)=~[W\bx\leq b]$. Additionally this mollifier \emph{crucially} satisfies the property that $\|\widetilde{\Oh}^{(\ell)}\|_1\leq O\br{\log^\ell n}$.\footnote{Here $\|f^{(\ell)}\|_1$ is the $1$-norm of the coefficients in the $\ell$-th derivative. In~\cite{harsha2013invariance}, they care about $\|f^{(4)}\|_1=\max_x \sum_{p,q,r,s}\abs{\partial_p\partial_q\partial_r\partial_s f(x)}$.} (3) So far they established that the Bentkus mollifier (which served as a proxy for $[W\bx\leq b]$) satisfies an approximate version of Eq.~\eqref{eq:mainresulthkmintro1}. In order to go from being close with respect to this Bentkus mollifier to multidimensional CDF closeness, they prove \emph{Gaussian} anti-concentration of polytopes. For this, they use a result of Nazarov~\cite{nazarov2003maximal} (as a black-box) which shows that the Gaussian surface area of a polytope is $O(\sqrt{\log n})$. These three steps allow them to prove Eq.~\eqref{eq:mainresulthkmintro1}.

    \vspace{-1.5mm}

    \subsection{First contribution: Invariance principle for Bentkus mollifier}
    \label{sec:benktkusintro}
     We begin by defining spectral functions. Let $f:\R^k\rightarrow \R$, we say $\psi:\sym{k}\rightarrow \R$ is a  \emph{spectral function} if $\psi(M)=f (\lambda(M))$ for all~$M\in \sym{k}$ where $\lambda(M)=(\lambda_1,\ldots,\lambda_k)$ are the $k$ eigenvalues of~$M$.  In other words, a spectral function $\psi (\cdot)$ depends on a function $\psi$ applied to the eigenvalues of its argument.
    We say $f$ satisfies an invariance principle if
    \[\E_{\bx\sim\U_n}\Br{\psi\br{\sum_i\bx_iA^i-B}}\approx_{\eps}\E_{\bg\sim\G^n}\Br{\psi\br{\sum_i\bg_iA^i-B}},\]
    for symmetric matrices $A_1,\ldots,A_n, B$.
    A conceptual challenge in proving an invariance principle even for smooth spectral functions is that standard Lindeberg-style proofs of invariance theorems use multivariate Taylor series of the mollifier function cannot be used here, since our functions act on the \emph{eigenvalues} of matrices.
    In the past, there have been various invariance principles~\cite{mossel2005noise,DBLP:conf/focs/Mossel08,isaksson2012maximally,harsha2013invariance, yao2019doubly} but none of them apply here; as far as we are aware invariance principles with non-diagonal $A^i,B$ have not been studied.
    In this work, we overcome this challenge and adapt the Lindeberg-style proofs of probabilistic invariance principles to prove its analogue for  spectral~functions.

    To this end,  recall that we are concerned with spectrahedra whose feasible regions are given by $\{x\in \R^n:\sum_i x_i A^i \preceq B\}$, which can alternatively be written as $\{x: \lambda_{\max}\br{\sum_i x_i A^i - B}\leq 0\}$. So we let our spectral function $f:\R^k\rightarrow \R$ to be  $f(\lambda)=[\max_i \lambda_i\leq 0]$ (recall that although our spectrahedron acts on $n$ bits on which we want to prove an invariance principle, our spectral function acts only on the $k$ eigenvalues). For this function, we can still use the Bentkus mollifier $\widetilde{\Oh}:\R^k\rightarrow \R$ as a smooth approximation to $f$.\footnote{In fact our analysis can allow arbitrary orthant functions which can be approximated  by a Bentkus mollifier.} So our first main contribution is to prove an invariance principle for the Bentkus mollifier applied to the spectra of matrices. We remark that in contrast to~\cite{harsha2013invariance}, we do not prove a general invariance principle for spectral functions, instead our spectral function is tailored for the Bentkus mollifier (which is also the case for~\cite{o2019fooling}).

    \vspace{-2mm}
    \paragraph{Fr\'echet derivatives.}
    Since our Bentkus mollifier is acting on the eigenspectra  of matrices, instead of multivariate Taylor expansion, we adopt \emph{Fr\'echet derivatives}, a notion of derivatives that is studied in Banach spaces. Unfortunately, Fr\'echet series (in contrast to standard multivariate series) are still not well understood. In fact even basic properties such as  continuity, Lipschitz continuity, differentiability, continuous differentiability, were only proven in the last three decades~\cite{s002200050279,doi:10.1287/moor.21.3.576,BHATIA1999231,doi:10.1137/S1052623400380584}, which have been well-known for centuries in standard calculus. In particular, even a succinct representation of high-order Fr\'echet derivatives~\cite{sendov2007higher,AMES20104288,AMES2012722,AMES20161459} for \emph{spectral} functions only appeared in the last decade.

    Fortunately for us, Sendov~\cite{sendov2007higher} provided a tensorial representation of high-order Fr\'echet series for spectral functions which we employ to analyze the Fr\'echet derivatives of the Bentkus mollifier. The challenge is in bounding the 3-tensors that appears in Sendov's theorem, which produce~$7$ terms corresponding to different permutations of the tensors after simplification. Three of these~$7$ terms can simply be upper bounded by $\|\widetilde{\Oh}^{(3)}\|_1$ which we know to be small for the Bentkus mollifier. We remark that these are exactly, and the only, terms that appear in the standard invariance principle proofs for linear forms. Intuitively this is not surprising since the first three terms simply correspond to the case when the $A^i,B$ are \emph{diagonal} which reduces a spectrahedron to a polytope. However, bounding the remaining terms is highly non-trivial and one of our technical contributions is in showing these remaining terms are bounded for the Bentkus mollifier.

    \vspace{-2.5mm}
    \paragraph{Bounding derivatives and obtaining invariance principle.} Bounding these last three terms  of the $3$-tensors significantly deviates from the analysis of~\cite{harsha2013invariance} since we need to deal with off-diagonal entries of matrices  which is unique to the matrix-spectrahedron case and is not faced in~\cite{harsha2013invariance,servedio2017fooling,o2019fooling}. To bound this, we use several properties of Fr\'echet derivatives such as, mean value theorems for Fr\'echet derivatives, divided differences representations of Fr\'echet derivatives~\cite{brinkhuis2005matrix}, and Dyson's theorem~\cite{bhatia2013matrix} which provides a useful integral expression for Fr\'echet derivatives (using the structure of the mollifier). More importantly, since we work with the Bentkus mollifier~\cite{bentkus1990smooth}, we  completely open up the Bentkus black-box and show various analytic properties of this mollifier $\widetilde{\Oh}$ in order to prove that our Fr\'echet derivatives are bounded.

     In order to go from bounded third-order Fr\'echet derivatives to a final invariance principle, we still need to borrow some results from random matrix theory to upper bound the moments of $\sum_i \bx_iA^i$. Although, the concentration of $\sum_i \bx_i A^i$ for uniformly random $\bx\sim\U_n$ is well-studied by standard matrix Chernoff bounds~\cite{tropp2015introduction}, we need better concentration of this random matrix variable at higher Schatten norms. For the diagonal polytope case~\cite{harsha2013invariance} used the standard hypercontractivity and~\cite{o2019fooling} used Rosenthal's inequality. Fortunately for us, a matrix-version of Rosenthal's inequality~\cite{mackey2014} was proven a few years back and we  use it to conclude our proof (in fact we also crucially rely on this inequality to construct our $\PRG$). Putting everything together, for arbitrarily small $\tau>0$, we obtain our main invariance principle for the Bentkus mollifier applied as a spectral function
    \begin{align}
    \label{intro:invarianceprinciple}
      \abs{\E_{\bx\sim\U_n}\Br{\widetilde{\Oh}\br{\sum_{i=1}^n\bx_iA^i-B}}-\E_{\bg\sim\G^{n}}\Br{\widetilde{\Oh}\br{\sum_{i=1}^n\bg_iA^i-B}}} \leq \poly(\log k, M, \tau).
    \end{align}
    We remark that the invariance principle above does not assume the positivity of the matrices. We believe this is a necessity for future work on fooling arbitrary spectrahedra.

    \vspace{-2.5mm}

    \subsection{Second contribution: Geometric properties of positive spectrahedra}
    Even with an invariance principle in hand, we are faced with the same challenges as~\cite{harsha2013invariance,servedio2017fooling,o2019fooling} to show an anti-concentration statement. Recall that our goal is to show that for a $(\tau,M)$-regular positive spectrahedron $S$, the expected value of the \emph{indicator function} $[x\in S]$ for $x\sim \U_n$ is close
    to the expected value of $[\bg\in S]$ for $\bg\sim \G^n$.
    This is ``almost" what we showed in the previous section except that the Bentkus mollifier $\widetilde{\Oh}$ in Eq.~\eqref{intro:invarianceprinciple} is replaced by the orthant indicator function $f(x)=[\max_i x_i\leq 0]$. In order to move from the smooth function distance to CDF distance, one particular approach taken by~\cite{harsha2013invariance,servedio2017fooling,chattopadhyay2019simple} is to use geometric properties of polytopes, and as far as we are aware this is widely open for spectrahedra.

    \subsubsection{Properties of positive spectrahedron}

    Understanding average sensitivity and noise sensitivity of geometric objects has been an important area in theoretical computer science. For the class of halfspaces, we have several results that upper bound these properties~\cite{peres2004noise,harsha2013invariance,diakonikolas2010bounded,kane2014average}, however upper bounds on these properties are poorly understood for the case of spectrahedra. Below, we prove upper bounds on these quantities. 
    \begin{result}[Geometric properties of positive spectrahedra]
    \label{thm:introgsa}
    Let $S$ be a positive spectrahedron and consider  $F:\pmset{n}\rightarrow \01$ defined as $F(x)=[x\in S]$. The average sensitivity of $F$ is $O(\sqrt{n})$, the $\varepsilon$-Boolean noise sensitivity of $F$ is $O(\sqrt{\varepsilon})$.
    \end{result}
    We remark that the noise-sensitivity statement we have above can be  viewed as a ``positive-matrix-analogue" version of the well-known Peres's theorem~\cite{peres2004noise}.
    In order to prove this statement, we first observe that the average sensitivity of $F$ being $O(\sqrt{n})$ immediately follows by the observation that positive spectrahedra correspond to \emph{unate functions} and Kane~\cite{kane2014average} showed $\AS(f)\leq O(\sqrt{n})$ if $f$ is unate (and a similar statement is known to be false for noise sensitivity).  One issue we need to handle when translating between noise sensitivity and average sensitivity is the following: in the standard technique of~\cite{peres2004noise,diakonikolas2010bounded,kane2014average}, one upper bounds the $\varepsilon$-noise sensitivity of a function $f$ by ``bucketing" the input variables into $m=O(1/\varepsilon)$ buckets $B_1,\ldots,B_m$ and reduces the function $f:\pmset{n}\rightarrow \pmset{}$ to a function $g:\pmset{m}\rightarrow \pmset{}$ defined as $g(b)=\sum_{\ell=1}^{m}b_i \sum_{i\in B_\ell}z_i A^i$ (for uniformly random $z$). One then upper bounds  $\NS_{\varepsilon}(f)$ using $\AS(g)$ (up to a factor $\varepsilon$). Clearly when using this technique to bound $\varepsilon$-noise sensitivity of halfspaces, both $f,g$ are intersections of halfspaces and one can upper bound the average sensitivity of $g$ using Kane's result~\cite{kane2014average} to be $O(\sqrt{m})$. However in our setting if  $f$ is an indicator of a \emph{positive} spectrahedron, then $g$ no longer needs to be an indicator of a \emph{positive} spectrahedron since $\sum_{i\in B_\ell}z_i A^i$ need not even be either a positive semidefinite matrix or a negative semidefinite matrix. We overcome this by modifying the bucketing procedure of~\cite{diakonikolas2010bounded} to ensure $g$ is an indicator of a unate function. However, in the process case we end up upper bounding $\NS_\varepsilon(f)$ by the ``average $2$-sensitivity" of $g$. We extend the results of Kane~\cite{kane2014average} by showing that even the ``average $2$-sensitivity" of $g$ is small for our setting.

    \subsubsection{Boolean anti-concentration}
    \label{intro:booleananti}
    For polytopes, Gaussian anti-concentration immediately follows from the fact that the Gaussian surface area of polytopes is bounded since its surface has only finite normed vectors. This is crucially used in~\cite{harsha2013invariance,servedio2017fooling,chattopadhyay2019simple}. However, it is not clear how to upper bound the $\GSA$ of positive spectrahedra due to its complicated geometric structures. Moreover, even with an upper bound on $\GSA$, we still do not know how to obtain Gaussian anti-concentration. Here, to move from mollifier-closeness to CDF closeness, we prove a \emph{Boolean} anti-concentration for positive spectrahedra, which is in fact \emph{stronger} than Gaussian anti-concentration, inspired by the Boolean anti-concentration for polytopes in~\cite{o2019fooling}.

    \paragraph{Regularity condition.} Before explaining the Boolean anti-concentration, we need to revisit the {\em regularity} condition, which is also used for polytopes. In~\cite{harsha2013invariance,servedio2017fooling}, it is assumed that every halfspace (or row in the matrix $W$) satisfies $\|W^i\|_2=1$ and $\|W^i\|_4\leq \tau$. One important question is: what is a regularity assumption for spectrahedra and for which assumptions can we \emph{show} anti-concentration? A natural possibility is to see if Nazarov's result~\cite{nazarov2003maximal} holds for spectrahedra (i.e., show anti-concentration in the weaker \emph{Gaussian} setting). To the best of our knowledge, this has firstly not been studied in literature. Moreover, it is not hard to see that, in order for the proof of Nazarov to work for spectrahedra, one can make a very \emph{strong} assumption that \emph{every}~$A^i$ satisfies $\lambda_{\min}(A^i)\geq 1$. However, this seems to significantly restrict the class of spectrahedra.

    In order to resolve this, we propose $\br{\tau,M}$-regularity as defined in Eq.~\eqref{eqn:regularity} and prove a stronger statement, i.e., Boolean anti-concentration for $(\tau,M)$-regular positive spectrahedra.
    We use this statement to go from closeness between the mollifier $\widetilde{\Oh}\br{\sum_i \bx_iA^i-B}$ and $\widetilde{\Oh}\br{\sum_i \bg_iA^i-B}$ (which we already established in Eq.~\eqref{intro:invarianceprinciple}) to closeness between $\Br{\sum_i \bx_iA^i\preceq B}$ and  $\Br{\sum_i \bg_iA^i\preceq B}$. In this direction, we prove a Littlewood-Offord type theorem for positive spectrahedra.

    \begin{result} [Littlewood-Offord for positive spectrahedra]
     If $(A^1,\ldots,A^n)$ are $(\tau,M)$-regular. Then every $\Lambda$, we have
    $$
    \Pr_{\bx\sim \U_n}\Br{ \lambda_{\max}\br{ \sum_i \bx_i A^i-B} \in [-\Lambda,\Lambda]}\leq O(\Lambda).
    $$
    \end{result}
    The classic Littlewood-Offord theorem~\cite{littlewood1939number,erdos1945lemma} anti-concentration inequality for a halfspace $w\in \R^n$ (satisfying $|w_i|\geq 1$) and $\alpha\in \R$ proves a bound on the probability that $\sum_i w_i \bx_i \in [\alpha,\alpha+2]$  (where $\bx\sim \U_n$). In~\cite{o2019fooling} they generalized this for \emph{intersections} of halfspaces and in the result above we show a matrix-version of Littlewood-Offord theorem. Intuitively, our statement shows the largest eigenvalue of a positive spectrahedron cannot all be very-concentrated in a small region (i.e., small eigenvalue regions have small measure over the Boolean cube).

    The proof of our result is similar to the proofs in~\cite{kane2014average,o2019fooling} which show anti-concentration for intersections of unate functions. There are a couple of subtleties for us: in~\cite{o2019fooling}, they  perform  random ``bucketing" of the coordinates in a polytope and show that with high probability, each bucket has ``significant" weight, which  follows immediately from the Paley-Zygmund inequality. However, for us, random bucketing does not produce a positive spectrahedron (the same issue which we faced in Theorem~\ref{thm:introgsa}), so instead we need to bucket in a non-standard way to go from a positive spectrahedron to a bucket which corresponds to a unate function. Next, to show that each bucket has significant weight (which in our case corresponds to large smallest eigenvalue), we invoke the matrix Chernoff bound for negatively correlated variables, proving our result. We remark that higher-dimensional extensions of the Littlewood-Offord theorem~\cite{frankl1988solution,tao2012littlewood} do not talk of eigenspectra of matrices and differs from our result.

    Using the standard bits-to-Gaussians trick, this also gives us Gaussian anti-concentration (i.e., the positive spectrahedra analogue of Nazarov's result~\cite{nazarov2003maximal} which is unknown as far as we are aware). Putting this together with our invariance principle statement we obtain our main~result.
    \begin{result}[Fooling positive spectrahedra]
    For every $(\tau,M)$-regular positive spectrahedron $S$,
    \begin{align}
    \label{eq:intromainresult}
    \big|\mathop{\mathbb{E}}_{\bx\sim \U_n} [\bx\in S]-\mathop{\mathbb{E}}_{\bg\sim \G^{n}} [\bg\in S] \big|\leq \poly(M,\log k,\tau).
    \end{align}
    \end{result}

\bigskip

     Apart from the applications of constructing pseudorandom generators (which we discuss in the next section) we believe that our invariance principle for the Bentkus mollifier of  \emph{arbitrary} spectrahedra, opening up the Bentkus mollifier (i.e., understanding the Bentkus functions which were almost used as a black-box in~\cite{harsha2013invariance,servedio2017fooling,o2019fooling}), the Littlewood-Offord theorem and noise sensitivity for positive spectrahedra  could be of independent~interest.

    \subsection{Applications}

    \subsubsection{Pseudorandom generators}
    We now briefly discuss how to use the invariance principle to obtain our pseudorandom generator. Our construction is based on the Meka-Zuckerman~\cite{meka2013pseudorandom} $\PRG$ construction for fooling halfspaces. We note in the passing that this same $\PRG$ (with different parameters) was also used by~\cite{harsha2013invariance,servedio2017fooling} and slight modification of it by~\cite{o2019fooling}. We omit the details of the $\PRG$ construction here referring the interested reader to Section~\ref{sec:prg} for an explicit construction.

    One subtlety in order to go from invariance principle to fooling the MZ-generator is the following: recall that our invariance principles showed that expected value under the uniform distribution was close to the expected value under the Gaussian distribution. However, in order to fool the MZ-generator one needs to show that the invariance principle proofs holds also for $k$-wise independent distributions. In this direction, we use a neat trick from~\cite{o2019fooling} that shows that in order to show invariance principles for $k$-wise independent distributions, it suffices to show just Boolean anti-concentration. Second we crucially use the fact that  the matrix Rosenthal inequality can be derandomized by analyzing its the original proof. Put together, this shows that our invariance principle proof holds for $k$-wise independent distributions and gives us our main $\PRG$~result.

    \begin{result} [PRG for positive spectrahedra]
    Let $S$ be a $(\tau,M)$-regular positive spectrahedron.
    There exists a $\PRG$ $G:\01^r\rightarrow \pmset{n}$ with $r=(\log n)\cdot\poly(\log k,M,1/\delta)$ that $\delta$-fools $S$ with respect to the uniform distribution for every $\tau\leq \poly(\delta/(\log k\cdot M))$.
   \end{result}

    \subsubsection{Learning theory}
     Learning geometric objects is a fundamental problem in computational learning theory.  An application of upper bounding noise sensitivity (in Theorem~\ref{thm:introgsa}) is in agnostic learning. The agnostic learning framework introduced by~\cite{kearns1994toward,haussler1992decision} is the following: let $\mathcal{C}\subseteq \{c:\pmset{n}\rightarrow \01\}$ be a concept class and $\De:\pmset{n}\times \01\rightarrow [0,1]$ be a distribution. Define
    $
    \textsf{opt}(\mathcal{C})=\min_{c\in \mathcal{C}} \Pr_{(x,b)\sim \De} [c(x)\neq b],
    $
    i.e., what is the \emph{best} approximation to $\De$ from within the concept class. The goal of an agnostic learner is the following: given many samples $(x,b)\sim \De$, the goal of a learner is to produce a hypothesis $h:\pmset{n}\rightarrow \01$ which satisfies
    $$
    \Pr_{(x,b)\sim \De} [h(x)\neq b]\leq \textsf{opt}(\mathcal{C})+\varepsilon.
    $$
    Note that if $\textsf{opt}(\mathcal{C})=0$, this is the standard PAC learning framework and agnostic learning models learnability under adversarial noise.  A natural restriction of this model is when the marginal of $\De$ on the first $n$ bits is the uniform distribution on $\01^n$. It is a folklore result~\cite{klivans2004learning} that  a function $f$ having low noise sensitivity can be approximated by low-degree polynomials (see~\cite[Lemma~2.7]{harsha2013invariance} for an explicit statement). Furthermore, the well-known L1-polynomial regression algorithm~\cite{kalai2008agnostically} shows how to learn low-degree polynomials in the agnostic framework. Putting these two connections together gives us the following theorem.
    \begin{result}[Learning positive spectrahedra] The concept class of positive spectrahedra (in $n$ variables with $k\times k$ symmetric matrices) can be agnostically learned under the uniform distribution in time $n^{O(\log k)}$ for every constant error parameter.
    \end{result}
     The previous best known result~\cite{klivans2008learning} for learning positive spectrahedra even in the PAC model~was $2^{O(n^{1/4})}$ (as far as we are aware); our result provides a substantially better complexity.

    \subsubsection{Discrepancy sets for spectrahedra}
    Understanding discrepancy sets for convex objects is a  fundamentally important problem in the fields of convex geometry, optimization, and a range of other areas. Prior works of~\cite{harsha2013invariance,servedio2017fooling,o2019fooling} constructed such discrepancy sets for polytopes, but a natural question is to extend their construction to  spectrahedra.  In our context, one application of our main result can be viewed as  the following: consider the set of all possible positive spectrahedra (over the Boolean cube) $S=\{x\in\set{-1,1}^n:\sum_i x_i A^i\preceq B\}$, then can we construct a  \emph{small} subset of the Boolean cube $\pmset{n}$ such that this set $\delta$-approximates the $\set{-1,1}^n$-volume of  every positive spectrahedron? One way to construct such a set is to construct a $\PRG$ for the class of functions. So an immediate corollary of our $\PRG$ for positive spectrahedra is the following theorem.\footnote{We remark that counting integer solutions to positive spectrahedra is not as  \emph{naturally motivated} as that for polytopes, but nevertheless understanding discrepancy sets for geometric objects is a fundamental question.}

    \begin{result}[Discrepancy set for positive spectrahedra] There is a deterministic algorithm  which, given a $(\tau,M)$-regular positive spectrahedron $S$,  runs in time $\exp(\log n,\log k, M,1/\delta)$ and outputs a $\delta$-approximation of the number of points in $\set{-1,1}^n$ contained in $S$ as long as $\tau\leq \poly(\delta/(M\log k))$.
    \end{result}

    \subsubsection{Intersection of (structured) polynomial threshold functions}
    Constructing $\PRG$s for  $\PTF$s has received a lot of attention. However, the best known seed length for fooling a degree-$k$ $\PTF$ on $n$ bits scales as $O(\log n\cdot 2^k)$ (over the Boolean space).  A simple observation we make is that fooling spectrahedra (on $n$ bits with $k\times k$ matrices) can be in fact be viewed as the more challenging task of fooling an \emph{intersection} of $k$ many degree-$k$ $\PTF$s.

    Recall that a spectrahedron is given by
    $
    S=\{x\in \R^n:B-\sum_ix_i A^i \succeq 0\}.
    $
    Without loss of generality, we may assume that the measure of $x$ satisfying $\det\br{\sum_i B-x_iA^i}=0$ is zero.
    Sylvester's criterion implies that a matrix $M$ (which in our case is $B-\sum_i x_iA^i$) is positive definite \emph{if and only if} the determinant of the $k$ principle minors of $M$ are positive. Hence, an alternate characterization of $S$ is the set of $x\in \R^n$ for which
    $$
    S=\bigwedge_{r=1}^k \Br{\textsf{det} \br{B-\sum_ix_i A^i }_{r\times r}>0}=\bigwedge_{r=1}^k \sign[p_r(x)]
    $$
    modulo a zero-measure set, where  $M_{r\times r}$ means the top left $r\times r$ principle minor of $M$. Clearly each determinantal expression produces a polynomial $p_r$ of degree at most $r$. So, our main result about fooling $S$, shows that there is a  \emph{structured}  class of \emph{intersections} of  degree-$k$ $\PTF$s (i.e., the class of polynomials which can be written as in terms of the above) which can be fooled by a $\PRG$ with seed length $O(\log n\cdot \log k\cdot M/\delta)$, which is exponentially better than using existing $\PRG$s for~$\PTF$s.

    We remark  that apriori, it is not even clear why should an \emph{arbitrary} polynomial even correspond to a spectrahedron as above? However, a well-known result of~\cite{helton2006noncommutative,gartner2012approximation} states that an arbitrary degree-$d$ polynomial $p\in \R[x_1,\ldots,x_n]$ with real coefficients  has a symmetric \emph{determinantal representation},\footnote{See~\cite{quarez2012symmetric} for a simple linear algebraic proof of this statement.} i.e.,  there exists symmetric $A^0,A^1,\ldots,A^n$ such that
    $$
    p(x_1,\ldots,x_n)=\textsf{det}\br{A^0+\sum_i x_i A^i}.
    $$
    where $A^i\in \textsf{Sym}\binom{n+d}{d}$. So, if we could fool arbitrary spectrahedra that might be a promising avenue to fool $\PTF$s and intersections of $\PTF$s.

    \subsection{Future work}
    Our work opens this new line of research into understanding $\PRG$s for spectrahedra with several novel techniques. This raises several questions for future work.

    \emph{\textbf{1}. Can we remove regularity for positive spectrahedra?} One of the crucial techniques that Servedio and Tan~\cite{servedio2017fooling} introduced (inspired by a prior work of Servedio~\cite{servedio2006every}) was decomposing a polytope into head and tail variables (i.e., tail coordinates in a halfspace which satisfy regularity and head coordinates are the dominant variables). They express the head variables as  CNF, use the result of Bazzi~\cite{bazzi2009polylogarithmic} to fool the head variables and invariance principles for tail variables. However, in our setting breaking up a single spectrahedron into head and tail variables is  unclear and even if possible, what is the analogue of the CNF for our~setting?

    \emph{\textbf{2}. Can we fool arbitrary spectrahedra?} Besides the difficulty in removing the regularity condition, another fundamental barrier we face here is, anti-concentration. What is the Gaussian surface area of a spectrahedron, even this is unknown (as far as we are aware). Our techniques such as bucketing, using Kane's result~\cite{kane2014average}, and Boolean anti-concentration~\cite{o2019fooling} crucially use the assumption of positivity. Going beyond this, might require new understanding on the geometric structures (like average sensitivity, noise sensitivity) about arbitrary spectrahedra.

    \emph{\textbf{3}. A general invariance principle for spectral functions?} Here, we showed our invariance principle specifically for the Bentkus mollifier. However, like the result of~\cite{harsha2013invariance} can we prove a general invariance principle for arbitrary smooth spectral functions? Given the applications of invariance principles, they are now  considered to be powerful techniques in computational complexity theory. Having an invariance principle for spectral functions could find more applications such as deciding noisy entangled quantum games~\cite{yao2019doubly}.

    \emph{\textbf{4}. Can we fool spectrahedral caps?} Let $S_{n-1}=\{x\in \R^n: \|x\|_2=1\}$ denote the $n$-dimensional sphere, then a \emph{spectrahedral cap} is the set of $S_{n-1}$ that is ``cut" by a spectrahedron, i.e., for a  spectrahedron $S$, we define the spectrahedral cap $C_S$ as $C_S= S_{n-1}\cap S$. In the polytope-setting, fooling spherical caps has received a lot of attention classically~\cite{harsha2013invariance,kothari2015almost} (with almost optimal seed length $\PRG$s). Can we similarly fool {spectrahedral caps}?

    \emph{\textbf{5}. Fooling polynomial threshold functions?} Can we make progress in finding better $\PRG$s for $\PTF$s using techniques we developed here for fooling arbitrary spectrahedra?

    \paragraph{Acknowledgements.} We thank Oded Regev  for pointing out a minor inconsequential error of the previous version. We also thank  Jop Bri\"et and Minglong Qin for several helpful comments. This collaboration earlier faced some bureaucratic issues. We are deeply grateful for the support from Jelani Nelson, Kewen Wu, Yitong Yin and others in the TCS community. P.Y.\ was supported by the National Key R\&D Program of China 2018YFB1003202, National Natural Science Foundation of China (Grant No. 61972191), the Program for Innovative Talents and Entrepreneur in Jiangsu, the Fundamental Research Funds for the Central Universities 0202/14380068 and Anhui Initiative in Quantum Information Technologies Grant No. AHY150100. Part of the work was done when P.Y.\ and S.A.\ were participating in the program "Quantum Wave in Computing" held at Simons Institute for the Theory for Computing.

    \paragraph{Organization.} In Section~\ref{sec:prelim} we introduce the mathematical aspects which we use in this paper, and state various lemmas in random matrix  theory and multidimensional calculus. In Section~\ref{sec:mollifier}, we introduce the Bentkus mollifier and discuss various properties. In  Section~\ref{sec:spectralderivatives} we state our main theorem regarding spectral derivatives of smooth functions and go on to bound the spectral derivatives for the Bentkus function (proving a technical lemma in Appendix~\ref{app:claim5proof}). In Section~\ref{sec:geometric} we prove an upper bound on the noise sensitivity of positive spectrahedra as well as our Littlewood-Offord theorem for this class. In Section~\ref{sec:invarianceprinciple} we prove our invariance principle theorem and go on to construct a pseudorandom generator for the class of positive spectrahedra.

    \section{Preliminaries}
    \label{sec:prelim}
    For an integer $n\geq 1$, let $[n]$ represent the set $\set{1,\ldots, n}$. Given a finite set $\X$ and a natural number~$k$, let $\X^k$ be the set $\X\times\cdots\times\X$, the Cartesian product of $\X$, $k$ times. Given $a=(a_1,\ldots, a_k)$ and a set $S\subseteq[k]$, we write $a_S$ and $a_{-S}$ to represent the projections of $a$ to the coordinates specified by $S$ and the coordinates outside $S$, respectively. For any $i\in[k]$, $a_{-i}$ represents $a_1,\ldots, a_{i-1},a_{i+1},\ldots, a_n$ and $a_{<i}$ represents $a_1,\ldots, a_{i-1}$. $a_{\leq i},a_{>i}, a_{\geq i}$ are defined similarly. For a distribution $\mu$  on $\X$, let  $\mu\br{x}$ represent the probability of $x\in\X$ according to $\mu$.  Let $X$ be a random variable distributed according to $\mu$. We use the same symbol to represent a random variable and its distribution whenever it is clear from the context. The expectation of a function $f$ on $\X$ is defined as $\expec{f(X)}=\mathbb{E}_{\bx\sim X}\Br{f(\bx)}=\sum_{x\in\X}\prob{X=x}\cdot f\br{x}=\sum_x\mu\br{x}\cdot f\br{x}$, where $\bx\sim X$ represents that $\bx$ is drawn according to $X$.  For any event $\mathcal{E}_x$ on $x$, $\Br{\mathcal{E}\br{x}}$ represents the indicator function of $\mathcal{E}$. In this paper, the lower-cased letters in bold $\bx,\by,\bz\cdots$ are reserved for random variables.

    \paragraph{Distributions.}  Throughout the paper, we denote~$\G$ (where $\G=\mathcal{N}(0,1)$) to be a standard normal distribution over $\R$  with mean $0$ and variance $1$. We denote $\U_n$ to be the uniform distribution on $\pmset{n}$. We say a joint distribution $X=(\bx_1,\ldots,\bx_n)$ is $t$-wise uniform  if the marginal distribution $X_S$ for any subset $S\subseteq [n]$ of size $\abs{S}=t$ is uniformly distributed (observe that the uniform distribution is clearly $t$-wise independent for every $t\geq 1$). A distribution $\H$ on functions $[n]\rightarrow[m]$ is said to be an $r$-wise uniform hash family if for $\bh\sim\H$, $\br{\bh\br{1},\ldots,\bh\br{n}}$ is $r$-wise uniform.

    \subsection{Derivatives and multidimensional Taylor expansion}
    We denote $\mathcal{C}^d$ as the set of all real functions that are $d$-time differentiable. For $f:\reals\rightarrow\reals$ in $\mathcal{C}^d$, we use $f^{\br{d}}$ to denote the $d$-th derivative of $f$. Given a function $F:\reals^k\rightarrow \reals$ and a $k$-dimensional multi-index $\alpha=\br{\alpha_1,\ldots,\alpha_m}\in\mathbb{N}^k $, $\partial_{\alpha}F$ denotes the mixed partial derivative taken $\alpha_i$ times in the $i$-th coordinate.

    \begin{fact}[\cite{rudin}]
      Let $k\in\mathbb{N}$ and $f:\reals^k\rightarrow\reals$ be a $\mathcal{C}^d$ function. Then for all $x,y\in\reals^k$,
      \[
      f\br{x+y}=\sum_{\alpha\in\mathbb{N}^k:\abs{\alpha}\leq d-1}\frac{\partial_{\alpha}f\br{x}}{\alpha!}\prod_{i=1}^{m}y_i^{\alpha_i}+\textsf{err}\br{x,y},
      \]
      where $\alpha!=\alpha_1!\cdots\alpha_m!$, $\abs{\alpha}=\sum_i\alpha_i$ and
      \[
      \abs{\textsf{err}\br{x,y}}\leq\sup_{v\in\reals^k}\sum_{\alpha\in\mathbb{N}^k:\abs{\alpha}= d}\abs{\partial_{\alpha}f\br{v}}\max_i\abs{y_i}^d.
      \]
    \end{fact}
    For a $t$-time differentiable function $f:\R^k\rightarrow \R$ and $s\leq t$, define
    $$
    \|f^{(s)}\|_1=\max\Big\{\sum_{p_1,p_2,\ldots,p_s\in [k]} |\partial_{p_1}\cdots\partial_{p_s} f(x)|: x\in \R^k\Big\}
    $$

    \begin{definition} \label{def:dividedifference}
     Let $f:\reals\rightarrow\reals$. For any distinct inputs $x_1,\ldots,x_n\in\reals$, the divided difference is defined recursively as follows.
      \begin{eqnarray*}
        &&f^{[0]}= f,\\
        &&f^{[i]}\br{x_1,\ldots,x_{i+1}}=\frac{f^{[i]}\br{x_1,\ldots,x_{i-1},x_i}-f^{[i]}\br{x_1,\ldots,x_{i-1},x_{i+1}}}{x_i-x_{i+1}}.
      \end{eqnarray*}
      For other values of $x_1,\ldots,x_{i+1}$, $f^{[i]}$ is defined by continuous extension.
    \end{definition}

    \begin{fact}[Mean value theorem for divided difference~\cite{de2005divided}]\label{fac:mvtdd}
    For every $f\in\mathcal{C}^n$ and $x_1,\ldots,x_{n+1}$, there exists $\xi\in(\min\set{x_1,\ldots,x_{n+1}},\max\set{x_1,\ldots,x_{n+1}})$ such that
      \[f^{[n]}\br{x_1,\ldots,x_{n+1}}=\frac{f^{(n)}\br{\xi}}{n!}.\]
    \end{fact}

    \subsection{Combinatorial properties of Boolean functions}

    Let $f:\01^n\rightarrow \01$, $g:\R^n\rightarrow \01$ and $S$ be a Borel set in $\reals^n$. We define the following combinatorial properties of Boolean-valued functions $f,g$.

    \begin{enumerate}
        \item Average sensitivity: $\AS(f)=\sum_{i=1}^n\Pr_{\bx}[f(\bx)\neq f(\bx\oplus e_i)]$, where the probability is taken uniformly in $\01^n$.
        \item $\varepsilon$-Noise sensitivity: $\NS_\varepsilon(f)= \Pr_{\bx,\by}[f(\bx)\neq f(\by)]$ where the probability is taken according to the distribution: $\bx$ is uniformly random in $\01^n$ and $\by$ is obtained from $\bx$ by  independently flipping each $\bx_i$  with probability $\varepsilon$.
    \end{enumerate}

    We refer interested readers to~\cite{o2014analysis} for more on these parameters and their applications to analysis of Boolean functions.

    \subsection{Matrix analysis and Random matrices}

    For any integer $k>0$, we use $\mat{k}$ and $\sym{k}$ to represent the set of $k\times k$ real matrices and symmetric matrices, respectively. For any matrix $X$, $\norm{X}_p$ represents the Schattern $p$-norm of $X$ and $\norm{X}$ represents the spectral norm of $X$. $\id_k$ represents a $k\times k$ identity matrix. The subscript~$k$ may be omitted whenever the dimension is clear from the context. We need the following results in matrix analysis.

    \begin{fact}\label{fac:bhatia}~\cite{10.2307/2589115}
      For any $k\times k$ real symmetric matrix $A$, let $B$ be its upper triangle part of $A$. Namely $B_{i,j}=A_{i,j}$ if $i\leq j$ and is $0$  otherwise. Then $\norm{B}\leq\frac{\ln k}{\pi}\norm{A}$.
    \end{fact}

    \begin{fact}\cite[Theorem 1.1]{Tropp2011}\label{fac:matrixchernoff}
      Let $n,k\geq1$ be integers and $X_1,\ldots, X_n$ be independent random $k\times k$ real symmetric matrices satisfy $0\preceq X_i\preceq R$ for $i\in [n]$. Set
      \[\mu=\lambda_{\min}\br{\frac{1}{n}\sum_{i=1}^n\expec{X_i}}.\]
      Then
      \[\prob{\lambda_{\min}\br{\sum_{i=1}^{n}X_i}\leq\br{1-\delta}\mu}\leq k\cdot\br{\frac{e^{-\delta}}{\br{1-\delta}^{1-\delta}}}^{\mu/R}\]
      for every $\delta\in[0,1)$.
    \end{fact}

    \begin{fact}\label{lem:normrandommatrixboolean}
    For every integer $m\geq 1$ and  $A_1,\ldots, A_n\in \sym{k}$ it holds that
      \[\expec{\norm{\sum_i\bg_iA^i}^{m}}\leq\br{1+2m\lceil\log k\rceil}^{m/2}\cdot\norm{\sum_i (A^i)^2}^{m/2}\]
      and
      \[\expec{\norm{\sum_i\bx_iA^i}^{m}}\leq\br{1+2m\lceil\log k\rceil}^{m/2}\cdot\norm{\sum_i (A^i)^2}^{m/2},
      \]
      where the expectations are taken over $\bx\sim \U_n$ and $\bg\sim \G^n$. Additionally, the second inequality still holds if $\bx$ is $2m\lceil\log k\rceil$-wise uniform.
    \end{fact}

    \begin{proof}
       It suffices to prove the second inequality as the first one follows by the standard bits-to-Gaussians tricks~\cite[Chapter 11]{o2014analysis}. Let $B=\sum_i\bx_iA^i$ where $\bx\sim\U_n$. The proof closely follows the argument in~\cite{10.1007/978-3-319-40519-3_8}, where Tropp proved the $m=1$ case. For any integer $p\geq 1$, it is proved in~\cite[Eqs.~(4.9,\ 4.11)]{10.1007/978-3-319-40519-3_8} that
      \[\expec{\Tr~B^{2p}}\leq k\cdot\br{\frac{2p+1}{e}}^p\cdot\norm{\sum_i(A^i)^2}^p.\]
      Thus
      \[\expec{\norm{B}^m}\leq\br{\expec{\Tr~B^{2pm}}}^{1/2p}\leq k^{1/2p}\cdot\br{\frac{2pm+1}{e}}^{m/2}\cdot\norm{\sum_i (A^i)^2}^{m/2}.\]
      Setting $p=\lceil\log k\rceil$, we conclude the result.  Since the proof involves only $2m\cdot \lceil \log k\rceil$ powers of $B$, it also holds true for $\bx$ being drawn from a $2m\cdot \lceil \log k\rceil$-wise uniform distribution.
    \end{proof}

    \begin{fact}[{Matrix Rosenthal inequality~\cite[Corollary 7.4]{mackey2014}}]\label{fac:ros}
     Let $X_1,\ldots, X_n$ be centered, independent random real symmetric matrices. Then
    \begin{eqnarray*}
      &&\br{\expec{\norm{\sum_i X_i}_{4p}^{4p}}}^{\frac{1}{4p}} \leq \sqrt{4p-1}\norm{\br{\sum_i\expec{X_i^2}}^{\frac{1}{2}}}_{4p}+\br{4p-1}\br{\sum_i\expec{\norm{X_i}_{4p}^{4p}}}^{\frac{1}{4p}}.
    \end{eqnarray*}
    This inequality still holds if $X_1,\ldots, X_n$ are $4p$-wise independent.
    \end{fact}

    \subsection{Matrix functions, spectral functions and Fr\'echet derivatives}

    Let $f:\R^k\rightarrow \R$ and $\lambda:\sym{k}\rightarrow \R^k$ where $\lambda (X)=\br{\lambda_1(X),\ldots,\lambda_k(X)}$ are the eigenvalues of $M$ sorted in a non-increasing order. We refer to $\lambda_{\max}=\lambda_1$ interchangeably. Let $F=f\circ \lambda :\sym{k}\rightarrow \R$.

    If $f:\reals\rightarrow\reals$ is an analytic function in $\reals$, namely its Taylor series converges in $\reals$, we  define $f\br{X}$ for general matrices using its Taylor expansion. It is not hard to see that the Taylor series still converges with matrix inputs. If $X$ is symmetric with a spectral decomposition $X=UDU^T$, where $D=\diag\br{\lambda_1\br{X},\ldots, \lambda_k\br{X}}$, then $f\br{X}=U\diag\br{f\br{\lambda_1\br{X}},\ldots,\lambda_k\br{X}}U^T$.

    The Fr\'echet derivatives are a notion of  derivatives defined in Banach space.  In this paper, we only concern about the Fr\'echet derivatives on matrix spaces. Readers may refer to~\cite{coleman2012calculus} for a more thorough treatment. The Fr\'echet derivatives are the maps that are defined as follows.

    \begin{definition}\label{def:Frechet}
    Given integers $m,n\geq 1$, a map $F:\mat{m}\rightarrow \mat{n}$ and $P,Q\in\mat{m}$, the Fr\'echet derivative of $F$ at $P$ with respect to $Q$ is defined to be
    \[DF\br{P}\Br{Q}=\frac{d}{dt}F\br{P+tQ}|_{t=0}.\]
    The $k$-th order Fr\'echet derivative of $F$ at $P$ with respect to $\br{Q_1,\ldots, Q_k}$ is defined recursively as
    \[D^kF\br{P}\Br{Q_1,\ldots, Q_k}=\frac{d}{dt}D^{k-1}F\br{P+tQ_k}\Br{Q_1,\ldots, Q_{k-1}}|_{t=0}.\]
    \end{definition}
    Fr\'echet derivatives share many common properties with the derivatives in Euclidean spaces, such as linearity, composition rules, Taylor expansions, etc. We refer the interested reader to~\cite{coleman2012calculus,bhatia2013matrix} for more.
    Some basic properties of Fr\'echet derivatives are summarized in the following fact.
    \begin{fact}~\cite[Chapter X.4]{bhatia2013matrix}\label{fac:frechetderivative}
    	Given $F,G:\mat{n}\rightarrow \mat{m}$ and $P,Q_1,\ldots, Q_k\in \mat{n}$, it holds that
    	\begin{enumerate}
    		\item $D\br{F+G}\br{P}\Br{Q}=DF\br{P}\Br{Q}+DG\br{P}\Br{Q}$.
    		
    		\item $D\br{F\cdot G}\br{P}\Br{Q}=DF\br{P}\Br{Q}\cdot G\br{P}+F\br{P}\cdot DG\br{P}\Br{Q}$.
    		
    		\item If $m=n$, $D\br{F\circ G}\br{P}\Br{Q}=\br{D\br{G\circ F}\br{P}\circ DF\br{P}}\Br{Q}$.
    		
    		\item $D^kF\br{P}\Br{Q_1,\ldots, Q_k}=D^kF\br{P}\Br{Q_{\sigma\br{1}},\ldots, Q_{\sigma\br{k}}}$ for every $k>0$ and permutation $\sigma\in S_k$.	
    	\end{enumerate}
    \end{fact}
    The following fact states that Fr\'echet derivatives can be expressed as divided differences.

    \begin{fact}\label{fac:fredivided}~\cite{brinkhuis2005matrix}
      Let $f:\reals\rightarrow\reals$ be twice differentiable and $X=\diag\br{x_1,\ldots,x_k}$ be a diagonal matrix whose spectrum is in $\reals$. For any matrix $A, B$, the following holds\footnote{In~\cite[Lemma~3.8]{brinkhuis2005matrix} this fact is proven when $A=B$ is a symmetric matrix and one can easily generalize their proof to obtain Eqs.~\eqref{eqn:1st},~\eqref{eqn:2nd} for general matrices $A, B$.}
      \begin{enumerate}
        \item
        \begin{equation}\label{eqn:1st}
        Df\br{X}\Br{A}=\br{f^{[1]}\br{x_{i_1},x_{i_2}}A_{i_1,i_2}}_{1\leq i_1,i_2\leq k}.
        \end{equation}
        \item
        \begin{equation}\label{eqn:2nd}
        D^2f\br{X}\Br{A,B}=\br{\sum_{j=1}^kf^{[2]}\br{x_{i_1},x_j,x_{i_2}}A_{i_1,j}B_{j,i_2}}_{1\leq i_1,i_2\leq k}.
        \end{equation}
      \end{enumerate}
    \end{fact}

    \begin{fact}[{Dyson's expansion~\cite[Chapter X.4]{bhatia2013matrix}}]\label{fac:ex2nd}
      Let $f\br{x}=e^x$. For any $X\in\sym{k}$ and $A\in\mat{k}$, it holds
      \[Df\br{X}\Br{A}=\int_{0}^{1}du~e^{\br{1-u}X}Ae^{uX}.\]
    \end{fact}
    \begin{lemma}\label{lem:ex2derivative}
      Let $f\br{x}=e^{-x^2/2}$. For any $X\in\sym{k}$ and $A, B\in\mat{k}$, it holds that
      \begin{eqnarray*}
        &&D^2f\br{X}\Br{A,B}\\
        &=&\frac{1}{4}\int_{0}^{1}du\int_{0}^{1}dv~\br{1-u}e^{-\br{1-u}\br{1-v}X^2/2}\br{XB+BX}e^{-\br{1-u}vX^2/2}\br{XA+AX}e^{-uX^2/2}\\
        &&+\frac{1}{4}\int_{0}^{1}du\int_{0}^{1}dv~ ue^{-\br{1-u}X^2/2}\br{XA+AX}e^{-u\br{1-v}X^2/2}\br{XB+BX}e^{-uvX^2/2}\\
        &&-\frac{1}{2}\int_{0}^{1}du~e^{-\br{1-u}X^2/2}\br{AB+BA}e^{-uX^2/2}.
      \end{eqnarray*}
      In particular, if $A=B=H$ is a symmetric matrix ,then
      \begin{eqnarray*}
        &&D^2f\br{X}\Br{H,H}\\
         &=&\frac{1}{4}\int_{0}^{1}du\int_{0}^{1}dv~\br{1-u}e^{-\br{1-u}\br{1-v}X^2/2}\br{XH+HX}e^{-\br{1-u}vX^2/2}\br{XH+HX}e^{-uX^2/2}\\
        &&+\frac{1}{4}\int_{0}^{1}du\int_{0}^{1}dv~\br ue^{-\br{1-u}X^2/2}\br{XH+HX}e^{-u\br{1-v}X^2/2}\br{XH+HX}e^{-uvX^2/2}\\
        &&-\int_{0}^{1}du~e^{-\br{1-u}X^2/2}H^2e^{-uX^2/2}.
      \end{eqnarray*}
    \end{lemma}

    Note that $f\br{x}=e^{-x^2/2}$ is analytical in $\reals$. Thus it is valid to define $f$ on arbitrary matrices.

    \begin{proof}
      For any $t\in(0,1)$, we define $g\br{x}=e^{-tx^2}$. By the definition of Fr\'echet derivatives
      \begin{eqnarray*}
        &&Dg\br{X}\Br{A}=\lim_{\eps\rightarrow 0}\frac{1}{\eps}\br{e^{-t\br{X+\eps A}^2}-e^{-tX^2}} \\
        &=&\lim_{\eps\rightarrow 0}\frac{1}{\eps}\br{e^{-t\br{X^2+\eps\br{XA+AX}+\eps^2 A^2}}-e^{-tX^2}}\\
        &=&\lim_{\eps\rightarrow 0}\frac{1}{\eps}\br{e^{-t\br{X^2+\eps\br{XA+AX}}}+O\br{\eps^2}-e^{-tX^2}}\\
        &=&Dh(-tX^2)[-t(XA+AX)]\\
        &=&-t\int_{0}^{1}du~e^{-\br{1-u}tX^2}\br{XA+AX}e^{-utX^2},
      \end{eqnarray*}
      where the second equality is from the fact that $\norm{e^{X+\eps Y}-e^X}=O\br{\eps}$, third equality holds for $h(x)=e^x$ and the last equality is from Fact~\ref{fac:ex2nd}.
      Setting $t=\frac{1}{2}$, we have
      \[Df\br{X}\Br{A}=-\frac{1}{2}\int_{0}^{1}du~e^{-\br{1-u}X^2/2}\br{XA+AX}e^{-uX^2/2}.\]
      Taking one more derivative on $X$ with respect to $B$, we conclude the result (using properties of Fr\'echet derivatives in items 2,3 of Fact~\ref{fac:frechetderivative}).
    \end{proof}

    \subsection{spectrahedra and Positive spectrahedra}

    \begin{definition}\label{def:regular}
      Given $\tau, M>0$, we say a sequence of $k\times k$ positive semidefinite matrices $\br{A_1,\ldots, A_n}$ is $\br{\tau, M}$-regular if
      \begin{equation}\label{eq:conditiononAi}
      \id\preceq\sum_{i=1}^n \br{A^i}^2\preceq M\cdot\id ~\mbox{and}~  A^i\preceq\tau\cdot\id \text{ for every } i\in [m]
    \end{equation}
    \end{definition}
    A \emph{spectrahedron} $S\subseteq\reals^k$ is a feasible region of a semidefinite program. Namely, the set $S=\set{x\in\reals^n:\sum_ix_iA^i\preceq B}$ for some symmetric matrices $A_1,\ldots, A_n, B$. We say $S$ is a \emph{positive spectrahedron} if either all $A^i$s are positive semidefinite or all $A^i$s are negative semidefinite $(\NSD)$. Moreover, it is $\br{\tau, M}$-regular if either $\br{A_1,\ldots, A_n}$ or $\br{-A_1,\ldots, -A_n}$ is $\br{\tau, M}$-regular.

      We say $S$ is an intersection of positive spetrahedrons if $S=S_1\cap S_2$ where $S_1$ and $S_2$ are positive spectrahedra whose matrices are all positive semidefinite and negative semidefinite, respectively. Note that it suffices to consider the intersections of two spetrahedrons as one can pack all $\PSD$ matrices into one large block-diagonal matrix (looking ahead this will only affect the parameters in our main results by a logarithmic factor). Packing the corresponding $B_i$s, one get a positive spectrahedron. Same for all negative semidefinite matrices.

    \subsection{Pseudorandomness}
    \begin{definition}
      A function $g:\set{-1,1}^r\rightarrow\set{-1,1}^n$ with seed length $r$, is said to $\delta$-fool a function $f:\set{-1,1}^n\rightarrow\reals$~if
      \[
      \abs{\E_{\bs\sim\U_r}\Br{f\br{g\br{\bs}}}-\E_{\bu\sim\U_n}[f\br{
      \bu}]}\leq\delta.
      \]
      The function $g$ is said to be an efficient pseudorandom generator $(\PRG)$ that $\delta$-fools a class $\Fe$ of $n$-variable functions if $g$ is computable by a deterministic uniform poly$(n)$-time algorithm and $g$ fools all function $f\in\Fe$.
    \end{definition}

    \subsection{Tensors}
     For $\ell\geq 1$, let $T^\ell$ be an $\ell$-tensor, i.e., $T^\ell:(\R^k)^{\times \ell}\rightarrow \R$. Note that an $\ell$-tensor is defined uniquely by the coefficients $\{T_{i_1,\ldots,i_\ell}:i_1,\ldots,i_\ell\in [k]\}$. Below we abuse notation by letting $T(i_1,\ldots,i_\ell)=T_{i_1,\ldots,i_\ell}$. Often we will use the natural bijection between $2\ell$-tensors acting on $\R^k$ and $\ell$-tensors acting on $\mat{k}$, i.e.,  for a $2\ell$-tensor $T:(\R^{k})^{\times 2\ell}\rightarrow \R$ defined~as
    $$
    T(x^1,\ldots,x^{2\ell})=\sum_{i_1,\ldots,i_{2\ell}\in [k]}T(i_1,\ldots,i_\ell,i_{\ell+1},\ldots,i_{2\ell})x^1_{i_1}\cdots x^{2\ell}_{i_{2\ell}},
    $$
    we can also view $T$ as $T':(\mat{k})^{\times \ell}\rightarrow \R$ defined by rearranging the terms above to obtain:
    $$
    T'(X^1,\ldots,X^\ell)=\sum_{i_1,j_1\in [n]}\sum_{i_2,j_2\in [k]}\cdots\sum_{i_\ell,j_\ell\in [k]}T(i_1,\ldots,i_\ell,j_1,\ldots,j_\ell)X^1_{i_1,j_1}\cdots X^{\ell}_{i_{\ell},j_\ell}
    $$

    Finally, we define a ``permutation folding" operator which takes a $(2\ell)$-tensor on $\R^k$ as defined above and produces a permutation to produce an $\ell$-tensor on $\mat{k}$.

    \begin{definition}~\cite{sendov2007higher}[Definition of $\diag^\sigma T$]
    \label{def:diadsigmadefinition}
    Let $T:(\R^k)^{\times t}\rightarrow \R$ be a $k$-tensor and $\sigma\in S_k$. Then we define $\diag^\sigma T:(\mat{k})^{\times t}\rightarrow \R$ as the following map
    \begin{align}
    \br{\diag^\sigma T}\br{(i_1,j_1)\ldots,(i_k,j_k)}=T(i_1,\ldots,i_k)\quad \text{ iff } \vec{i}=\sigma\vec{j},
    \end{align}
    and $0$ otherwise.
    \end{definition}

    \section{Bentkus mollifier}
    \label{sec:mollifier}
    In this paper, we are interested in smooth approximators of the function $\psi:\reals^k\rightarrow\reals$ defined as
    \begin{equation}\label{eqn:psi}
      \psi\br{x}=\Br{\max_ix_i\leq 0}.
    \end{equation}
    To this end, we introduce the Bentkus mollifier defined by Bentkus in~\cite{bentkus1990smooth} and establish several new properties. Readers may refer to~\cite{bentkus1990smooth,fang2020high} for a more thorough treatment.

    \begin{definition}~\cite{bentkus1990smooth}\label{def:bentkus}
    Let $g:\reals\rightarrow\reals$ be a function defined as
    \begin{equation}\label{eqn:gfunction}
        g\br{x}=\int_{-\infty}^{x}\frac{1}{\sqrt{2\pi}}e^{-t^2/2}dt
    \end{equation}
    For every integer $k\geq 1$, define  $G:\reals^k\rightarrow~\reals$~as
    \begin{equation}\label{eqn:Gfunction}
        G\br{x_1,\ldots,x_k}=\prod_{i=1}^kg\br{x_i}.
    \end{equation}
    The subscript $k$ may be omitted whenever it is clear from the context.
    \end{definition}

    \subsection{Properties of the mollifier and its derivatives}
    From the definition of $g$ in Eq.~\eqref{eqn:gfunction}, it is easy to calculate that
    \begin{eqnarray}
      &&g'\br{x}=\frac{1}{\sqrt{2\pi}}e^{-x^2/2} \label{eqn:g'}\\
      &&g''\br{x}=-\frac{x}{\sqrt{2\pi}}\exp\br{-x^2/2}\label{eqn:g''}\\
      &&g'''\br{x}=\frac{1}{\sqrt{2\pi}}\br{x^2-1}\exp\br{-x^2/2}. \label{eqn:g'''}
    \end{eqnarray}
    In order to simplify calculations, we introduce the function
    \begin{equation}\label{eqn:phibar}
      \bar{g}\br{x}=\frac{g'\br{x}}{g\br{x}}.
    \end{equation}
    \begin{fact}\label{fac:ophimonotone}\cite[Page 10]{fang2020high}
      It holds that
      \begin{align}
    \label{eq:derivativeofophi}
        &\ophi'(u)=-(u+\ophi(u))\cdot \ophi(u);\\
        &\ophi''\br{u}=\br{u^2-1}\ophi\br{u}+3u\ophi\br{u}^2+2\ophi\br{u}^3.
    \end{align}
    Also  $\ophi$ is positive and monotone decreasing in $\reals$. $\ophi'$ is negative in $\reals$.
    \end{fact}
    \begin{fact}~\cite[Section 7.1]{feller2008introduction}\label{fac:tailboundGaussian}
      For any $x\geq 0$, it holds that
      \[\frac{e^{-x^2/2}}{\sqrt{2\pi}}\br{\frac{1}{x}-\frac{1}{x^3}}\leq1-g\br{x}\leq\frac{e^{-x^2/2}}{x\sqrt{2\pi}}.\]
    \end{fact}

    The following lemma immediately follows from Fact~\ref{fac:ophimonotone} and Fact~\ref{fac:tailboundGaussian}.

    \begin{lemma}\label{lem:ophiupperbound}
    For any $\Delta\geq 1$ and $x\in\reals$ with $\abs{x}\leq\Delta$, it holds that
    $$\abs{\ophi\br{x}}\leq 2\Delta,\abs{\ophi'\br{x}}\leq 3\Delta\abs{\ophi\br{x}}, \abs{\ophi''\br{x}}\leq 15\Delta^2\abs{\ophi\br{x}}.$$
    \end{lemma}

    \subsection{Properties of the spectral norm of the mollifier}
In this section, we establish several properties of Bentkus mollifier, which hasn't been studied to the best of our knowledge. We first state a crucial fact that Bentkus proved about the derivatives of the mollifier, which is the only fact needed and used by prior works~\cite{harsha2013invariance,servedio2017fooling,chattopadhyay2019simple,o2019fooling}.
    \begin{fact}\cite{bentkus1990smooth}\label{fac:benktus2}
      It holds that for any integer $t,k\geq 1$
      \begin{equation}\label{eqn:supremumphi}
        \sup_{x\in\reals^k}\onenorm{G^{(t)}\br{x}}\leq C_t\log^{t/2} (k+1)
      \end{equation}
      for some constant $C_t$ only depending on $t$.
    \end{fact}
    
    \begin{lemma}\label{lem:normlowvalue}
      For any $x\in \reals^k$, if there exist more than $3\log k$ indices satisfying $x_i\leq 0$, then $\onenorm{G^{(1)}\br{x}}\leq O\br{\frac{1}{k^2}}$.
    \end{lemma}

    \begin{proof}
    Note that $g\br{z}\leq\frac{1}{2}$ if $z\leq 0$. Let $T=\set{i:x_i\leq 0}$. Then
      \begin{eqnarray*}
        \onenorm{G^{(1)}\br{x}}&=&\sum_{i=1}^k\abs{g'\br{x_i}\prod_{j\neq i}g\br{x_j}}\\
        &=&\sum_{i\in T}\abs{g'\br{x_i}\prod_{j\neq i}g\br{x_j}}+\abs{\sum_{i\notin T}g'\br{x_i}\prod_{j\neq i}g\br{x_j}}\\
        &\leq&\frac{\abs{T}}{2^{\abs{T}-1}}+\frac{1}{2^{\abs{T}}}\abs{\sum_{i\notin T}g'\br{x_i}\prod_{\substack{j\neq i:\\j\notin T}}g\br{x_j}}\leq\frac{\abs{T}}{2^{\abs{T}-1}}+\frac{2\sqrt{2\log k}}{2^{\abs{T}}},
      \end{eqnarray*}
     where the equality used that the terms are all positive and the second inequality is from Fact~\ref{fac:benktus2} and that 
     \[\abs{\sum_{i\notin T}g'\br{x_i}\prod_{\substack{j\neq i:\\j\notin T}}g\br{x_j}}=\onenorm{G^{(1)}\br{x_{T^c}}}.\]
      The upper bound is $O\br{\frac{1}{k^2}}$ if $\abs{T}\geq 3\log k$.
    \end{proof}

    \begin{claim}\label{claim:gg'}
      For any $x>y$, it holds that
         \begin{equation}\label{eqn:gg'}
        \abs{\frac{g\br{x}g'\br{y}-g'\br{x}g\br{y}}{x-y}}\leq\br{1+\abs{x}}\exp\br{-\frac{y^2}{2}}=\br{1+\abs{x}}g'(y)\cdot \sqrt{2\pi}.
      \end{equation}
    \end{claim}
    \begin{proof}
      \begin{eqnarray*}
      &&\abs{\frac{g\br{x}g'\br{y}-g'\br{x}g\br{y}}{x-y}} \\
      &=&\frac{1}{2\pi}\abs{\int_{-\infty}^{0}\frac{\exp\br{-\frac{1}{2}\br{y^2+\br{t+x}^2}}-\exp\br{-\frac{1}{2}\br{x^2+\br{t+y}^2}}}{x-y}dt} \\
      &\leq&\frac{1}{2\pi}\exp\br{-\frac{x^2+y^2}{2}}\int_{-\infty}^{0}\abs{\exp\br{-\frac{t^2}{2}}\frac{\exp\br{-ty}-\exp\br{-tx}}{x-y}}dt\\
      &=&\frac{1}{2\pi}\exp\br{-\frac{x^2+y^2}{2}}\int_{-\infty}^{0}\abs{\exp\br{-\frac{t^2}{2}-tx}\frac{1-\exp\br{-t(y-x)}}{y-x}}dt\\
      &\leq&\frac{1}{2\pi}\exp\br{-\frac{x^2+y^2}{2}}\int_{-\infty}^{0}\abs{\exp\br{-\frac{t^2}{2}-tx}t}dt\\
      &=&\frac{1}{2\pi}\exp\br{-\frac{y^2}{2}}\int_{-\infty}^{0}\abs{\exp\br{-\frac{1}{2}\br{t+x}^2}t}dt\\
      &=&\frac{1}{2\pi}\exp\br{-\frac{y^2}{2}}\br{\exp\br{-\frac{x^2}{2}}+\sqrt{2\pi}x-x\int_{x}^{\infty}e^{-t^2/2}dt}\\
      &\leq&\br{1+\abs{x}}\exp\br{-\frac{y^2}{2}},
      \end{eqnarray*}
      where the second inequality used $|1-e^{-z}|\leq |z|$ for $z\geq 0$.
    \end{proof}
    For every $\theta>0$, we define the Bentkus mollifier as follows.
    \begin{equation}\label{eqn:flambda}
      G_{\theta}\br{x}=\Pr_{\bg\sim\G^k}\big[x+\theta\bg\leq 0\big]
    \end{equation}
    It is not hard to verify that
    \begin{equation}\label{eqn:gtheta}
      G_{\theta}\br{x}=\prod_{i=1}^{n}\int_{-\infty}^{-\frac{x_i}{\theta}}\frac{1}{\sqrt{2\pi}}e^{-x_i^2/2}=G\br{-\frac{x_1}{\theta},\cdots,-\frac{x_k}{\theta}}.
    \end{equation}

     The following fact states that $G_{\theta}\br{\cdot+\alpha}/G_{\theta}\br{\cdot-\alpha}$ is a good approximator of $\psi$ defined in Eq.~\eqref{eqn:psi} except a small inner/outer region near the ``boundary" which is made precise below.
    \begin{fact}[Lemma 6.7 and Fact 6.8 in \cite{o2019fooling}]\label{fac:ost}
      For any $\delta,\theta\in(0,1)$, $x\in\reals^k$ there exists $\Lambda=\Theta\br{\theta\cdot \sqrt{\log (k/\delta)}}$ and $\alpha=\Theta\br{\theta\cdot \sqrt{\log (k/\delta)}}$ such that the following holds.

      \begin{enumerate}
        \item $\abs{G_{\theta}\br{x+\alpha}-\psi\br{x}}\leq\delta$ if $\max_ix_i\leq-\Lambda$.
        \item $\abs{G_{\theta}\br{x-\alpha}-\psi\br{x}}\leq\delta$ if $\max_ix_i\geq\Lambda$.
        \item $G_{\theta}\br{x+\alpha}-\delta\leq\psi\br{x}\leq G_{\theta}\br{x-\alpha}+\delta$ for all $x\in\reals^k$.
      \end{enumerate}
      where $x+\alpha=\br{x_1+\alpha,\ldots,x_k+\alpha}$
    \end{fact}
    Let $A^i=\diag\br{A^i_1,A^i_2}$ and $D=\diag\br{D_1,D_2}$ be  block diagonal matrices. To keep the notations succinct, we set $A\br{x}= \sum_i x_iA^i-D$.

    \begin{fact}\cite[Lemma 6.9]{o2019fooling}\label{fac:ostanti}
      Let $k, \delta,\theta,\Lambda,\alpha$ be the parameters satisfying Fact~\ref{fac:ost}. Let $\Psi,\Psi_{\theta}:\sym{k}\rightarrow\reals$ be the functions defined as $\Psi\br{M}=\psi\br{\lambda\br{M}}$, $\Psi_{\theta}\br{M}=G_{\theta}\br{\lambda\br{M}}$, where $\psi$ is defined in Eq.~\eqref{eqn:psi} and $G_{\theta}$ is defined in Eq.~\eqref{eqn:flambda}, $\bx$ and $\bx'$ be two random variables in $\reals^k$ satisfying~that
       \[
       \abs{\E\Br{\Psi_{\theta}\br{A\br{\bx}+\beta\id}}-\E\Br{\Psi_{\theta}\br{A\br{\bx'}+\beta\id}}}\leq\eta,
       \]
       for both $\beta=\alpha$ and $\beta=-\alpha$. Then, it holds that
      \[
      \abs{\E\Br{\Psi\br{A\br{\bx}}}-\E\Br{\Psi\br{A\br{\bx'}}}}\leq\eta+3\delta+\Pr\Br{\lambda_{\max}\br{A\br{\bx}}\in(-\Lambda,\Lambda]}.
      \]
    \end{fact}

    \section{Computing spectral derivatives}

    \label{sec:spectralderivatives}

In this section use the result by Sendov~\cite{sendov2007higher}  to bound the spectral derivatives of~functions.

    \subsection{Formulas for spectral derivatives}

    Before we describe the main theorem of this section, we need the following notation introduced by Sendov in~\cite{sendov2007higher} to calculate the high-order Fr\'echet derivatives of spectral functions.

    \begin{definition}\cite{sendov2007higher}
    \label{def:definingtensorssendov}
    Let $t\geq 1$ and $x\in \R^t$. Let $T:(\R^k)^{\times t}\rightarrow \R$ be a $t$-tensor. For every, $\ell\in [t]$, define  a $(t+1)$-tensor $T^{\ell}_{\out}:(\R^k)^{\times (t+1)}\rightarrow \R$ as~follows
     \[
     (T^{\ell}_{\out})(i_1,\ldots,i_{t+1})= \begin{cases}
              0 & i_\ell=i_{t+1} \\
              \frac{T(i_1,\ldots,i_{\ell-1},i_{t+1},i_{\ell+1},\ldots,i_t)-T(i_1,\ldots,i_{\ell-1},i_{\ell},i_{\ell+1},\ldots,i_t)}{x_{i_{t+1}}-x_{i_{\ell}}} & i_\ell\neq i_{t+1}.
           \end{cases}
        \]
    Finally, for every $\ell\in [t]$, define
     \[
     T_{\sigma}(x)= \begin{cases}
              \nabla f(x) & \ell=1, \sigma=(1) \\
              \br{ T(x)}^{\ell}_{\out} &\ell\leq t-1\\
              \nabla T_{\sigma} (x) &\ell= t,
           \end{cases}
        \]
        where  $\sigma(\ell)$ is defined as follows: let $\sigma$ be a permutation of $[k]$ given in the cycle decomposition, then $\sigma(\ell)$ is a permutation of $[k+1]$ elements whose cycle representation is the same as $\sigma$ except that the element $k+1$ is inserted after the $\ell$th element and before the $(\ell+1)$th element in the cycle representation of $\sigma$.\footnote{For better intuition, consider a simple example: let $\sigma=(12)(3)$ be a permutation on $[3]$, then $\sigma(\cdot)$ is a permutation on $[4]$ defined as follows: $\sigma(1)$ is $(142)(3)$, similarly $\sigma(2)=(124)(3)$, $\sigma(3)=(12)(34)$, $\sigma(4)=(12)(3)(4)$.}
    \end{definition}

    We are now ready to state the Sendov's formula for high-order Fr\'echet derivatives of spectral functions.
    \begin{theorem}\cite{sendov2007higher}
    \label{thm:sendovmain}
     Let $F:\sym{k}\rightarrow \R$ be a spectral function (i.e., $F=f\circ \lambda$ for $f:\R^k\rightarrow \R$). Then for any $X\in\sym{k}$ satisfying that all the eigenvalues are distinct, $F$ is $t$-times differentiable at $X$ \emph{if and only if} $f$ is $t$-times differentiable at $\lambda(X)$. If $f(x_1,\ldots,x_n)=\sum_{i=1}^n g\br{x_i}$ for $g:\reals\rightarrow\reals$, then for any $X\in\sym{k}$, $F$ is $t$-times Fr\'echet differentiable at $X$ \emph{if and only if} $f$ is $t$-times differentiable at $\lambda(X)$, i.e., the distinctness of the eigenvalues is not necessary anymore.

    Moreover, for every $\sigma \in S_t,x\in \R^k$, let $T_{\sigma}(x): (\R^k)^{\times t}\rightarrow \R$ be a $t$-tensor as defined in Definition~\ref{def:definingtensorssendov} (which depends on the function $f$).\footnote{Think of $x\in \R^k$ as the eigenvalues of $X\in\sym{k}$, i.e., $x=\lambda(X)$.} Then,
    for every $U_1,\ldots,U_t\in \sym{k}$, we have
    $$
    D^t F\br{X}\Br{U_1,\ldots,U_t}=\br{\sum_{\sigma\in S_t}\diag^\sigma T_{\sigma}(\lambda(X)) }(V^TU_1V,\ldots,V^TU_tV),
    $$
    where $V$ satisfies $X=V\br{\diag(\lambda(X)}V^T$ and $\diag^\sigma T:(\mat{k})^t\rightarrow \R$ is a $t$-tensor on the set $\sym{k}$ (as defined in Definition~\ref{def:diadsigmadefinition}).
    \end{theorem}

    \subsection{Third order Fr\'echet derivatives of smooth functions}

    In this section, we explicitly compute the third order Fr\'echet derivatives of spectral functions.
    \begin{theorem}
    \label{thm:spectralthm}
    Let $k,n\geq 1$. Let $f:\R^k\rightarrow \R$ be a $3$-times differentiable symmetric function and $\lambda:\sym{k}\rightarrow \R^k$ be the map $\lambda(M)=\br{\lambda_1(M),\ldots,\lambda_k(M)}$ for every $M\in \sym{k}$. Let $F:\sym{k}\rightarrow \R$ be defined as $F(M)=(f\circ \lambda)(M)$ for all $M\in \sym{k}$. Then, for every $P\in \sym{k}$ with \emph{distinct} eigenvalues and $H\in\sym{k}$, let $P=V\br{\diag\br{\lambda\br{P}}}V^T$ be a spectral decomposition of $P$ and $H=VQV^T$. Then $D^3F\br{P}\Br{Q,Q,Q}$ is the summation of the following terms.

    \begin{enumerate}
    	\item $\sum_{i_1}\nabla^3_{i_1,i_1,i_1}f\br{x}H_{i_1,i_1}^3$
    	\item $\sum_{i_1\neq i_2}\nabla^3_{i_1,i_2,i_1}f\br{x}H_{i_1,i_1}^2H_{i_2,i_2}$
        \item $ \sum_{i_1\neq i_2\neq i_3}(\nabla^3_{i_1,i_2,i_3} f\br{x})\cdot  H_{i_1,i_1}H_{i_2,i_2}H_{i_3,i_3}$
        \item $\sum_{i_1 \neq i_2} \br{\frac{\nabla^2_{i_2,i_2}-\nabla^2_{i_1,i_2}}{x_{i_2}-x_{i_1}}-\frac{\nabla_{i_2}-\nabla_{i_1}}{(x_{i_2}-x_{i_1})^2}}f\br{x}H_{i_2,i_2}H_{i_2,i_1}^2$
    \item $\sum_{i_1\neq i_2\neq i_3}\frac{\nabla^2_{i_2,i_3}-\nabla^2_{i_1,i_3}}{x_{i_2}-x_{i_1}}f\br{x}H_{i_1,i_2}^2H_{i_3,i_3}$
    \item $\sum_{i_1\neq i_2\neq i_3}\br{\frac{\nabla_{i_3}-\nabla_{i_1}}{(x_{i_3}-x_{i_2})(x_{i_3}-x_{i_1})}-      \frac{\nabla_{i_2}-\nabla_{i_1}}{(x_{i_3}-x_{i_2})(x_{i_2}-x_{i_1})}}f\br{x}H_{i_1,i_2}H_{i_2,i_3}H_{i_3,i_1}$
        \item $\sum_{i_1\neq i_2\neq i_3}\br{ \frac{\nabla_{i_2}-\nabla_{i_3}}{(x_{i_3}-x_{i_1})(x_{i_2}-x_{i_3})}-      \frac{\nabla_{i_2}-\nabla_{i_1}}{(x_{i_3}-x_{i_1})(x_{i_2}-x_{i_1})}}f\br{x}H_{i_1,i_3}H_{i_2,i_1}H_{i_3,i_2},$
    \end{enumerate}
    where $x=\br{\lambda_1\br{P},\ldots,\lambda_k\br{P}}$.
    \end{theorem}

    \begin{proof}
    To prove this theorem, we first apply Theorem~\ref{thm:sendovmain} for $t=3$ to obtain
    \begin{align}
    \label{eq:usingsendovinmainthm}
    D^3 F\br{P}\Br{Q,Q,Q}=\br{\sum_{\sigma\in S_3}\diag^\sigma T_{\sigma}(\lambda(P)) }(H,H,H).
    \end{align}
    We next carefully express each quantity in the summation using the definition of these tensors and upper bound each term. To this end, we break down all the six elements of $S_3$ and analyze them separately as follows.

    \textbf{Case 1: $\sigma=(1)(2)(3)$.} Then $T_\sigma(x)=\nabla^3f(x)$.

    \textbf{Case 2: $\sigma=(12)(3)$.} First, observe that considering $\sigma=(12)$ we get

    \[ \br{T_{(12)}(x)}_{i_1,i_2}=\begin{cases}
          0 & i_1=i_2 \\
          \frac{1}{x_{i_2}-x_{i_1}}\cdot \br{\nabla_{i_2}-\nabla_{i_1}}f\br{x} & i_1\neq i_2 \\
       \end{cases}
    \]
    Now, in order to compute $T_{(12)(3)}$, we need to compute $\nabla T_{(12)}(x)$ which can be written as follows
    \begin{align*} &\br{T_{(12)(3)}(x)}_{i_1,i_2,i_3}\\
    &=\begin{cases}
          0 & i_1=i_2 \\
          \frac{1}{x_{i_3}-x_{i_1}}\cdot \br{\nabla^2_{i_3,i_3}-\nabla^2_{i_1,i_3}}f\br{x}-      \frac{1}{(x_{i_3}-x_{i_1})^2}\cdot \br{\nabla_{i_3}-\nabla_{i_1}}f\br{x} & i_2=i_3\neq i_1 \\
          \frac{1}{x_{i_2}-x_{i_3}}\cdot \br{\nabla^2_{i_2,i_3}-\nabla^2_{i_3,i_3}}f\br{x}+      \frac{1}{(x_{i_2}-x_{i_3})^2}\cdot \br{\nabla_{i_2}-\nabla_{i_3}}f\br{x} & i_1=i_3\neq i_2 \\
          \frac{1}{x_{i_2}-x_{i_1}}\cdot \br{\nabla^2_{i_2,i_3}-\nabla^2_{i_1,i_3}}f\br{x} & i_1\neq i_2\neq i_3 \\
       \end{cases}
    \end{align*}
    \textbf{Case 3: $\sigma=(13)(2)$.} First note that for $\sigma=(1)(2)$, we have $T_{(1)(2)}=\nabla^2f$ and $\sigma(1)=(13)(2)$. So, we need to compute $\br{\nabla^2 f}f\br{x}^1_{\out}$ and we get

    \[ \br{T_{(13)(2)}(x)}_{i_1,i_2,i_3}=\begin{cases}
          0 & i_1= i_3 \\
          \frac{1}{x_{i_3}-x_{i_1}}\cdot \br{\nabla^2_{i_3,i_2}-\nabla^2_{i_1,i_2}}f\br{x} & i_1\neq i_3
       \end{cases}
    \]
    \textbf{Case 4: $\sigma=(1)(23)$.} First note that for $\sigma=(1)(2)$, we have $T_{(1)(2)}=\nabla^2f$ and $\sigma(2)=(1)(23)$. So, we need to compute $\br{\nabla^2 f}f\br{x}^2_{\out}$ and we get

    \[ \br{T_{(1)(23)}(x)}_{i_1,i_2,i_3}=\begin{cases}
          0 & i_2= i_3 \\
          \frac{1}{x_{i_3}-x_{i_2}}\cdot \br{\nabla^2_{i_3,i_1}-\nabla^2_{i_2,i_1}}f\br{x} & i_2\neq i_3
       \end{cases}
    \]
    \textbf{Case 5: $\sigma=(123)$.} Let $\sigma=(12)$, then $\sigma(2)=(123)$. So we need to compute $\br{T_{(12)}}f\br{x}^2_{\out}$  and we obtain
    \begin{align*} &\br{T_{(123)}(x)}_{i_1,i_2,i_3}\\
    &=    \begin{cases}
          \frac{1}{(x_{i_2}-x_{i_1})^2}\cdot \br{\nabla_{i_2}-\nabla_{i_1}}f\br{x} & i_2\neq i_3= i_1 \\
          \frac{1}{(x_{i_3}-x_{i_1})^2}\cdot \br{\nabla_{i_3}-\nabla_{i_1}}f\br{x} & i_1=i_2\neq i_3 \\
          \frac{1}{(x_{i_3}-x_{i_2})(x_{i_3}-x_{i_1})}\cdot \br{\nabla_{i_3}-\nabla_{i_1}}f\br{x}-      \frac{1}{(x_{i_3}-x_{i_2})(x_{i_2}-x_{i_1})}\cdot \br{\nabla_{i_2}-\nabla_{i_1}}f\br{x} & i_1\neq i_3\neq i_2 \\
          0 & \text{ otherwise}
       \end{cases}
    \end{align*}
    \textbf{Case 6: $\sigma=(132)$.} Let $\sigma=(12)$, then $\sigma\tau(1)=(132)$. So we need to compute $\br{T_{(12)}}f\br{x}^1_{\out}$  and we obtain.
    \begin{align*} &\br{T_{(132)}(x)}_{i_1,i_2,i_3}\\
    &=    \begin{cases}
          -\frac{1}{(x_{i_2}-x_{i_1})^2}\cdot \br{\nabla_{i_2}-\nabla_{i_1}}f\br{x} & i_1\neq i_3= i_2 \\
          \frac{1}{(x_{i_3}-x_{i_2})^2}\cdot \br{\nabla_{i_3}-\nabla_{i_2}}f\br{x} & i_2=i_1\neq i_3 \\
          \frac{1}{(x_{i_3}-x_{i_1})(x_{i_2}-x_{i_3})}\cdot \br{\nabla_{i_2}-\nabla_{i_3}}f\br{x}-      \frac{1}{(x_{i_3}-x_{i_1})(x_{i_2}-x_{i_1})}\cdot \br{\nabla_{i_2}-\nabla_{i_1}}f\br{x} & i_1\neq i_3\neq i_2 \\
          0 & \text{ otherwise}
       \end{cases}
    \end{align*}
    Using the above cases we can now rewrite Eq.~\eqref{eq:usingsendovinmainthm} as
    $$
        \sum_{\sigma \in S_3}T_\sigma(x)(H,H,H)=\sum_\sigma\sum_{\substack{i_1,i_2,i_3}}      \br{T_{\sigma}(x)}_{i_1,i_2,i_3} H_{i_1,i_{\sigma(1)}}H_{i_2,i_{\sigma(2)}}H_{i_3,i_{\sigma(3)}}
    $$

    Let's write this out as follows: by $T_i$, we mean $T_{case (i)}$ above
    \begin{align*}
         \sum_{i_1,i_2,i_3}&(T_{1})_{i_1,i_2,i_3}H_{i_1,i_1}H_{i_2,i_2}H_{i_3,i_3}+(T_{2})_{i_1,i_2,i_3}H_{i_1,i_2}H_{i_2,i_1}H_{i_3,i_3}+(T_{3})_{i_1,i_2,i_3}H_{i_1,i_3}H_{i_2,i_2}H_{i_3,i_1}\\
        &+(T_{4})_{i_1,i_2,i_3}H_{i_1,i_1}H_{i_2,i_3}H_{i_3,i_2}+(T_{5})_{i_1,i_2,i_3}H_{i_1,i_2}H_{i_2,i_3}H_{i_3,i_1}+(T_{6})_{i_1,i_2,i_3}H_{i_1,i_3}H_{i_2,i_1}H_{i_3,i_2}
    \end{align*}
     and in particular, since $H$ is symmetric the above simplifies to
    \begin{align}
    \label{eq:mainsum}
    \begin{aligned}
         \sum_{i_1,i_2,i_3}&(T_{1})_{i_1,i_2,i_3}H_{i_1,i_1}H_{i_2,i_2}H_{i_3,i_3}+(T_{2})_{i_1,i_2,i_3}H_{i_1,i_2}^2H_{i_3,i_3}+(T_{3})_{i_1,i_2,i_3}H_{i_1,i_3}^2H_{i_2,i_2}\\
        &+(T_{4})_{i_1,i_2,i_3}H_{i_1,i_1}H_{i_2,i_3}^2+(T_{5})_{i_1,i_2,i_3}H_{i_1,i_2}H_{i_2,i_3}H_{i_3,i_1}+(T_{6})_{i_1,i_2,i_3}H_{i_1,i_3}H_{i_2,i_1}H_{i_3,i_2}
    \end{aligned}
    \end{align}

    Now, we will break up this sum into $5$ cases as follows which will give us our theorem statement.

    \paragraph{Case (i): $i_1=i_3\neq i_2$.} Then Eq.~\eqref{eq:mainsum} reduces to the following
    \begin{align}
       \sum_{i_1,i_2} H_{i_1,i_1}^2H_{i_2,i_2}\br{T_1+T_3}+    H_{i_1,i_1}H_{i_2,i_1}^2\br{T_2+T_4+T_5+T_6}
    \end{align}
    Note that when we say $T_q$ above, we mean $(T_q)_{i_1,i_2,i_3}=(T_q)_{i_1,i_2,i_1}$ (since $i_3=i_1$). Let us now plug in the values of the corresponding $T_q$s into the formula and rewrite the above as follows
    \begin{align}
    \begin{aligned}
       &\sum_{i_1 \neq  i_2} H_{i_1,i_1}^2H_{i_2,i_2}\br{\nabla^3_{i_1,i_2,i_1}f\br{x}+0}+\\
       & \quad +H_{i_1,i_1}H_{i_2,i_1}^2\br{\frac{\nabla^2_{i_2,i_1}-\nabla^2_{i_1,i_1}}{x_{i_2}-x_{i_1}}+\frac{\nabla_{i_2}-\nabla_{i_1}}{(x_{i_2}-x_{i_1})^2}+\frac{\nabla^2_{i_1,i_1}-\nabla^2_{i_2,i_1}}{x_{i_1}-x_{i_2}}+\frac{\nabla_{i_2}-\nabla_{i_1}}{(x_{i_2}-x_{i_1})^2}}f\br{x}\\
        &=\sum_{i_1\neq i_2}H_{i_1,i_1}^2H_{i_2,i_2}\br{\nabla^3_{i_1,i_2,i_1}f\br{x}}+2H_{i_1,i_1}H_{i_2,i_1}^2\br{\frac{\nabla^2_{i_2,i_1}-\nabla^2_{i_1,i_1}}{x_{i_2}-x_{i_1}}+\frac{\nabla_{i_2}-\nabla_{i_1}}{(x_{i_2}-x_{i_1})^2}}f\br{x}
       \end{aligned}
    \end{align}
    \textbf{Case (ii): $i_1=i_2\neq i_3$.} Then Eq.~\eqref{eq:mainsum} reduces to
    \begin{align}
       \sum_{i_1,i_3} H_{i_1,i_1}^2H_{i_3,i_3}\br{T_1+T_2}+    H_{i_1,i_1}H_{i_3,i_1}^2\br{T_3+T_4+T_5+T_6}
    \end{align}
    The above simplies to the following
    \begin{align}
    \begin{aligned}
       &\sum_{i_1\neq i_3} H_{i_1,i_1}^2H_{i_3,i_3}\br{\nabla^3_{i_1,i_1,i_3}f\br{x}+0}+\\
       & \quad +H_{i_1,i_1}H_{i_3,i_1}^2\br{\frac{\nabla^2_{i_3,i_1}-\nabla^2_{i_1,i_1}}{x_{i_3}-x_{i_1}}+\frac{\nabla^2_{i_3,i_1}-\nabla^2_{i_1,i_1}}{x_{i_3}-x_{i_1}}+\frac{\nabla_{i_3}-\nabla_{i_1}}{(x_{i_3}-x_{i_1})^2}+\frac{\nabla_{i_3}-\nabla^2_{i_1}}{(x_{i_3}-x_{i_1})^2}}f\br{x}\\
       &=\sum_{i_1 \neq i_3} H_{i_1,i_1}^2H_{i_3,i_3}\br{\nabla^3_{i_1,i_1,i_3}f\br{x}}+
       2H_{i_1,i_1}H_{i_3,i_1}^2\br{\frac{\nabla^2_{i_3,i_1}-\nabla^2_{i_1,i_1}}{x_{i_3}-x_{i_1}}+\frac{\nabla_{i_3}-\nabla_{i_1}}{(x_{i_3}-x_{i_1})^2}}f\br{x}
       \end{aligned}
    \end{align}
    \textbf{Case (iii): $i_2=i_3\neq i_1$.} Then Eq.~\eqref{eq:mainsum} reduces to
    \begin{align}
       \sum_{i_1,i_2} H_{i_2,i_2}^2H_{i_1,i_1}\br{T_1+T_4}+    H_{i_2,i_2}H_{i_2,i_1}^2\br{T_2+T_3+T_5+T_6}
    \end{align}
    The above simplifies to the following
    \begin{align}
    \begin{aligned}
       &\sum_{i_1\neq i_2} H_{i_2,i_2}^2H_{i_1,i_1}\br{\nabla^3_{i_1,i_2,i_2}f\br{x}+0}+\\
       & \quad +H_{i_2,i_2}H_{i_2,i_1}^2\br{\frac{\nabla^2_{i_2,i_2}-\nabla^2_{i_1,i_2}}{x_{i_2}-x_{i_1}}-\frac{\nabla_{i_2}-\nabla_{i_1}}{(x_{i_2}-x_{i_1})^2}+\frac{\nabla^2_{i_2,i_2}-\nabla^2_{i_1,i_2}}{x_{i_2}-x_{i_1}}-\frac{\nabla_{i_2}-\nabla_{i_1}}{(x_{i_2}-x_{i_1})^2}}f\br{x}\\
        &=\sum_{i_1 \neq i_2} H_{i_2,i_2}^2H_{i_1,i_1}\br{\nabla^3_{i_1,i_2,i_2}f\br{x}}+ 2H_{i_2,i_2}H_{i_2,i_1}^2\br{\frac{\nabla^2_{i_2,i_2}-\nabla^2_{i_1,i_2}}{x_{i_2}-x_{i_1}}-\frac{\nabla_{i_2}-\nabla_{i_1}}{(x_{i_2}-x_{i_1})^2}}f\br{x}
       \end{aligned}
    \end{align}
    \textbf{Case (i)+ Case (ii)+ Case (iii).} We first upper bound these three cases to get the desired upper bound in the theorem statement. First summing the three cases, we have
    \begin{align}
    \label{eq:sumcase123}
    \begin{aligned}
    \sum_{i_1\neq i_2}&H_{i_1,i_1}^2H_{i_2,i_2}\br{\nabla^3_{i_1,i_2,i_1}+\nabla^3_{i_1,i_2,i_2}+\nabla^3_{i_2,i_1,i_1}}f\br{x}\\
    &+6\sum_{i_1 \neq i_2} H_{i_2,i_2}H_{i_2,i_1}^2\underbrace{\br{\frac{\nabla^2_{i_2,i_2}-\nabla^2_{i_1,i_2}}{x_{i_2}-x_{i_1}}-\frac{\nabla_{i_2}-\nabla_{i_1}}{(x_{i_2}-x_{i_1})^2}}f\br{x}}_{(\star)}
    \end{aligned}
    \end{align}

    \paragraph{Case (iv): $i_2=i_3= i_1$.} Then Eq.~\eqref{eq:mainsum} reduces to
    \begin{align}
       \sum_{i_1} H_{i_1,i_1}^3\br{T_1+T_2+T_3+T_4+T_5+T_6}=\sum_{i_1}H_{i_1,i_1}^3\nabla^3_{i_1,i_1,i_1}f
    \end{align}

    \paragraph{Case (v): $i_2\neq i_3\neq  i_1$.} Then Eq.~\eqref{eq:mainsum} stays the same and we get
    \begin{align}
    \label{eq:mainsumreduction}
    \begin{aligned}
         \sum_{i_1,i_2,i_3}&(\nabla^3_{i_1,i_2,i_3} f)\cdot  H_{i_1,i_1}H_{i_2,i_2}H_{i_3,i_3}\\
         &+\frac{\nabla^2_{i_2,i_3}-\nabla^2_{i_1,i_3}}{x_{i_2}-x_{i_1}}f\br{x}H_{i_1,i_2}^2H_{i_3,i_3}
         + \frac{\nabla^2_{i_3,i_2}-\nabla^2_{i_1,i_2}}{x_{i_3}-x_{i_1}}f\br{x} H_{i_1,i_3}^2H_{i_2,i_2}
        +\frac{\nabla^2_{i_3,i_1}-\nabla^2_{i_2,i_1}}{x_{i_3}-x_{i_2}}f\br{x}H_{i_1,i_1}H_{i_2,i_3}^2\\
       & +\br{\frac{\nabla_{i_3}-\nabla_{i_1}}{(x_{i_3}-x_{i_2})(x_{i_3}-x_{i_1})}-      \frac{\nabla_{i_2}-\nabla_{i_1}}{(x_{i_3}-x_{i_2})(x_{i_2}-x_{i_1})}}f\br{x}H_{i_1,i_2}H_{i_2,i_3}H_{i_3,i_1}\\
        &+\br{ \frac{\nabla_{i_2}-\nabla_{i_3}}{(x_{i_3}-x_{i_1})(x_{i_2}-x_{i_3})}-      \frac{\nabla_{i_2}-\nabla_{i_1}}{(x_{i_3}-x_{i_1})(x_{i_2}-x_{i_1})}}f\br{x}H_{i_1,i_3}H_{i_2,i_1}H_{i_3,i_2}
    \end{aligned}
    \end{align}
    This concludes the proof of the theorem statement.
    \end{proof}
    \subsection{Main theorem: Fr\'echet derivatives of Bentkus function}
    We now state the main theorem which bounds all the terms that appear in the theorem in the previous section. Let $G:\reals^k\rightarrow\reals$ be the Bentkus function given in Definition~\ref{def:bentkus}.
    \begin{theorem}\label{thm:derivative}
      Let $k\geq 1$ be an integer and $\Psi:\sym{k}\rightarrow\reals$ be a function defined as $\psi\br{M}=\br{G\circ\lambda}\br{M}$ where $G$ is given in Definition~\ref{def:bentkus}. Given $\Delta\geq 1$ $X\in\sym{k}$ with eigenvalues $\lambda\br{X}=\br{x_1,\ldots, x_k}$ satisfying that $\norm{X}\leq\Delta$, it holds that
      \[\abs{D^3\Psi\br{X}\Br{H,H,H}}\leq O\br{\Delta^2\cdot \log^3 k\cdot \|H\|^3}.\]
    \end{theorem}
    The following corollary simply follows from the definition of $G_{\theta}$ in Eq.~\ref{eqn:flambda} and the chain rule of Fr\'echet derivatives in Fact~\ref{fac:frechetderivative}.

    \begin{corollary}\label{cor:derivaativebentkus}
      Let $k\geq 1$ be an integer and $\theta>0, \alpha\in\reals$ and $\Psi_{\theta}:\sym{k}\rightarrow\reals$ be a function defined as $\Psi_{\theta}\br{M}=\br{G_{\theta}\circ\lambda}\br{M+\alpha\id}$, where $G_{\theta}$ is given in Eq.~\eqref{eqn:flambda}. Given $\Delta\geq 1$, $X\in\sym{k}$ with eigenvalues $\lambda\br{X}=\br{x_1,\ldots, x_k}$ satisfying that $\norm{X}\leq\Delta$, it holds that
      \[\abs{D^3\Psi_{\theta}\br{X+\alpha\id}\Br{H,H,H}}\leq O\br{\frac{\Delta^2+\alpha^2}{\theta^3}\cdot\log^3 k\cdot\|H\|^3}.\]
    \end{corollary}

    In order to prove the theorem above, We  upper bound all the terms listed in Theorem~\ref{thm:spectralthm} individually in the following sections (in increasing order of difficulty). Given the calculations are fairly technical we break down the analysis in the following sections for modularity and reader convenience. In Section~\ref{sec:b123} we bound the first three terms in Theorem~\ref{thm:spectralthm} (this is the easy case since the analysis is very similar to what happens in~\cite{harsha2013invariance} by directly using known properties of the Bentkus function), in Section~\ref{sec:b4} and \ref{sec:b5} we bound the fourth and fifth term (this already deviates from the analysis of~\cite{harsha2013invariance}) and finally in Section~\ref{sec:b6} we bound the sixth and seventh term (this calculation is fairly involved and deviates significantly from prior works, since we need to deal with various aspects of Fr\'echet derivatives, new properties of Bentkus function and the \emph{non-diagonal} entries of the matrices $H$ which is unique to the matrix-spectrahedron case and is not faced in~\cite{harsha2013invariance,servedio2017fooling,o2019fooling}).

    As spectral functions and spectral norms are unitarily invariant, Assume that $X=\diag\br{x_1,\ldots, x_k}$ is diagonal without loss of generality. To adopt Theorem~\ref{thm:spectralthm}, we assume that all the $x_1,\ldots, x_n$ are distinct. We claim that the general case follows by the continuity argument:  notice that $\ln (G_{\theta}\br{x})=\sum_{i=1}^n\ln (g\br{-\frac{x_i}{\theta}})$ from Eq.~\eqref{eqn:gtheta}. Thus by Theorem~\ref{thm:sendovmain}, $\ln (G_{\theta}\circ \lambda)$ is infinitely Fr\'echet differentiable at any $X\in\sym{k}$ as $\ln (G_{\theta})$ is infinitely times differentiable. This further implies by definition that $\Psi_{\theta}$ is infinitely  Fr\'echet differentiable.

    \subsection{Bounding terms ($1$)-($5$)  in Theorem~\ref{thm:spectralthm} for Bentkus function}

    Let $G:\reals^k\rightarrow\reals$ be the Bentkus function given in Definition~\ref{def:bentkus}. Recall that $G(x)=\prod_i g(x_i)$, where $g\br{x}=\frac{1}{\sqrt{2\pi}}\int_{-\infty}^{-x}e^{-t^2/2}dt$. Recall the notation $g'\br{x}=\frac{1}{\sqrt{2\pi}}e^{-x^2/2}$ and $\overline{g}\br{x}=g'(x)/g(x)$.

    \subsubsection{Bounding terms $(1,2,3)$ in Theorem~\ref{thm:spectralthm}}
    \label{sec:b123}
    \begin{lemma}[Bounding terms $(1,2,3)$]\label{lem:123}
    The following three terms
    	$$
    	\abs{\sum_{i_1}\nabla^3_{i_1,i_1,i_1}G\br{x}H_{i_1,i_1}^3},\quad  \abs{\sum_{i_1\neq i_2}\nabla^3_{i_1,i_2,i_1}G\br{x}H_{i_1,i_1}^2H_{i_2,i_2}}, \quad  \abs{\sum_{i_1\neq i_2\neq i_3}\nabla^3_{i_1,i_2,i_3} G\br{x}\cdot  H_{i_1,i_1}H_{i_2,i_2}H_{i_3,i_3}}
    	$$
    	can be upper bound by $O(\log^{1.5}k \cdot \|H\|^3)$.
    \end{lemma}

    \begin{proof}
    The first upper bound is straightforward. Observe that
    $$
    \abs{\sum_{i_1}\nabla^3_{i_1,i_1,i_1}G\br{x}H_{i_1,i_1}^3}\leq \max_{i}|H_{i,i}|^3\cdot \sum_{i_1}\abs{\nabla^3_{i_1,i_1,i_1}G\br{x}}\leq \max_{i}|H_{i,i}|^3\cdot \|G^{(3)}\br{x}\|_1\leq \|H\|^3\cdot \log^{1.5}k,
    $$
    where the second inequality follows by definition of $\|G^{(3)}\|_1$ and the last inequality used $\max_{i,j}|H_{i,j}|\leq \|H\|$ (the latter being the spectral norm of $H$) and Fact~\ref{fac:benktus2} to conclude $\|G^{(3)}\|_1\leq O\br{\log^{1.5}k}$.
    Similarly, the remaining two terms can also be bounded exactly as above (by observing that $\sum_{i_1\neq i_2}\nabla^3_{i_1,i_2,i_1}G$ and $\sum_{i_1\neq i_2\neq i_3}(\nabla^3_{i_1,i_2,i_3} G)$ appear in the expression of $\|G^{(3)}\|_1$).
    \end{proof}

    \subsubsection{Bounding term ($4$)  in Theorem~\ref{thm:spectralthm}}
    \label{sec:b4}
    In order to bound the remaining terms in Theorem~\ref{thm:spectralthm}, we need the following claim.

    \begin{claim}\label{claim:phiG}
      It holds that
      \begin{enumerate}
        \item $\sum_{i_1\neq i_2}\ophi\br{x_{i_1}}\abs{G\br{x}H_{i_2,i_2}H_{i_1,i_2}^2}\leq O\br{\sqrt{\log k}\cdot\norm{H}^3}.$
        \item $\sum_{i_1\neq i_2}\ophi\br{x_{i_2}}\abs{G\br{x}H_{i_2,i_2}H_{i_1,i_2}^2}\leq O\br{\sqrt{\log k}\cdot\norm{H}^3}.$
        \item $\sum_{i_1\neq i_2\neq i_3}\abs{\ophi\br{x_{i_2}}\ophi\br{x_{i_3}}G\br{x}H_{i_1,i_2}^2H_{i_3,i_3}}\leq O\br{\log k\cdot \norm{H}^3}$.
      \end{enumerate}

    \end{claim}
    \begin{proof}
      For Item 1, we have
      \[\sum_{i_1\neq i_2}\ophi\br{x_{i_1}}\abs{G\br{x}H_{i_2,i_2}H_{i_1,i_2}^2}\leq\sum_{i_1}\ophi\br{x_{i_1}}G\br{x}\cdot\max_{i_1}\sum_{i_2}\abs{H_{i_2,i_2}H_{i_1,i_2}^2}\leq\onenorm{G^{(1)}\br{x}}\cdot \norm{H}^3\]
      where the last inequality is because
      \begin{equation}\label{eqn:HH2}
        \max_{i_1}\sum_{i_2}\abs{H_{i_2,i_2}H_{i_1,i_2}^2}\leq\norm{H}\max_{i_1}\br{H^2}_{i_1,i_1}\leq\norm{H}^3,
      \end{equation}
    using the fact that $\max_{ij}|H_{ij}|\leq \|H\|$. Using Fact~\ref{fac:benktus2} shows the first inequality.  Item 2 follows by the same reason.

    For Item 3, we have
    \begin{eqnarray*}
      \sum_{i_1\neq i_2\neq i_3}\abs{\ophi\br{x_{i_2}}\ophi\br{x_{i_3}}G\br{x}H_{i_1,i_2}^2H_{i_3,i_3}} &=&\sum_{i_2\neq i_3}\abs{\ophi\br{x_{i_2}}\ophi\br{x_{i_3}}G\br{x}}\max_{i_2,i_3}\sum_{i_1}\abs{H_{i_1,i_2}^2H_{i_3,i_3}}\\
      &\leq& O\br{\log k\cdot \norm{H}^3}
    \end{eqnarray*}
    where the inequality is because $\sum_{i_2\neq i_3}\abs{\ophi\br{x_{i_2}}\ophi\br{x_{i_3}}G\br{x}}$ appears in $G^{(2)}(x)$ and then we use Fact~\ref{fac:benktus2} to upper bound it by $\log k$ and additionally we use that
    \begin{equation}\label{eqn:H2H2}
      \max_{i_2,i_3}\sum_{i_1}\abs{H_{i_1,i_2}^2H_{i_3,i_3}}\leq\norm{H}\cdot \max_{i_2}\br{H^2}_{i_2,i_2}\leq\norm{H}^3.
    \end{equation}
    \end{proof}

    \begin{lemma}[Bounding terms $(4)$ in Theorem~\ref{thm:spectralthm}]\label{lem:term4}
    We have
    $$
    \sum_{i_1 \neq i_2} H_{i_2,i_2}H_{i_2,i_1}^2\br{\frac{\nabla^2_{i_2,i_2}G-\nabla^2_{i_1,i_2}G}{x_{i_2}-x_{i_1}}-\frac{\nabla_{i_2}G-\nabla_{i_1}G}{(x_{i_2}-x_{i_1})^2}}\leq O\br{\Delta^2\cdot\sqrt{\log k}\norm{H}^3},
    $$
    where $\Delta=\max_i |x_i|$.
    \end{lemma}
    \begin{proof}
     First observe that
    $$
    \nabla_{i_2}G(x)=g'(x_{i_2})\prod_{j\neq i_2}G(x_j)=\ophi(x_{i_2})\cdot G(x),
    $$
    and similarly we have
    $$
    \nabla_{i_2,i_2}G(x)=\ophi(x_{i_2})\nabla_{i_2}G(x)+G(x)\nabla_{i_2}\ophi(x_{i_2})=\br{\ophi(x_{i_2})^2 -(x_{i_2}+\ophi(x_{i_2}))\ophi(x_{i_2})}G(x)=-x_{i_2}\ophi(x_{i_2})G(x),
    $$
    where we used Fact~\ref{fac:ophimonotone}. Now, let us start upper bounding the lemma statement as follows
    \begin{align*}
    &\abs{\sum_{i_1 \neq i_2} H_{i_2,i_2}H_{i_2,i_1}^2\br{\frac{\nabla^2_{i_2,i_2}G-\nabla^2_{i_1,i_2}G}{x_{i_2}-x_{i_1}}-\frac{\nabla_{i_2}G-\nabla_{i_1}G}{(x_{i_2}-x_{i_1})^2}}} \\
    &\leq \sum_{i_1\neq i_2}\Big|-\frac{\ophi(x_{i_1})\ophi(x_{i_2})+x_{i_2}\ophi(x_{i_2})}{x_{i_2}-x_{i_1}}-\frac{\ophi(x_{i_2})-\ophi(x_{i_1})}{(x_{i_2}-x_{i_1})^2}\Big|\cdot | G(x) \cdot H_{i_2,i_2}H_{i_2,i_1}^2 |\\
    &= \sum_{i_1\neq i_2}\Big|-\frac{\ophi(x_{i_1})\ophi(x_{i_2})+x_{i_2}\ophi(x_{i_2})}{x_{i_2}-x_{i_1}}-\frac{\ophi'(\xi_{i_1,i_2})}{x_{i_2}-x_{i_1}}\Big|\cdot | G(x) \cdot H_{i_2,i_2}H_{i_2,i_1}^2 |\\
    &=  \sum_{i_1\neq i_2} \Big|\frac{\ophi(x_{i_1})\ophi(x_{i_2})+x_{i_2}\ophi(x_{i_2})}{x_{i_2}-x_{i_1}}-\frac{(\xi_{i_1,i_2}+\ophi(\xi_{i_1,i_2}))\ophi(\xi_{i_1,i_2})}{x_{i_2}-x_{i_1}}\Big|\cdot | G(x) \cdot H_{i_2,i_2}H_{i_2,i_1}^2 |\\
    &\leq  \sum_{i_1\neq i_2} \underbrace{\Big|\frac{x_{i_2}\ophi(x_{i_2})-\xi_{i_1,i_2}\ophi(\xi_{i_1,i_2})}{x_{i_2}-x_{i_1}}\Big|\cdot | G(x) \cdot H_{i_2,i_2}H_{i_2,i_1}^2 |}_{:=(1)}+\underbrace{\Big|\frac{\ophi(x_{i_1})\ophi(x_{i_2})-\ophi(\xi_{i_1,i_2})^2}{x_{i_2}-x_{i_1}}\Big|\cdot | G(x) \cdot H_{i_2,i_2}H_{i_2,i_1}^2 |}_{:=(2)},
    \end{align*}
    where the first equality used the mean-value theorem to obtain a $\xi_{i_1,i_2}\in [x_{i_1},x_{i_2}]$, second equality used Eq.~\eqref{eq:derivativeofophi}.  We now bound both these terms separately as follows.

    \paragraph{Term 1 upper bound.}

    Note that $\xi_{i_1,i_2}$ is between $x_{i_1}$ and $x_{i_2}$. The first term is upper bounded~by
    \begin{eqnarray*}
      && \sum_{i_1\neq i_2} \Big|\frac{x_{i_2}\ophi(x_{i_2})-\xi_{i_1,i_2}\ophi(\xi_{i_1,i_2})}{x_{i_2}-\xi_{i_1,i_2}}\Big|\cdot | G(x) \cdot H_{i_2,i_2}H_{i_2,i_1}^2 |\\
      &=&\sum_{i_1\neq i_2}\abs{\br{1-\eta_{i_1,i_2}^2}\ophi\br{\eta_{i_1,i_2}}-\eta_{i_1,i_2}\ophi\br{\eta_{i_1,i_2}}^2}\cdot\abs{G(x) \cdot H_{i_2,i_2}H_{i_2,i_1}^2}\\
      &\leq&2\Delta^4\sum_{i_1\neq i_2}\ophi\br{\eta_{i_1,i_2}}\abs{G(x) \cdot H_{i_2,i_2}H_{i_2,i_1}^2}
    \end{eqnarray*}
    for some $\eta_{i_1,i_2}$ between $x_{i_2}$ and $\xi_{i_1,i_2}$, where we apply a mean value theorem for the function $x\ophi\br{x}$ for the equality and Lemma~\ref{lem:ophiupperbound} for the inequality.
    Note that $\ophi\br{\cdot}$ is nonnegative and monotone decreasing by Fact~\ref{fac:ophimonotone}. Thus the first term is upper bounded by
    \begin{eqnarray*}
      2\Delta^2\sum_{i_1\neq i_2}\max\set{\ophi\br{x_{i_1}},\ophi\br{x_{i_2}}}\abs{G(x) \cdot H_{i_2,i_2}H_{i_2,i_1}^2}
    \end{eqnarray*}
    which, in turn, is upper bounded by
    $O\br{\Delta^4\cdot\sqrt{\log k}\cdot\norm{H}^3}$ from Fact~\ref{fac:benktus2} and Eqs~\eqref{eqn:HH2},~\eqref{eqn:H2H2}.

    \paragraph{Term 2 upper bound.} By triangle inequality we upper bound the second term by
    \begin{align}
    \label{eq:casestarterm2}
    \begin{aligned}
    &\sum_{i_1\neq i_2}\Big|\frac{\ophi(x_{i_1})\ophi(x_{i_2})-\ophi(\xi_{i_1,i_2})^2}{x_{i_2}-x_{i_1}}\Big|\cdot | G(x) \cdot H_{i_2,i_2}H_{i_2,i_1}^2 |\\
    &\leq \sum_{i_1\neq i_2}\Big|\frac{\ophi(x_{i_1})\ophi(x_{i_2})-\ophi(x_{i_1})^2}{x_{i_2}-x_{i_1}}\Big|\cdot | G(x) \cdot H_{i_2,i_2}H_{i_2,i_1}^2 |+ \sum_{i_1\neq i_2}\Big|\frac{\ophi(x_{i_1})^2-\ophi(\xi_{i_1,i_2})^2}{x_{i_2}-x_{i_1}}\Big|\cdot | G(x) \cdot H_{i_2,i_2}H_{i_2,i_1}^2 |.
    \end{aligned}
    \end{align}
    We first upper bound the first quantity in Eq.~\eqref{eq:casestarterm2} as follows.
    \begin{align}
    \label{eq:case4term2}
    \begin{aligned}
     &\sum_{i_1\neq i_2}\Big|\frac{\ophi(x_{i_1})\ophi(x_{i_2})-\ophi(x_{i_1})^2}{x_{i_2}-x_{i_1}}\Big|\cdot | G(x) \cdot H_{i_2,i_2}H_{i_2,i_1}^2 |\\
    &=\sum_{i_1\neq i_2}\Big|\frac{\ophi(x_{i_2})-\ophi(x_{i_1})}{x_{i_2}-x_{i_1}}\Big|\cdot | G(x)|\cdot |\ophi(x_{i_1})| \cdot |H_{i_2,i_2}H_{i_2,i_1}^2 |\\
    &=\sum_{i_1\neq i_2}|\ophi'(\zeta_{i_1,i_2})|\cdot | G(x)|\cdot |\ophi(x_{i_1})| \cdot |H_{i_2,i_2}H_{i_2,i_1}^2 |\quad\quad\quad\mbox{}\\
    &\leq 6\Delta^2\cdot \sum_{i_1\neq i_2} | G(x)|\cdot |\ophi(x_{i_1})| \cdot |H_{i_2,i_2}H_{i_2,i_1}^2 | \quad\quad\quad\quad\quad\mbox{}\\
    &\leq6\Delta^2 \onenorm{G^{(1)}}\norm{H}^3.  \quad\quad\quad\quad\quad\mbox{}\\
    &\leq O\br{\Delta^2\cdot\sqrt{\log k}\cdot\norm{H}^3},
    \end{aligned}
    \end{align}
    where $\zeta_{i_1,i_2}$ between $x_{i_1}$ and $x_{i_2}$,  first inequality uses Fact~\ref{lem:ophiupperbound}, the second inequality uses Eqs.~\eqref{eqn:HH2},~\eqref{eqn:H2H2} and the last inequality is from Fact~\ref{fac:benktus2}.

    We now bound the second term in Eq.~\eqref{eq:casestarterm2} as follows

    \begin{align*}
    &\sum_{i_1\neq i_2}\Big|\frac{\ophi(x_{i_1})^2-\ophi(\xi_{i_1,i_2})^2}{x_{i_2}-x_{i_1}}\Big|\cdot | G(x) \cdot H_{i_2,i_2}H_{i_2,i_1}^2 |\\
    &\leq\sum_{i_1\neq i_2}\Big|\frac{\ophi(x_{i_1})^2-\ophi(\xi_{i_1,i_2})^2}{\xi_{i_1,i_2}-x_{i_1}}\Big|\cdot | G(x) \cdot H_{i_2,i_2}H_{i_2,i_1}^2 |\quad\quad\quad\quad\tag{for $\xi$ is between $x_{i_1}$ and $x_{i_2}$}\\
    &= 2\sum_{i_1\neq i_2}\abs{\ophi\br{\eta_{i_1,i_2}}\ophi'\br{\eta_{i_1,i_2}}}\cdot\abs{G\br{
    x}}\cdot\abs{H_{i_2,i_2}H_{i_2,i_1}^2}\quad\quad\quad\quad\mbox{(for some $\eta_{i_1,i_2}$ between $x_{i_1}$ and $\xi_{i_1,i_2}$)}\\
    &\leq 12\Delta^2\sum_{i_1\neq i_2}\abs{\ophi\br{\eta_{i_1,i_2}}}\cdot \abs{G\br{
    x}}\cdot\abs{H_{i_2,i_2}H_{i_2,i_1}^2}\quad\quad\quad\quad\quad\tag{Lemma~\ref{lem:ophiupperbound}}\\
    &\leq12\Delta^2\sum_{i_1\neq i_2}\max\set{\ophi\br{x_{i_1}},\ophi\br{x_{i_2}}}\cdot\abs{G\br{x}}\cdot\abs{H_{i_2,i_2}H_{i_1,i_2}^2}. \quad\quad\quad\quad\quad\tag{Fact~\ref{fac:ophimonotone}}
    \end{align*}
    Further applying Fact~\ref{fac:benktus2} and putting together  Eqs.~\eqref{eqn:H2H2}\eqref{eqn:HH2}, we conclude that it can be upper bounded by $O\br{\Delta^2\sqrt{\log k}\norm{H}^3}$.
    \end{proof}

    \subsubsection{Bounding term ($5$) in Theorem~\ref{thm:spectralthm}}
    \label{sec:b5}
    \begin{lemma}[Bounding terms $(5)$ in Theorem~\ref{thm:spectralthm}]\label{lem:bound5}
    We have
     \begin{eqnarray*}
        \abs{\sum_{i_1\neq i_2\neq i_3 }\frac{\nabla^2_{i_2,i_3}G\br{x}-\nabla^2_{i_1,i_3}G\br{x}}{x_{i_2}-x_{i_1}}H_{i_1,i_2}^2H_{i_3,i_3}}
    \leq O\br{\Delta\cdot\log k\cdot \norm{H}^3}.
      \end{eqnarray*}
    \end{lemma}

    \begin{proof}
    \begin{align*}
      &\abs{\sum_{i_1\neq i_2\neq i_3 }\frac{\nabla^2_{i_2,i_3}G\br{x}-\nabla^2_{i_1,i_3}G\br{x}}{x_{i_2}-x_{i_1}}H_{i_1,i_2}^2H_{i_3,i_3}}
    \\
    =&\abs{\sum_{i_1\neq i_2\neq i_3}\frac{\ophi\br{x_{i_3}}\br{\ophi\br{x_{i_2}}-\ophi\br{x_{i_1}}}}{x_{i_2}-x_{i_1}}G\br{x}H_{i_1,i_2}^2H_{i_3,i_3}}\\
    =&\abs{\sum_{i_1\neq i_2\neq i_3}\ophi'\br{\xi_{i_1,i_2}}\ophi\br{x_{i_3}}G\br{x}H_{i_1,i_2}^2H_{i_3,i_3}}\tag{for some $\xi_{i_1,i_2}$ between $x_{i_1}$ and $x_{i_2}$}\\
    \leq & 3\Delta\sum_{i_1\neq i_2\neq i_3}\abs{\max\set{\ophi\br{x_{i_1}},\ophi\br{x_{i_2}}}\ophi\br{x_{i_3}}G\br{x}H_{i_1,i_2}^2H_{i_3,i_3}}\quad\quad\quad\quad\quad\tag{Fact~\ref{fac:ophimonotone} and Lemma~\ref{lem:ophiupperbound}}\\
    \leq& O\br{\Delta\cdot\log k\cdot\norm{H}^3},
    \end{align*}
    where the last inequality is from Fact~\ref{fac:benktus2} and Eqs.~\eqref{eqn:HH2}\eqref{eqn:H2H2}.
    \end{proof}

    \subsection{Bounding terms $(6,7)$ in Theorem~\ref{thm:spectralthm} for Bentkus function}
    \label{sec:b6}

    Let $G:\reals^k\rightarrow\reals$ be the Bentkus function given in Definition~\ref{def:bentkus}. Recall that $G(x)=\prod_i g(x_i)$, where $g\br{x}=\frac{1}{\sqrt{2\pi}}\int_{-\infty}^{-x}e^{-t^2/2}dt$. Recall the notation $g'\br{x}=\frac{1}{\sqrt{2\pi}}e^{-x^2/2}$ and $\overline{g}\br{x}=g'(x)/g(x)$.
    The terms are restated here for convenience.
    \begin{lemma}[Bounding terms $(6,7)$ in Theorem~\ref{thm:spectralthm}]\label{lem:casev}
    \begin{equation}\label{eqn:casev}
    \abs{\sum_{i_1\neq i_2\neq i_3}\frac{\frac{\ophi\br{x_{i_1}}-\ophi\br{x_{i_3}}}{x_{i_3}-x_{i_1}}-\frac{\ophi\br{x_{i_1}}-\ophi\br{x_{i_2}}}{x_{i_2}-x_{i_1}}}{x_{i_3}-x_{i_2}}G\br{x}H_{i_1,i_2}H_{i_2,i_3}H_{i_3,i_1}}\leq O\br{\Delta\log^2k\norm{H}^3}
    \end{equation}
    \end{lemma}
     This is the most involved part. Note that the left hand side is unchanged if we zero out all diagonal entries of $H$. And further note that $\norm{H-\diag\br{H}}\leq 2\norm{H}$ where $\diag\br{H}$ is a diagonal matrix obtained by diagonalizing $H$. Thus, we may assume that the diagonal elements in $H$ are zeros without loss of generality. We break down the analysis into two cases (the first one being the simpler case).

    \subsubsection{Case 1: Many negative $x_i$s. }  The simpler case is when the number of negative $x_i$s is ``large".
    \begin{lemma}\label{lem:negative}
      If $\abs{\set{i:x_i<0}}>3\log k$, then the quantity in Eq.~\eqref{eqn:casev} is upper bounded by $O\br{\Delta^2\frac{\sqrt{\log k}}{k}\cdot\norm{H}^3}$.
    \end{lemma}
    \begin{proof}
      Applying Fact~\ref{fac:mvtdd} a mean value theorem of divided difference and Lemma~\ref{lem:ophiupperbound}, the term in Eq.~\eqref{eqn:casev} is upper bounded by
      \begin{eqnarray*}
        &&O\br{\Delta^2\sum_{i_1\neq i_2\neq i_3} \ophi(\zeta_{i_1,i_2,i_3})}G\br{x}\abs{H_{i_1,i_2}H_{i_2,i_3}H_{i_3,i_1}} \\
        &\leq &O\br{\Delta^2\sum_{i_1\neq i_2\neq i_3}\max\set{\ophi\br{x_{i_1}},\ophi\br{x_{i_2}},\ophi\br{x_{i_3}}}}G\br{x}\abs{H_{i_1,i_2}H_{i_2,i_3}H_{i_3,i_1}} \\
        &\leq& O\br{\Delta^2\onenorm{G^{(1)}}\max_{i_1}\sum_{i_2,i_3}\abs{H_{i_1,i_2}H_{i_2,i_3}H_{i_3,i_1}}}\\
        &\leq& O\br{\Delta^2\cdot k\cdot \onenorm{G^{(1)}\br{x}}\cdot\max_{i_1,i_2}\sum_{i_3}\abs{H_{i_1,i_2}H_{i_2,i_3}H_{i_3,i_1}}}\\
        &\leq&O\br{\Delta^2\cdot k\cdot\onenorm{G^{(1)}\br{x}}\cdot\norm{H}^3}\\
        &\leq&O\br{\frac{\Delta^2\sqrt{\log k}}{k}\cdot\norm{H}^3}
      \end{eqnarray*}
      where the first inequality is from the positivity and monotonicity of $\ophi\br{\cdot}$ due to Fact~\ref{fac:ophimonotone}  to conclude that $|\ophi(\zeta_{i_1,i_2,i_3})|\leq \max\{|\ophi(x_{i_1})|,|\ophi(x_{i_2})|,|\ophi(x_{i_3})|\}$; the second last inequality is from the following~fact
       \begin{equation}\label{eqn:h3}
      \sum_{i_3}\abs{H_{i_1,i_2}H_{i_2,i_3}H_{i_3,i_1}}\leq\norm{H}\sqrt{\br{\sum_{i_3}H_{i_1,i_3}^2}\br{\sum_{i_3}H_{i_2,i_3}^2}}\leq\norm{H}^3;
    \end{equation}
    the last inequality is from Lemma~\ref{lem:normlowvalue} (which uses that the number of negative $x_i$s is $\leq 3\log k$).
    \end{proof}

    \subsubsection{Case 2: A few negative $x_i$s}
      We now assume that $\abs{\set{i:x_i<0}}\leq 3\log k$  and this case the most complicated and upper bounding it is the most technical. We push this proof to Appendix~\ref{app:claim5proof}.

    \begin{proof}[Proof of Theorem~\ref{thm:derivative}]
    Combining Theorem~\ref{thm:spectralthm} and Lemmas~\ref{lem:123},~\ref{lem:term4},~\ref{lem:bound5},~\ref{lem:casev}, we obtain our result.
    \end{proof}

    \section{Properties of positive spectrahedra}
    \label{sec:geometric}
    \subsection{Average sensitivity and Noise sensitivity}
    In this section we prove certain combinatorial properties of positive spectrahedra. Understanding the average sensitivity and noise sensitivity is a fundamental question in Boolean analysis and learning theory. Proving bounds on these quantities for geometric objects has also received a lot of attention. In this direction,~\cite{harsha2013invariance,kane2014average} proved upper bounds on average sensitivity of halfspaces and the well-known Peres' theorem~\cite{peres2004noise} bounds the noise-sensitivity of halfspaces. In this section, we show analogous bounds to these papers also hold for positive spectrahedra.

    \begin{theorem}[Matrix version of Peres theorem]
    \label{thm:GSAspectrahedra}
    Let $S$ be a {positive} spectrahedron defined as
    $$
    S=\big\{x\in \R^n: \sum_i x_i A^i\preceq B,\hspace{2mm} A^1,\ldots,A^n,B\in \sym{k} \text{ and } A^i \text{ is } \PSD \text{ for } i\in [n]\big\}.
    $$
    Let $f(x)=[x\in S]$ for $x\in \pmset{n}$. Then the $\varepsilon$-noise sensitivity of $f$ is $\NS_\varepsilon(f)=O(\sqrt{\varepsilon})$.
    \end{theorem}

    \begin{theorem}
    \label{cor:intersectionofpspec}
    Let $S^1,S^2$ be $2$ distinct positive spectrahedra specified by $\{A^1_j,\ldots,A^n_j,B_j\}_{j\in [2]}$ respectively, where $A^i_1\succeq 0$ and $A^i_2\preceq 0$ for all $i$. Let
    $$
    F(x)=\bigwedge_{j=1}^2 \Br{\sum_i x_i A^i_j \preceq B_j}
    $$
    be an intersection of positive spectrahedra. Then
    $
    \AS(F)\leq O(\sqrt{n}).$
    \end{theorem}

    The proof of Theorem~\ref{thm:GSAspectrahedra} follows closely the proof of Kane~\cite{kane2014average} who showed that $k$-facet polytopes have $\varepsilon$-noise sensitivity at most $O(\sqrt{\varepsilon\log k})$. Before stating the Kane's result, we need to introduce the following notion.

    \begin{definition}[Unate function]
    \label{def:unate}
    A function $f:\set{-1,1}^n\rightarrow \01$ is \emph{unate} if it satisfies the following: for every $i\in [n]$, $f$ is either increasing or decreasing with respect to the $i$th coordinate, i.e., for every $i\in [n]$, either $f(x_1,\ldots,x_{i-1},-1,x_{i+1},\ldots,x_n)\leq f(x_1,\ldots,x_{i-1},1,x_{i+1},\ldots,x_n)$ for all~$x$ or $f(x_1,\ldots,x_{i-1},-1,x_{i+1},\ldots,x_n)\geq f(x_1,\ldots,x_{i-1},1,x_{i+1},\ldots,x_n)$ for all~$x$.
    \end{definition}

    In particular, Kane proved the following stronger statement.
    \begin{theorem}\cite{kane2014average}
    \label{thm:kane}
    Let $f_1,\ldots,f_k:\pmset{n}\rightarrow \01$ be unate functions and let $F:\pmset{n}\rightarrow \01$ be defined as $F(x)=\bigwedge_i f_i(x)$. Then the average sensitivity of $F$ satisfies $\AS(F)\leq O(\sqrt{n\log (k+1)})$.\footnote{There is a $+1$ compared to Kane's result to ensure that the result is valid for $k=1$.}
    \end{theorem}

    It is not hard to see that a positive spectrahedron is a unate function so Theorem~\ref{thm:kane} holds for us as well for $k=1$. Hence we have the following corollary. 
    \begin{corollary}
    \label{cor:asforspec}
    Let $S$ be as defined in Theorem~\ref{thm:GSAspectrahedra}. Let $F:\pmset{n}\rightarrow \01$ be defined as $F(x)=1$ if and only if $x\in S$. Then $\AS(F)\leq O(\sqrt{n})$.
    \end{corollary}

Recall that we are interested in noise sensitivity of $F$. In the same paper, Kane~\cite{kane2014average} adapts the well-known techniques of~\cite{diakonikolas2010bounded} to show that the $\varepsilon$-{noise sensitivity} of \emph{intersections} of halfspaces is at most $O(\sqrt{\varepsilon \log k})$ and remarks that  such a bound \emph{does not} hold for the intersections of unate functions. Below, we show that one can modify the proof of~\cite{diakonikolas2010bounded} to also show that the noise sensitivity of positive  spectrahedra  can be bounded by the ``average 2-sensitivity" of positive spectrahedra which we show is $O(\sqrt{\varepsilon})$ by modifying Kane's proof in Theorem~\ref{thm:kane}. This proves  Theorem~\ref{thm:GSAspectrahedra}.

    \begin{proof}[Proof of Theorem~\ref{thm:GSAspectrahedra}]
    In order to prove the theorem, we first show  that for a function  $f:\pmset{n}\rightarrow \01$ defined as
    $$
    f(x)=\Br{\sum_{i=1}^n x_i A^i\preceq B}
    $$
    for $A^1,\ldots,A^n,B\in \sym{k}$  and $A^i$ is $\PSD$ for $i\in [n]$, the $\varepsilon$-noise sensitivity of $f$ satisfies
    $$
    \NS_\varepsilon(f)=\Pr_{\substack{(\bx,\by) \\ \varepsilon-\text{correlated}}}[f(\bx)\neq f(\by)]\leq O(\sqrt{\varepsilon}).
    $$
    For simplicity let us assume that $\varepsilon=1/m$, for some integer $m$ which divides $n$ (since $\NS_\varepsilon$ is a non-decreasing function in $\varepsilon$, we can even round $\varepsilon$ down to satisfy this condition).

    In order to analyze $\NS_\varepsilon(f)$ we first observe that one can generate an $\varepsilon$-correlated pair of strings $(x,y)\in \pmset{n}$ as follows\footnote{We deviate from~\cite{diakonikolas2010bounded} in this process of generating correlated strings. The reason for this modification is, we require that within every bucket, the induced spectrahedron has to be either $\PSD$ or $\NSD$, which isn't guaranteed in the original bucketting procedure of~\cite{diakonikolas2010bounded,kane2014average}}:
    \begin{enumerate}
        \item Pick a uniformly random string $\bz\sim\U_n$.

        \item Randomly partition $[n]$ into $m$ disjoint buckets $C_1,\ldots,C_m\subseteq [n]$ such that $\cup_i C_i=[n]$. Furthermore, for $z\in\pmset{n}$  (picked in step $1$),  split each bucket as follows: for every $\ell\in [m]$, split $C_\ell$ into $C_{\ell,1}$ and $C_{\ell,-1}$ such that $C_{\ell,1}$ corresponds to the \emph{positive} coordinates in $z_{C_\ell}$ and $C_{\ell,-1}$ corresponds to the \emph{negative} coordinates in $z_{C_\ell}$. So overall there are $2m$ disjoint buckets $\{C_{\ell,s}:\ell\in [m],s\in \pmset{}\}$ such that $\cup_{\ell,s} C_{\ell,s}=[n]$. Set $\tilde{C}_{\ell}=C_{\ell,1}$ if $\ell\leq m$ and $\tilde{C}_{\ell}=C_{\ell-m,-1}$ if $\ell>m$.

        \item Corresponding to each bucket $\tilde{C}_{\ell}$, pick a uniformly random bit $\mathbf{b}_{\ell}\sim\U_1$.

        \item We obtain $\bx$ as follows: for every $\ell\in [2m]$, obtain $\bx$ from $\bz$ by multiplying all the bits in $\bz_{\tilde{C}_{\ell}}$ by~$\mathbf{b}_{\ell}$.

        \item We obtain $\by$ as follows: pick a uniformly random $\ell\in [m]$
        and flip the signs of $\bx_i$ (obtained in step $4$) for all the indices $i$ in $C_{\ell}$, i.e., $\by_i=-\bx_i$ if $i\in C_{\ell}$ and $\by_i=\bx_i$~otherwise.
    \end{enumerate}

    Observe that the the $(\bx,\by)$ obtained in step $(4,5)$ are uniform and $\varepsilon$-correlated. To see this, first observe that the probability of obtaining $x\in \pmset{n}$ is given by
    \begin{align*}
    \Pr_{\substack{\bz\sim\U_n,\\\{C_k\},\bb\sim\U_{2m}}}\Br{\bx=x}&=\Pr_{\bz,C,\bb}\Br{\bz_{\tilde{C}_1}\cdot \bb_1=x_{\tilde{C}_1},\ldots \bz_{\tilde{C}_{2m}}\cdot \bb_{2m}=x_{\tilde{C}_{2m}}}\\
    &=\prod_{i=1}^{2m}\Pr_{\bz,C,\bb}\Br{\bz_{\tilde{C}_i}\cdot \bb_i=x_{\tilde{C}_i}\vert  \bz_{\tilde{C}_{<i}}\cdot \bb_{<i}=x_{\tilde{C}_{<i}}}\\
    &=\prod_{i=1}^{2m}\Pr_{\bz,C,\bb}\Br{\bz_{\tilde{C}_i}\cdot \bb_i=x_{\tilde{C}_i}}=\prod_{i=1}^{2m}\frac{1}{2^{|\tilde{C}_i|}}=\frac{1}{2^n},
    \end{align*}
    where the third equality is because $\bz,\bb$ are uniformly random and final equality is because $\cup_k \tilde{C}_k=~[n]$ and $\set{\tilde{C}_i}_{1\leq i\leq 2m}$ are disjoint. In order to see $(\bx,\by)$ are $\varepsilon$-correlated, observe that for a fixed $i\in [n]$ the probability $\bx_i$ differs from $\by_i$ is exactly the probability $i$ lies in the bucket $C_\ell=C_{\ell,1}\cup C_{\ell,-1}$ picked in Step (5) above. The probability of picking a bucket $C_\ell$ is exactly $1/m=\varepsilon$.  This event happens independently over all the coordinates $i\in [m]$, hence $y$ is $\varepsilon$-correlated with $x$.

    Now that we have shown $(\bx,\by)$ are $\varepsilon$-correlated, we next observe that for a fixed $z$ and buckets  $\tilde{C}_1,\ldots,\tilde{C}_{2m}$, we can write $f:\pmset{n}\rightarrow \01$ as a function $g:\pmset{2m}\rightarrow \01$ defined~as

    \begin{equation}\label{eqn:gb}
    g(b)=\Br{\sum_{q=1}^{2m} b_q \sum_{j\in \tilde{C}_{q}}z_j A^j\preceq B}.
    \end{equation}
Similarly, one can define $f(y)$ as $g(b')$ where $b'$ is obtained from $b$ by picking a uniformly random $\pmb{\ell}\sim [m]$ and flipping
    $b_{\pmb{\ell}},b_{\pmb{\ell}+m}$ where $C_{\pmb{\ell}}$ is the bucket chosen in Step $(5)$.
    Furthermore, observe that
    $$
    \NS_\varepsilon(f)=\Pr_{\substack{(\bx,\by) \\ \varepsilon-\text{correlated}}}[f(\bx)\neq f(\by)]=\Pr_{\substack{\bb\sim\U_{2m},\\\pmb{\ell}\sim[m]  }}[g(\bb)\neq g(\bb^{\pmb{\ell},\pmb{\ell}+m})],
    $$
    where $\bb^{i,j}$ is obtained by flipping the $i,j$th coordinates in $\bb$ and $\pmb{\ell}$ are chosen uniformly random in~$[m]$. We can further upper bound the quantity above by
    \begin{align}
    \label{eq:upperboundsingnsbyas}
    \begin{aligned}
    \NS_\varepsilon(f)&=\Pr_{\substack{\bb\sim\U_{2m},\\\pmb{\ell}\sim[m]  }}[g(\bb)\neq g(\bb^{\pmb{\ell},\pmb{\ell}+m})]\\
    &\leq \Pr_{\substack{\bb\sim\U_{2m},\\\pmb{\ell}\sim[m] }}[g(\bb)\neq g(\bb^{\pmb{\ell}})]+\Pr_{\substack{\bb\sim\U_{2m},\\\pmb{\ell}\sim[m]  }}[g(\bb^{\pmb{\ell}})\neq g(\bb^{\pmb{\ell},\pmb{\ell}+m})]\\
    &=\Pr_{\substack{\bb\sim\U_{2m},\\\pmb{\ell}\sim[m]  }}[g(\bb)\neq g(\bb^{\pmb{\ell}})]+\Pr_{\substack{\bb\sim\U_{2m},\\\pmb{\ell}\sim[m]  }}[g(\bb)\neq g(\bb^{\pmb{\ell}+m})]=\frac{1}{m}\AS(g),
    \end{aligned}
    \end{align}
    where the second equality used the fact that $\bb,\pmb{\ell}$ are uniform over their respective domains and the last equality used the definition of $\AS(g)$ to obtain
    $$
    \AS(g)=\sum_{\ell=1}^{2m}\Pr[g(\bx)\neq g(\bx^\ell)]=\sum_{\ell=1 }^m\Pr[g(\bx)\neq g(\bx^\ell)]+\sum_{\ell=m+1}^{2m}\Pr[g(\bx)\neq g(\bx^\ell)].
    $$
     We now finally upper bound the average sensitivity of $g$. Observe that $\sum_{j\in \tilde{C}_q} \bz_j A^j$ is either $\PSD$ or $\NSD$ (since all the $\bz_j$ in the bucket $\tilde{C}_q$ have the same sign and $A^j$s are all $\PSD$ by definition). From Eq.~\eqref{eqn:gb}, observe that $g$ is a unate function. Hence, we have
    \begin{align}
        \label{eq:nsatmostlogk}
        \NS_\varepsilon(f)\leq \frac{1}{m}\AS(g)\leq O\br{\sqrt{\frac{1}{m}}}=O(\sqrt{\varepsilon}),
    \end{align}
    where the first inequality is by Eq.~\eqref{eq:upperboundsingnsbyas}, second inequality uses Theorem~\ref{thm:kane} and the last equality used the definition of $m=1/\varepsilon$. This concludes the proof of the theorem.
    \end{proof}

    We now prove Theorem~\ref{cor:intersectionofpspec} which bounds the average sensitivity of intersections of positive spectrahedra.
    \begin{proof}[Proof of Theorem~\ref{cor:intersectionofpspec}]
    The proof is very similar to the proof of the theorem above. Let $m=\lceil 1/\varepsilon\rceil$. We follow the same bucketing steps $(1)-(5)$ in Theorem~\ref{thm:GSAspectrahedra} to obtain a $g:\pmset{2m}\rightarrow \01$ given by
    $$
    g(b)=\Br{\sum_{q=1}^{2m} b_q \sum_{j\in C_{q}}z_1 A^j_1\preceq B_1}\cdot \Br{\sum_{q=1}^{2m} b_q \sum_{j\in C_{q}}z_1 A^j_2\preceq B_2}.
    $$
    Observe that $g$ is an intersection of positive spectrahedra and by definition each  positive spectrahedron is a unate function. So, by Theorem~\ref{thm:kane}, we have
    $$
    \AS(g)\leq \sqrt{1/\varepsilon}=O(\sqrt{m}).
    $$
    This concludes the proof of the corollary.
    \end{proof}

    \subsection{Boolean Anti-concentration: Littlewood Offord for spectrahedra}
    We now prove the main lemma  which shows that the largest eigenvalues of positive spectrahedra cannot be very concentrated. In particular, we show that for a uniformly random $\bx\sim\U_n$, the probability that the random matrix $D=\sum_i \bx_i A^i-B$ has the largest eigenvalue in a small interval is fairly small. This anti-concentration statement will be crucial in our invariance principle proof when we move from the Bentkus mollifier to our CDF function. In the passing we remark that, prior to this work, we aren't even aware if the weaker \emph{Gaussian} analogue of this statement was known (in particular, the results of~\cite{harsha2013invariance,servedio2017fooling} only require Gaussian anti-concentration for which they use a result of Nazarov~\cite{nazarov2003maximal} as a black-box).

    In order to prove our main theorem (stated below), we follow the result of~\cite{o2019fooling,kane2014average} closely since they are able to handle intersections of unate functions which is the case for positive spectrahedra. However, there are two subtleties.

    \begin{enumerate}[($i$)]
        \item In~\cite{o2019fooling} they bucket the set of halfspaces (which form the polytope) and show that each bucket has significant weight. Crucially for them, they use the fact that intersections of halfspaces are still unate functions. But this is not the case for positive spectrahedra. For this, we need to modify the bucketing procedure (akin to what happens in the proof of Theorem~\ref{thm:GSAspectrahedra}) so that this bucketing of positive spectrahedra still results in a unate function.
        \item  In~\cite{o2019fooling} they prove an analogue of Lemma~\ref{lem:buckettingspectra} which shows that each bucket has ``significant weight". However our proof deviates significantly from the proof in~\cite{o2019fooling}. For them, proving the statement in the lemma (for diagonal matrices), follows directly from Paley-Zygmund inequality, but as far as we are aware, we do not have a matrix-version of this inequality. Due to this difficulty, we modify their proof and use the matrix Chernoff bound to prove the statement above.
    \end{enumerate}

    \begin{theorem}
    \label{thm:gauanticonc1spec}
    Let $k\geq 0$ be an integer and $\tau\leq\frac{1}{100\sqrt{\log k}}$. Let $\{B_1,B_2\}\subseteq \sym{k}$, $\{A^i_1\}_{i\in [n]}$ and $\{A^i_2\}_{i\in [n]}$ be sequences of $\PSD$ and $\NSD$ matrices, respectively. They satisfy that for all $i\in [n],j\in [2]$,
    $A^i_1\preceq \tau\cdot \id, A^i_2\succeq-\tau\id$ and $ \sum_{i=1}^n(A^i_j)^2\succeq\id$. Then   for every $\Lambda\geq 20\tau\log k$, we have
    $$
    \Pr_{\bx\sim\U_n}\Br{\exists j\in [2] \text{ s.t. } \lambda_{\max}\br{ \sum_i \bx_i A^i_j-B_j} \in (-\Lambda,\Lambda]}\leq O(\Lambda).
    $$
    \end{theorem}

    Using the standard bits-to-Gaussians trick~\cite[Chapter 11]{o2014analysis}, we have the following corollary.

    \begin{corollary}
    \label{cor:gauanticonc1spec}
    Let $k\geq 0$ be an integer and $\tau\leq\frac{1}{\log k}$. Let $\{B_1,B_2\}\subseteq \sym{k}$, $\{A^i_1\}_{i\in [n]}$ and $\{A^i_2\}_{i\in [n]}$ be sequences of $\PSD$ and $\NSD$ matrices, respectively. They satisfy that for all $i\in [n],j\in [2]$,
    $A^i_1\preceq \tau\cdot \id, A^i_2\succeq-\tau\id$ and $ \sum_i(A^i_j)^2\succeq\id$ . Then for every $\Lambda\geq 20\tau\log k$, we have
    $$
    \Pr_{\bg\sim\G^n}\Br{\exists j\in [2] \text{ s.t. } \lambda_{\max}\br{ \sum_i \bg_i A^i_j-B_j} \in (-\Lambda,\Lambda]}\leq O(\Lambda).
    $$
    \end{corollary}

    In order to prove this theorem we will use the following two lemmas by~\cite{o2019fooling}. Before stating these lemmas, we introduce a few definitions from~\cite{o2019fooling} (adapted to our setting of positive spectrahedra). For the rest of the section, we let $F:\pmset{n}\rightarrow \01$ be the indicator of an intersection of positive spectrahedra, i.e., for every $j\in [2]$, let $F_j(x)=\Br{\sum_{i=1}^n x_i A_j^i\preceq B_j}$, where $\set{A^i_j}_{j\in\set{1,2}}$ are sequences of $\PSD$($\NSD$) matrices and
    \begin{align}
    \label{eq:intersectionofpsd}
        F(x)=\bigwedge_{j=1}^2 F_j(x)=\bigwedge_{j=1}^2 \Br{\sum_{i=1}^n x_i A_j^i\preceq B_j}.
    \end{align}

    \begin{enumerate}
        \item For a set $S\subseteq \pmset{n}$, let $\calE(S)$ be the fraction of $n\cdot 2^{n-1}$ edges which have one endpoint in $S$ and one endpoint in $S^c$ (i.e., complement of $S$).
        \item  We let $H_j\subseteq \pmset{n}$ be the \emph{indicator-set} for $F_j$, i.e., $x\in H_j$ if and only if $F_j(x)=1$. Additionally, suppose we have sets $\{\bar{H}_1,\bar{H}_2\}$ such that $H_j\subseteq \bar{H}_j$ such that $\bar{H}_j$ are also the indicator-sets of  unate functions. Let $\partial H_j=\bar{H}_j\backslash H_j$.
        \item For $\alpha\in [0,1]$, we say $\partial H_j$ is \emph{$\alpha$-semi thin} if for every $x\in H_j$, at least an $\alpha$-fraction of its hypercube-neighbours (i.e., set of $y\in \pmset{n}$ for which $d(x,y)=1$) are outside $\partial H_j$.
        \item We now define a few sets: let
        $$
        F=\bar{H}_1\cap\bar{H}_2, \qquad F^\circ={H}_1\cap {H}_2, \qquad \partial F=F\backslash F^\circ
        $$
    \end{enumerate}
    With this terminology, we have the following lemma that bounds the number of edges that cross~$F$.
    \begin{lemma}[{\cite[Theorem~7.18]{o2019fooling}}]
    \label{lem:semithinvol}
    For $j\in [2]$, let $H_j$ be as defined above. Suppose $H_j$ is $\alpha$-semi thin, then
    $$
    \vol(\partial F)\leq O\br{\frac{1}{\alpha\sqrt{n}}}
    $$
    \end{lemma}
    Using this lemma, we get the following theorem (which is the analogue of~\cite[Theorem~7.19]{o2019fooling}).
    \begin{theorem}
    \label{thm:alphamanylittlewood}
    Let $\lambda>0, \alpha\in [0,1], \{B_1, B_2\} \subseteq \sym{k}$.  Let $\{A^i_j\}_{i\in [n],j\in [2]}\subseteq\sym{k}$ satisfy that $A^i_1\succeq 0, A^i_2\preceq 0$ for all $i\in[n]$. At least $\alpha$-fraction of $i\in[n]$ satisfy that $A^i_1\succeq\lambda \cdot \id$ and $A^i_2\preceq-\lambda \cdot \id$. Then, we have  
    $$
    \Pr_{\bx\sim\U_n}\Br{\exists j\in [2] \text{ s.t. } \lambda_{\max}\br{ \sum_i \bx_i A^i_j-B_j} \in (-2\lambda,0]}\leq O\br{\frac{1}{\alpha \sqrt{n}}}.
    $$
    \end{theorem}

    \begin{proof}
    Let $\{A^i_j\},\{B_j\}$ be as in the theorem statement. Let
    $$
    H_j=\Big\{x\in \pmset{n}: \lambda_{\max}\br{\sum_i x_i A^i_j - B_j}\leq -2\lambda\Big\},\quad \bar{H}_j=\Big\{x\in \pmset{n}:\lambda_{\max}\br{\sum_i x_i A^i_j-B_j} \leq 0\Big\}.
    $$
    Clearly we then have that
    $$
    \partial H_j=\set{x\in \pmset{n}:\lambda_{\max}\br{ \sum_i x_i A^i_j-B_j} \in (-2\lambda,0]}
    $$
    and
    $$
    \partial F= \set{x\in \pmset{n}:\exists j\in [2] \text{ s.t. } \lambda_{\max}\br{ \sum_i x_i A^i_j-B_j} \in (-2\lambda,0]}.
    $$
    Since we assumed that at least an $\alpha$-fraction of $i$s satisfied $A^i_1\succeq \lambda\cdot \id$ and $A^i_2\preceq-\lambda\cdot\id$, it follows that $H_j$ is $\alpha$-semi thin, hence we can apply Lemma~\ref{lem:semithinvol} to obtain the theorem statement.
    \end{proof}

    Using this theorem, we are now ready to prove our main technical lemma which says that we can always ``randomly bucket" our positive spectrahedron so that many of these buckets have ``pretty large" smallest eigenvalue.

    \begin{lemma}
    \label{lem:buckettingspectra}
     Let $\{A^i\}_{i\in [n]}\subseteq\sym{k}$ be a sequence of positive semidefinite matrices which is $\br{\tau, M}$-regular with $\tau\leq \frac{1}{100\sqrt{\log k}} $. Let $m\geq \frac{1}{10\tau^2\log k}$ and $\pi:[n]\rightarrow [m]$ be a random hash function that independently assigns each $i\in [n]$ to a uniformly random bucket  in $[m]$. For $c\in [m]$, let
    $$
    \sigma_c=\sum_{j\in \pi^{-1}(c)}A^j
    $$
    and we say the bucket $c\in C$ is \emph{good} if $\sigma_c\succeq\frac{1}{2\tau m}\cdot \id$. Then,
    $$
    \Pr\Br{ \text{at most }  3m/4 \text{ buckets } c\in [m]  \text{ are good }}\leq \exp\br{- m/4}.
    $$
    \end{lemma}
    \begin{proof}
    Let $\bz_i\in \01$ be a random variable satisfying $\Pr[\bz_i=1]=1/m$. Let $Z_i=\bz_i\cdot A^i$, hence one can write $\sigma_c=\sum_i Z_i$. In particular, this implies
    $$
    \Exp\Br{\sigma_c}=\frac{1}{m} \sum_i A^i\succeq\frac{1}{\tau\cdot m}\sum_i \br{A^i}^2\succeq\frac{1}{\tau m},
    $$
    where we used $A^i \preceq \tau\cdot \id$.
    Applying Fact~\ref{fac:matrixchernoff} (for $\delta=1/2$, $\mu=1/\tau m$, $R=\tau$) we have
    \[
    \prob{\sum_i Z_i\succeq\frac{1}{2\tau m}\id}\geq 1-k\cdot\br{\frac{2}{e}}^{\frac{1}{2\tau^2 m}}\geq\frac{9}{10}
    \]
    For $j\in [n]$ and $c\in [m]$ define random variables
    $$
    Y_{c,j}=\begin{cases}1~\mbox{if $\pi(j)=c$}\\0~\mbox{otherwise},\end{cases} \quad \text{ and }\qquad X_j=\Br{\sum_{c=1}^m Y_{c,j}\sigma_c\succeq\frac{1}{2\tau m}\id}.
    $$
    Using the Claim~\ref{claim:Xj} below, $X_1,\ldots, X_n$ are negatively associated. Thus we may apply the Chernoff bound to $\sum_{i=1}^m X_i$ which has mean at least $3m/4$, which gives us the lemma statement.

    \begin{claim}\label{claim:Xj}
    The random variables $X_1,\ldots, X_n$ are  negatively associated.
    \end{claim}
    \begin{proof}
    From ~\cite[Page 35, Example 3.1]{dubhashi_panconesi_2009}, the set of random variables $\set{Y_{c,j}}_{1\leq c\leq m}$ are negatively  associated for $j\in [n]$. Note that $\set{Y_{1,j},\ldots, Y_{m,j}}_{j\in [n]}$ are $n$ independent families of random variables. By~\cite[Page 35]{dubhashi_panconesi_2009}, $\set{Y_{c,j}}_{c\in [m],j\in [n]}$ are negatively associated. Given $\sigma_1,\ldots, \sigma_m$, $\Br{\sum_{c=1}^m Y_{c,j}\sigma_c\succeq\frac{1}{2\tau m}\id }$ is a monotone non-decreasing function of $Y_{c,1},\ldots, Y_{c,n}$. Thus from~\cite[Page 35]{dubhashi_panconesi_2009}, $X_1,\ldots, X_m$ are negatively associated.
    \end{proof}
    The proof of this claim concludes the proof of the lemma.
    \end{proof}

    We are now ready to proof our main theorem.

    \begin{proof}[Proof of Theorem~\ref{thm:gauanticonc1spec}]
    For $j\in[2]$, let $f_j(x)=\sum_{i=1}^nx_iA^i_j$. Let $\pi:[n]\rightarrow [2m]$ be a random hash function that independently assigns each $i\in [n]$ to uniformly random bucket in $[2m]$. Let  $C_1,\ldots,C_{2m}\subseteq [n]$ be the buckets and $z\in \pmset{2m}$ be uniformly random. Consider the function $g_j:\pmset{2k}\rightarrow \sym{k}$ defined as
    $$
    g_j(z)=\sum_{q=1}^{2m}z_q \cdot \sum_{i\in C_q}A^i_j.
    $$
    For $q\in [2m]$, define $\bar{A}^q_j=\sum_{i\in C_q}A^i_j$, so $g_j(z)=\sum_q z_q \bar{A}^q_j$. Observe that distribution of $f_j$ and $g_j$ are the same, i.e., for every $D\in \sym{k}$ we have
    \begin{align}
    \label{eq:distributionafterbuckeetingsame}
    \Pr_{\bz\sim\U_{2m},\{C_i\}}[g_j(\bz)=D]=\Pr_{\bx\sim\U_n}[f_j(\bx)=D].
    \end{align}
    In order to see this we argue that the $n$-bit string $w\in \pmset{n}$ defined as $w_i=z_q$ iff $i\in C_q$, is uniformly random. To show this, we first prove the following: for $z\in \pmset{2m}$, let $S=\{q\in [2m]: z_q=1\}$ and $T=\cup_{q\in S}C_q$. Then, observe that  for every $T\subseteq [n]$, we have $\Pr_{\bz,\{C_q\}}[\textbf{T}=T]=2^{-n}$ (for every $i\in [n]$, the probability of $i\in C_q$ is $1/(2m)$ and the probability $C_q$  is included in $T$ is $1/2$ since $z_q$ is a uniformly random bit, hence for every $i\in [n]$, we have $\Pr_{\bz,\{C_q\}}[i\in T]=\sum_{i=1}^{2m}(1/2m)\cdot (1/2)=1/2$ and this is independent for every $i\in [n]$ by construction). It is now easy to see that $w$ is uniformly random because
    $$
    \Pr_{\bz,\{C_j\}}[W=w]=\sum_{T}\Pr[\textbf{T}=T]\cdot \Pr[W=w\vert \textbf{T}=T]=\frac{1}{2^n}\sum_{T}\Pr[W=w\vert \textbf{T}=T]=2^{-n},
    $$
    where the last equality used the fact that once we fix $T$, then all the bits of $w$ which are $1$ are fixed.

    For $m=\frac{1}{20\tau^2\log k}$, let $\pi:[n] \rightarrow [2m]$ be a random hash that buckets these $n$ variables (jointly for $j\in [2]$).     By Lemma~\ref{lem:buckettingspectra}, we argued that, with probability at least $1-e^{-m/2}$, at least $9m/5$ of the $2m$ buckets are good for $j=1$, i.e., a good bucket $q\in [2m]$ for $j=1$ satisfies $\sum_{i\in \pi^{-1}(q)}A^i_1\succeq \frac{1}{4\tau m}\cdot \id$. For the same reason, with probability at least $1-e^{-m/2}$, at least $9m/5$ of the $2m$ buckets are good for $j=2$, i.e., a good bucket $q\in [2m]$ for $j=2$ satisfies $\sum_{i\in \pi^{-1}(q)}A^i_2\preceq-\frac{1}{4\tau m}\cdot \id$. Applying a union bound, at least $8m/5$ of $2m$ buckets are good for every $j\in [2]$  with probability at least $1-2\cdot e^{-m/2}$.

     By the argument in the start of the proof, we know that after bucketing, we can convert each $f_j$ into a function $g_j:\pmset{2m}\rightarrow \sym{k}$ such that $f_j$ and $g_j$ have the same distribution. Now we can invoke Theorem~\ref{thm:alphamanylittlewood} as follows: we know that a $4/5$-fraction of $q\in[2m]$ satisfy $\bar{A}^q_1\succeq \frac{1}{4\tau m}\cdot \id$ and $\bar{A}^q_2\preceq-\frac{1}{4\tau m}\cdot \id$, so we have
    \begin{align*}
      \Pr_{\bz\sim\U_{m}}\Br{\exists j\in[2]~\text{s.t.}~\lambda_{\max}\br{\sum_{q=1}^m \bz_q \bar{A}^q_j-B_j}\in(-1/2\tau m,0]}\leq O\br{\sqrt{\frac{1}{m}}}+2e^{-m/2}.
    \end{align*}
    We now prove the main theorem statement. In order to do so, first observe that, we can partition the bound on the LHS into $\lceil 2\Lambda\tau m\rceil$ intervals as $\Lambda\geq1/2\tau m$ from our choice of parameters.\footnote{To be precise, for a vector $v\in \R^k$, observe that the event $\Br{\forall i\in [k]: v_i\leq b_i+\Lambda, \text{ and }  \exists j\in [k]: v_j\geq b_j-\Lambda}$ can be broken down into the intersections of $\Lambda/2\tau m$ events given by $\bigwedge_{\ell=0}^{2\Lambda\tau m-1}[\forall i\in [k]:  v_i\leq b_i+\Lambda-\ell/2\tau m, \text{ and } \exists j\in [\ell]:  v_j> b_j-\Lambda-(\ell+1)/2\tau m]$.} and by a union bound  we have

    \begin{eqnarray*}
      &&\Pr_{\bx\sim \U_{n}}\Br{\exists j\in [2] \text{ s.t. } \lambda_{\max}\br{ \sum_i \bx_i A^i_j-B_j} \in (-\Lambda,\Lambda]}\\
       &\leq& O\br{\Lambda\cdot\tau\cdot m\br{\sqrt{\frac{1}{m}}+\exp(-\Omega(m/2))}}\\
    \end{eqnarray*}
    From the choice of the parameters, the first term above dominates. And thus
    \[
    \Pr_{\bx\sim\U_{n}}\Br{\exists j\in [2] \text{ s.t. } \lambda_{\max}\br{ \sum_i \bx_i A^i_j-B_j} \in [-\Lambda,0]}\leq O\br{\Lambda}.
    \]
    Similarly one can also show when the LHS of the equation above is replaced with $(0,\Lambda]$. Hence we get our theorem statement.
    \end{proof}

    \section{Invariance principle for positive spectrahedra}
    \label{sec:invarianceprinciple}
    In this section, we establish our main invariance principle.
    \subsection{Invariance principle for the spectral Bentkus mollifier}

    We now prove our main lemma which is an invariance principle for the Bentkus mollifier. We remark that our analysis is the standard Lindeberg-style argument for proving invariance principles, but when applied to the spectral Bentkus mollifier.  We first write out the Fr\'echet series for the Bentkus mollifier, which we then upper bound using our main Theorem~\ref{thm:derivative}. In order to upper bound the error terms in the Fr\'echet series, we use the matrix Rosenthal inequality (in Fact~\ref{fac:ros}) to bound the moments of random matrices (we remark that this inequality will also be useful in our $\PRG$ construction). Superficially, our proof techniques resemble the previous invariance principle proofs used in~\cite{harsha2013invariance,servedio2017fooling,o2019fooling}, but the quantities we need to bound are very different from their~analysis since we are dealing with matrices.

    \begin{lemma}\label{lem:invphitheta}
    Let $k\geq 1,\theta,\tau\in(0,1)$ and $\Psi_{\theta}:\sym{k}\rightarrow\reals$ be defined as $\Psi_{\theta}\br{Q}=\br{G_{\theta}\circ\lambda}\br{Q}$ where $G_{\theta}$ is the Bentkus mollifier defined in Eq.~\eqref{eqn:flambda}. Let $S_1,S_2$ be $(\tau,M)$-regular positive spectrahedra specified by matrices $\{A^1_1,\ldots,A^n_1,B_1\}$ and  $\{A^1_2,\ldots,A^n_2,B_2\}$ respectively. Let $A^i=\diag\br{A^i_1,A^i_2}$ and $B=\diag\br{B_1, B_2}$ be block diagonal matrices. Then
    \begin{align*}
      \abs{\E_{\bx\sim\U_n}\Br{\Psi_{\theta}\br{\sum_{i=1}^n\bx_iA^i-B}}-\E_{\bg\sim\G^{n}}\Br{\Psi_{\theta}\br{\sum_{i=1}^n\bg_iA^i-B}}}\leq O\br{\frac{\log^{7} k}{\theta^3}\cdot (M+\|B\|^2)\cdot (M\cdot \tau)^{1.5}}.
    \end{align*}
    This inequality still holds if $\bx$ is $\br{80\log k}$-wise uniform.
    \end{lemma}
    \begin{proof}

     Let $t=\lceil 1/\tau\rceil$. Let $\Hi=\{h:[n]\rightarrow [t]\}$ be a family of $(80\log k)$-wise uniform hashing functions, i.e., for every subset $I\subseteq [n]$ of size at most $80\log k$,  and $b\in [t]^I$, we~have
     $$
     \Pr_{\bh\in \Hi}\Br{\bh(i)=b_i}=\frac{1}{t^{|I|}},
     $$
     where the probability is taken over a uniformly random function $h\in \Hi$. Fix an $h\in \Hi$ (think of $h$ as a partition of $[n]$ into $t$ blocks $S_1,\ldots,S_t \subseteq [n]$, where $S_i=h^{-1}(i)$ for all $i\in [t]$). For $\bx\sim\U_n$ and $\by\sim\G^n$ let us divide $\bx,\by$ into blocks $\bx^1,\ldots,\bx^t$ and $\by^1,\ldots,\by^t$ according to $h$. It is not hard to see that $\bx^i\sim\U^{|h^{-1}(i)|}$ and $\by^i\sim\G^{|h^{-1}(i)|}$. We now upper bound the quantity
     \begin{align}
     \label{eq:invprinclhs}
     \abs{\mathop{\mathbb{E}}_{\bx\sim\U_{n}}\Br{\Psi_{\theta}\br{\sum_{i=1}^n \bx_i A^i-B}}-\mathop{\mathbb{E}}_{\by\in \G^{n}}\Br{\Psi_{\theta}\br{\sum_{i=1}^n \by_i A^i-B}}}
     \end{align}
     by the standard hybrid argument. Let $\{Z^0,\ldots,Z^t\}$ be a set of random variable on $n$ coordinates such that $Z^0$ is the uniform distribution on $\pmset{n}$ and $Z^t$ is uniform in $\G^n$. To this end, define $Z^\ell$ as follows: for $j\in [\ell]$, let $Z^\ell_{|h^{-1}(j)}=\by^j$ and for  $\ell<j\leq t$ let $Z^\ell_{|h^{-1}(j)}=\bx^j$. It is easy to see that $Z^0\sim\U_n$ and $Z^t\sim\G^n$. We now can upper bound Eq.~\eqref{eq:invprinclhs}~as
     \begin{align}
     \label{eq:upperboundinvprinclhs}
     \begin{aligned}
    &\abs{\mathop{\mathbb{E}}_{\bx\sim\U_{n}}\Br{\Psi_{\theta}\br{\sum_{i=1}^n \bx_i A^i-B}}-\mathop{\mathbb{E}}_{\by\sim \G^{n}}\Br{\Psi_{\theta}\br{\sum_{i=1}^n \by_i A^i-B}}}\\
    &=\abs{\sum_{\ell=1}^{t} \mathop{\mathbb{E}}_{\substack{\bx\sim\U_{n}\\\by\sim \G^n}}\Br{\Psi_{\theta}\br{\sum_{i=1}^n Z^\ell_i A^i-B}}-\mathop{\mathbb{E}}_{\substack{\bx\sim\U_{n}\\\by\sim \G^n}}\Br{\Psi_{\theta}\br{\sum_{i=1}^n Z^{\ell-1}_{i} A^i-B}}}\\
    &\leq \sum_{\ell=1}^{t} \abs{\mathop{\mathbb{E}}_{\substack{\bx\sim\U_{n}\\\by\sim \G^n}}\Br{\Psi_{\theta}\br{\sum_{i=1}^n Z^\ell_i A^i-B}}-\mathop{\mathbb{E}}_{\substack{\bx\sim\U_{n}\\\by\sim \G^n}}\Br{\Psi_{\theta}\br{\sum_{i=1}^n Z^{\ell-1}_{i} A^i-B}}}
    \end{aligned}
     \end{align}
     We now upper bound each of the $t$ quantities on the RHS of Eq.~\eqref{eq:upperboundinvprinclhs}. Fix $\ell\in [t]$ and let us assume for simplicity that $h^{-1}(\ell)=[m]$. By definition of $Z^\ell$ we observe that $Z^\ell_{j}=Z^{\ell+1}_{j}$ for all $j\in \{m+1,\ldots,n\}$ and in fact we have
     $$
     Z^\ell=(\bx_1,\ldots,\bx_m,Z_{m+1},\ldots,Z_n),\quad Z^{\ell+1}=(\by_1,\ldots,\by_m,Z_{m+1},\ldots,Z_n),
     $$
     where $\bx_i\sim\U_1$ and $y_i\in \G$ is uniform in their respective domains. Crucially note that $Z_{m+1},\ldots,Z_n$ is independent of the $\bx_i$s or $\by_i$s by definition of $Z^\ell,Z^{\ell+1}$. Rewriting the $\ell$-th term in Eq.~\eqref{eq:upperboundinvprinclhs}, we~get
     \begin{align}
     \label{eq:P1Q1P2Q2}
    \abs{\mathop{\mathbb{E}}_{\substack{\bx\sim\U_{n}\\\by\sim \G^n}}\Br{\Psi_{\theta}\br{\underbrace{\sum_{i=1}^m \bx_i A^i}_{Q}+\underbrace{\sum_{i=m+1}^n Z_i A^i -B }_{P}}}-\mathop{\mathbb{E}}_{\substack{\bx\sim\U_{n}\\\by\sim \G^n}}\Br{\Psi_{\theta}\br{\underbrace{\sum_{i=1}^m \by_i A^i}_{R}+ \underbrace{\sum_{i=m+1}^n Z_i A^i-B}_{P}}}}
     \end{align}
     Let us analyze both these quantities separately. We can first write the Fr\'echet series for both these expressions as
     \begin{align}
     \label{eq:frechetforfirst}
     \Psi_{\theta}(Q+P)=\Psi_{\theta}(P)+D\Psi_{\theta}\br{P}\Br{Q}+\frac{1}{2}D^2\Psi_{\theta}\br{P}\Br{Q,Q}+\frac{1}{6}D^3\Psi_{\theta}\br{P'}\Br{Q,Q,Q}
     \end{align}
     where $P'=P+\xi Q$ for some $\xi\in[0,1]$.\footnote{This follows directly from the mean value theorem for Fr\'echet derivatives~\cite{ambrosetti1995primer}.}
     \begin{align}
      \label{eq:frechetforsecond}
     \Psi_{\theta}(R+P)=\Psi_{\theta}(P)+D\Psi_{\theta}\br{P}\Br{R}+\frac{1}{2}D^2\Psi_{\theta}\br{P}\Br{R,R}+\frac{1}{6}D^3\Psi_{\theta}\br{P''}\Br{R,R,R},
     \end{align}
     where $P''=P+\xi' R$ for some $\xi\in[0,1]$.

     Now, observe that since the first moment and the second moment of $\bx$ match with the standard normal distributions. Thus we have that
     \begin{align}
     \label{eq:distributionsame}
     \begin{aligned}
     \E_{\substack{\bx\sim\U_{n}\\\by\sim \G^n}}\Br{D\Psi_{\theta}\br{P}\Br{R}}&= \E_{\substack{\bx\sim\U_{n}\\\by\sim \G^n}}\Br{D\Psi_{\theta}\br{P}\Br{Q}}\\
     \E_{\substack{\bx\sim\U_{n}\\\by\sim \G^n}}\Br{D^2\Psi_{\theta}\br{P}\Br{R,R}}&= \E_{\substack{\bx\sim\U_{n}\\\by\sim \G^n}}\Br{D^2\Psi_{\theta}\br{P}\Br{Q,Q}}.
     \end{aligned}
     \end{align}
     So by taking the difference of Eq.~\eqref{eq:frechetforsecond} and Eq.~\eqref{eq:frechetforfirst}, only the third order spectral derivatives remain to be bounded. For this, we now use the Corollary~\ref{cor:derivaativebentkus} and obtain
     \begin{align}
     \label{eq:usingsendovmainthm2}
     \big|D^3\Psi_{\theta}\br{P'}\Br{Q,Q,Q}\big|\leq O\br{\frac{\Delta_1^2}{\theta^3}\log^3 k\cdot\norm{Q}^3}\\
      \big|D^3\Psi_{\theta}\br{P''}\Br{R,R,R}\big|\leq O\br{\frac{\Delta_2^2}{\theta^3}\log^3 k\cdot\norm{R}^3}.
     \end{align}
    where $\Delta_1=\norm{P'}$ and $\Delta_2=\norm{P''}$.

    Thus, the absolute value of Eq.~\eqref{eq:P1Q1P2Q2} is upper bounded by
    \begin{equation}\label{eqn:invbound}
     \frac{\log^3 k}{\theta^3}\expec{\Delta_1^2\norm{Q}^3+\Delta_2^2\norm{R}^3}\leq \frac{\log^3 k}{\theta^3}\br{\expec{\norm{P'}^4}^{1/2}\expec{\norm{Q}^6}^{1/2}+\expec{\norm{P''}^4}^{1/2}\expec{\norm{R}^6}^{1/2}},
    \end{equation}
    where the inequality is by Cauchy-Schwarz inequality.

    Using \Cref{lem:normrandommatrixboolean} and the fact that $\sum_i(A^i)^2\preceq M\cdot \id$, we have
    \begin{align}
    \label{eq:normboundonP'''}
    \expec{\norm{P'}^4}\leq O\br{\log^2 k\cdot M^2+\norm{B}^4},\quad \expec{\norm{P''}^4}\leq O\br{\log^2 k\cdot M^2+\norm{B}^4}
    \end{align}
    We now upper bound the last term in Eq.~\eqref{eqn:invbound} using the following claim.
    \begin{claim}\label{claim:normQ}
      It holds that
    $\expec{\norm{Q}^6}\leq O\br{\log^6 k \cdot \tau^3 \cdot M^3}$, \quad   $\expec{\norm{R}^6}\leq O\br{\log^6 k \cdot \tau^3 \cdot M^3}$.
    \end{claim}

Before proving this claim, observe that combining Claim~\ref{claim:normQ} with Eq.~\eqref{eq:normboundonP'''},~\eqref{eqn:invbound},  we can upper bound Eq.~\eqref{eqn:invbound} (and in turn Eq.~\eqref{eq:P1Q1P2Q2})~by
    \begin{align*}
        O\br{\frac{\log^3 k}{\theta^3}\cdot \br{M \log k + \|B\|^2}\cdot \br{\log^3 k \cdot \tau^{1.5} \cdot M^{1.5}}}\leq O\br{\frac{\log^{7} k}{\theta^3}\cdot (M+\|B\|^2)\cdot (M\cdot \tau)^{1.5}}
    \end{align*}

     Putting together this inequality with Eq.~\eqref{eq:upperboundinvprinclhs}, we finally get
     \begin{align*}
    \Big|\mathop{\mathbb{E}}_{x\sim\U_{n}}\Br{\Psi_{\theta}\br{\sum_{i=1}^n x_i A^i}}-\mathop{\mathbb{E}}_{y\sim\G^{n}}\Br{\Psi_{\theta}\br{\sum_{i=1}^n y_i A^i}}\Big|
    \leq O\br{\frac{\log^{7} k}{\theta^3}\cdot (M+\|B\|^2)\cdot (M\cdot \tau)^{1.5}},
     \end{align*}
concluding the theorem proof. We now prove the claim above.

    \begin{proof}[Proof of Claim~\ref{claim:normQ}]
    Note that $Q=\sum_{i=1}^n\bx_iA^i$, where $\br{\bx_1,\ldots, \bx_n}$ is i.i.d.~with $\prob{\bx_i= 1}=\prob{\bx_i= -1}=\frac{1}{2t}$ and $\prob{\bx_i=0}=1-1/t$. Then using Fact~\ref{fac:ros}, we have
      \begin{eqnarray*}
        \expec{\norm{Q}_{8p}^{8p}}^{1/8p}     &\leq&\sqrt{8p-1}\Big\|\br{\frac{1}{t}\sum_i \br{A^i}^2}^{1/2}\Big\|_{8p}+\br{8p-1}\br{\frac{1}{t}\sum_i\norm{A^i}_{8p}^{8p}}^{1/8p}\\
        &\leq&\sqrt{8p-1}\cdot\sqrt{\frac{M}{t}}\cdot  k^{\frac{1}{8p}}+\br{8p-1}\br{\frac{\tau^{8p-2}\cdot  k\cdot M}{t}}^{1/8p}
      \end{eqnarray*}
      where the second inequality used $\sum_i\br{A^i}^2\preceq M\cdot \id$ for both terms and
      $0\preceq A^i\preceq \tau\id$  for upper bounding the second term.
      Setting $p=10\log  k$, $t=1/\tau$ we have
      \[
      \br{\expec{\norm{Q}_{8p}^{8p}}}^{1/8p}\leq O\br{\sqrt{\log k}\cdot\sqrt{\tau} \cdot\sqrt{M}+ \log k \cdot \tau\cdot (M/\tau)^{1/(80\log k)}}=O\br{\log k \cdot\sqrt{\tau} \cdot\sqrt{M}}.
      \]
      Thus, we have
      \[
      \expec{\norm{Q}^6}\leq\br{\expec{\norm{Q}_{8p}^{8p}}}^{\frac{3}{4p}}\leq O\br{\log^6 k \cdot \tau^3 \cdot M^3},
      \]
      where  in the first inequality note that the LHS is the spectral norm and the RHS is the $(8p)$-Schatten norm. This proves the first inequality in the claim statement. The second inequality in the claim follows by the exact same argument (since Fact~\ref{fac:ros} applies to even $\sum_i {\bg}_iA^i$).
    \end{proof}
    The proof of this claim concludes the proof of the theorem. Additionally, observe that since the largest Schatten power of $Q$ that we use is $8p=80\log k$, the proof of this theorem also works for $\bx$ that is $(80\log k)$-wise uniform.
    \end{proof}

    \subsection{Invariance principle for positive spectrahedra}
    We are now ready to prove our main theorem, which involves combining our anti-concentration Theorem~\ref{thm:gauanticonc1spec} and our invariance principle for Bentkus mollifier in Lemma~\ref{lem:invphitheta}.\footnote{We remark that our theorem statements should also hold true for a larger class of \emph{proper distributions} as considered in~\cite{harsha2013invariance}, which requires one to extend our main Theorem~\ref{eq:usingsendovinmainthm} to show that even the $4$th order  spectral derivatives can be bounded by $\|f^{(4)}\|_1$. We believe this should be possible and leave this to be made rigorous for future work.}

    \begin{theorem}
    \label{thm:invspecfinal}
    Let $k\geq 1$, $M\geq 1,\gamma\geq 1, \tau \in [0,1],\delta\in [0,1]$.
     Let $S_1,S_2$ be $(\tau,M)$-regular positive spectrahedra specified by matrices $\{A^1_1,\ldots,A^n_1,B_1\}\in \sym{k}$ and  $\{A^1_2,\ldots,A^n_2,B_2\}\in \sym{k}$ respectively satisfying
     $\norm{B_1},\norm{B_2}\leq~\gamma$.
     Let $S=S_1\cap S_2$.  If $\mu$ is a $\br{80\log k}$-wise uniform distribution over $\pmset{n}$, then

    \[
    \abs{\E_{\bx\sim\mu}[\bx\in S]-\E_{\bg\sim\G^n}[\bg\in S]}\leq C\cdot\br{M+\gamma^2}^{1/5}\cdot\log^{7/5}k\cdot M^{3/10}\cdot\tau^{3/10},
    \]
    for some universal constant $C>0$.
    \end{theorem}

    \begin{proof}
     Again for notational simplicity, let $A^i=\diag\br{A^i_1,A^i_2}$ and $B=\diag\br{B_1, B_2}$ be block diagonal matrices. We conclude the result by combining Fact~\ref{fac:ostanti}, Lemma~\ref{lem:invphitheta} and Corollary~\ref{thm:gauanticonc1spec} as follows: first Lemma~\ref{lem:invphitheta} implies
    \begin{align*}
    \begin{aligned}
      &\abs{\E_{\bx\sim\mu}\Br{\Psi_{\theta}\br{\sum_{i=1}^n\bx_iA^i-B}}-\E_{\bg\sim\G^{n}}\Br{\Psi_{\theta}\br{\sum_{i=1}^n\bg_iA^i-B}}} \leq O\br{\frac{\log^{7} k}{\theta^3}\cdot (M+\|B\|^2)\cdot (M\cdot \tau)^{1.5}},
      \end{aligned}
    \end{align*}
    In particular, using Fact~\ref{fac:ostanti} (for $D=B-\beta\cdot \id$ and $D=B+\beta\cdot \id$), the ``if" condition of Fact~\ref{fac:ostanti} is satisfied with
    $$
    \eta=O\br{\frac{\log^{7} k}{\theta^3}\cdot \br{M+(\gamma+\beta)^2}\cdot (M\cdot \tau)^{1.5}}
    $$
    where $\beta=O(\theta\cdot \sqrt{\log k/\delta})$.
       In particular, Fact~\ref{fac:ostanti} and Corollary~\ref{cor:gauanticonc1spec} now together imply that
      \begin{align*}
      &\abs{\E_{\bx\sim \mu}\Br{\Psi\br{\sum_{i=1}^n\bx_iA^i-B}}-\E_{\bg\sim \G^n}\Br{\Psi\br{\sum_{i=1}^n\bg_iA^i-B}}}\\
      &\leq\gamma+3\delta+\Pr_{\bg\sim \G^n}\Br{\lambda_{\max}\br{\sum_{i=1}^n\bg_iA^i-B}\in[-\Lambda,\Lambda]}\\
      &=O\br{\frac{\log^{7} k}{\theta^3}\cdot \br{M+\br{\gamma+\theta\cdot \sqrt{\log (k/\delta)}}^2}\cdot  (M\cdot \tau)^{1.5}+\delta+\Lambda}\\
       &\leq O\br{\frac{\log^{7} k}{\theta^3}\cdot \br{M+\br{\gamma+\sqrt{\log (k/\delta)}}^2}\cdot  (M\cdot \tau)^{1.5}+\delta+\Lambda}
      \end{align*}
      Let us fix
       \begin{align*}
       \theta\leftarrow\delta, \quad
       \theta\leftarrow\Lambda, \quad
        \br{ (M\cdot \tau)^{1.5}\cdot\log^{7} k\cdot \br{M+\br{\gamma+\sqrt{\log k}}^2}}^{1/5}\leftarrow \theta.
     \end{align*}
      This gives us
       \begin{align*}
      &\abs{\E_{\bx\sim \mu}\Br{\Psi\br{\sum_{i=1}^n\bx_iA^i-B}}-\E_{\bg\sim \G^n}\Br{\Psi\br{\sum_{i=1}^n\bg_iA^i-B}}}\leq \br{ (M\cdot \tau)^{1.5}\cdot\log^{7} k\cdot (M+\gamma^2)}^{1/5}.
      \end{align*}
    \end{proof}

    \subsection{Application: Pseudorandom generators for positive spectrahedra.}
    \label{sec:prg}
    We are now ready to describe our pseudorandom generator for fooling positive spectrahedra. Our $\PRG$ is based on the well-known construction of Meka and Zuckerman~\cite{meka2013pseudorandom} which we describe~now.  We remark that the same $\PRG$ (with minor modifications and different parameter settings) was used in~\cite{meka2013pseudorandom,harsha2013invariance,servedio2017fooling} in order to obtain $\PRG$s for polytopes.

     \paragraph{Meka-Zuckerman PRG.}
     We begin by describing the Meka-Zuckerman $\PRG$. Let us fix the parameters $\delta\in (0,1)$,
     $
     \tau= \Omega(\delta^{10/3}/(\log^{5}k\cdot M\cdot (M+ \gamma^{2}) ))
     $ so that we have $\br{M+\gamma^2}^{1/5}\cdot\log^{7/5}k\cdot M^{3/10}\cdot\tau^{3/10}=\delta$ (where the LHS of this equality is the upper bound obtained in our invariable principle proof). Let $t=\lceil 1/\tau\rceil$ and consider the family of $(80\log k)$-wise uniform functions $\Hi=\{h:[n]\rightarrow [t]\}$, i.e., for every for every subset $I\subseteq [n]$ of size at most $80\log k$,  and $b\in [t]^I$, we have
     $$
     \Pr_{\bh\in \Hi}\Br{\bh(i)=b_i}=\frac{1}{t^{|I|}},
     $$
     where the probability is taken over a uniformly random function $h\in \Hi$. Efficient constructions of such hash function families are known with $|\Hi|=O(n^{80\log k})$. For simplicity (as in the proof of~\cite{meka2013pseudorandom,harsha2013invariance}), we also assume that for every $j\in [t]$, we have $|h^{-1}(j)|=n/t$.  Let $m=n/t$ and $G_0:\01^s\rightarrow \pmset{m}$ generate a $(80\log k)$-wise  uniform distribution over $\pmset{m}$, i.e., for every $I\subseteq [n]$ of size at most $80\log k$ and $b\in \pmset{I}$, we have
     $$
     \Pr_{\substack{\bz\in \01^s\\\bx=G_0(\bz)}}[\bx_i = \bb_i \text{ for all } i\in I]=\frac{1}{2^{|I|}},
     $$
     where the probability is taken over uniformly random $z\in \01^s$. It is well-known by~\cite{naor1993small} that efficient  constructions of generators $G_0$ are known for $s=O(\log k \log n)$.  Finally, we are ready to describe the Meka-Zuckerman generator: for a given hash function family $\Hi$ and generator $G_0$, define $G:\Hi\times (\01^{s})^t\rightarrow \pmset{n}$ by
     $$
     G(h,z^1,\ldots,z^t)=x, \qquad \text{ where } x_{\vert h^{-1}(i)}=G_0(z^i) \text{ for } i\in [t].
     $$
     Clearly the seed length of this generator is
     $$
     O\br{(\log n)(\log k)+(\log n)(\log k)\frac{1}{\tau}}=O((\log n)(\log k)/\tau)=(\log n)\cdot\poly(\log k,M,1/\delta,\gamma),
     $$
     where the first term is the logarithm of the  number of elements of the hash function family $|\Hi|$, the second term because we have  $s=O((\log n) (\log k))$ and recall that we picked $t=O(1/\tau)$ and the final equality used the bound on $\tau$ we fixed at the start of the proof.

     We now restate our main theorem and prove it.

    \begin{theorem}
    Let $\delta \in (0,1)$, $k,n,M\geq 1$ and $\tau\leq \delta^{10/3}/(\log^{5}k\cdot M\cdot (M+ \gamma^{2}) )$.
     Let $S_1,S_2$ be $(\tau,M)$-regular positive spectrahedra specified by matrices $\{A^1_1,\ldots,A^n_1,B_1\}\in \sym{k}$ and  $\{A^1_2,\ldots,A^n_2,B_2\}\in \sym{k}$ with $\|B_1\|,\|B_2\|\leq \gamma$. Let $S=S_1\cap S_2$.
    There exists a $\PRG$ $G:\01^r\rightarrow \pmset{n}$ with
    $$
    r=(\log n)\cdot\poly(\log k,M,1/\delta,\gamma)
    $$
    that $\delta$-fools $S$ with respect to the uniform distribution.
    \end{theorem}

    The proof of this theorem is a generic statement that allows one to go from invariance principles proven using the proof techniques  to construct $\PRG$s. The proof uses the same proof ideas of Harsha, Klivans and Meka~\cite[Section~7.2]{harsha2013invariance} (except that now we directly proved \emph{Boolean} anti-concentration instead of the weaker \emph{Gaussian} anti-concentration as proven by~\cite{harsha2013invariance}). We provide the  proof  below for completeness.

    \begin{proof}
      Again for notational simplicity, let $A^i=\diag\br{A^i_1,A^i_2}$ and $B=\diag\br{B_1, B_2}$ be block diagonal matrices.
      The $\PRG$ $G$ will be the Meka-Zuckerman $\PRG$ defined above, so the seed length $r=(\log n)\cdot\poly(\log k,M,1/\delta,\gamma)$ immediately follows.
    \begin{align}
    \label{eq:finalprganti}
    \begin{aligned}
      \abs{\E_{\bx\sim\U_r}\Br{\Psi_{\theta}\br{\sum_{i=1}^n\br{G(\bx)}_iA^i-B}}-\E_{\bg\sim\G^{n}}\Br{\Psi_{\theta}\br{\sum_{i=1}^n\bg_iA^i-B}}}
      \leq O\br{\frac{\log^{7} k}{\theta^3}\cdot (M+\|B\|^2)\cdot (M\cdot \tau)^{1.5}},
      \end{aligned}
    \end{align}
    where we used the fact that $G(x)$ for uniformly random $x\in \01^r$ generates a $(80\log k)$-wise uniform distribution and Lemma~\ref{lem:invphitheta} holds for every
     $(80\log k)$-wise uniform distribution $\mu$. Repeating the same calculation that we did in the proof of Theorem~\ref{thm:invspecfinal}, we get
      \begin{align*}
      &\abs{\E_{\bx\sim \U_r}\Br{\Psi\br{\sum_{i=1}^n\br{G(\bx)}_iA^i-B}}-\E_{\bg\sim \G^n}\Br{\Psi\br{\sum_{i=1}^n\bg_iA^i-B}}}\\
      &\leq\gamma+3\delta+\Pr_{\bg\sim \G^n}\Br{\lambda_{\max}\br{A\br{\bg}}\in(-\Lambda,\Lambda]}\\
      &=O\br{\frac{\log^{7} k}{\theta^3}\cdot (M+\|B\|^2)\cdot (M\cdot \tau)^{1.5} +\delta+\Lambda},
      \end{align*}
      and using our assumption on $\tau$ (and the same parameters as in Theorem~\ref{thm:invspecfinal}), this implies that
      $$
      \abs{\E_{\bx\sim\U_r}[G(\bx)\in S]-\E_{\bg\sim\G^n}[\bg\in S]}\leq \delta,
      $$
    hence proving our theorem statement. 
    \end{proof}
\newcommand{\etalchar}[1]{$^{#1}$}

    \appendix
    \addtocontents{toc}{\protect\setcounter{tocdepth}{1}}

    \section{Proof of Lemma~\ref{lem:casev}: Case 2}
    \label{app:claim5proof}
    Recall that the goal is to prove the following inequality
    \begin{equation}
    \label{eq:restatecasevversion2}
    \abs{\sum_{i_1\neq i_2\neq i_3}\frac{\frac{\ophi\br{x_{i_1}}-\ophi\br{x_{i_3}}}{x_{i_3}-x_{i_1}}-\frac{\ophi\br{x_{i_1}}-\ophi\br{x_{i_2}}}{x_{i_2}-x_{i_1}}}{x_{i_3}-x_{i_2}}G\br{x}H_{i_1,i_2}H_{i_2,i_3}H_{i_3,i_1}}\leq O\br{\Delta\log^2k\norm{H}^3}
    \end{equation}
       First observe that the LHS of the inequality above can be rephrased as follows.
    \begin{align}\label{eqn:casevversion2}
    &\abs{\sum_{i_1\neq i_2\neq i_3}\frac{\frac{\ophi\br{x_{i_1}}-\ophi\br{x_{i_3}}}{x_{i_3}-x_{i_1}}-\frac{\ophi\br{x_{i_1}}-\ophi\br{x_{i_2}}}{x_{i_2}-x_{i_1}}}{x_{i_3}-x_{i_2}}G\br{x}H_{i_1,i_2}H_{i_2,i_3}H_{i_3,i_1}}\nonumber\\
     &=\abs{2\sum_{i_1\neq i_2\neq i_3\atop x_{i_3}>x_{i_2}}\frac{\frac{g'\br{x_{i_1}}g\br{x_{i_3}}-g\br{x_{i_1}}g'\br{x_{i_3}}}{x_{i_3}-x_{i_1}}g\br{x_{i_2}}-\frac{g'\br{x_{i_1}}g\br{x_{i_2}}-g\br{x_{i_1}}g'\br{x_{i_2}}}{x_{i_2}-x_{i_1}}g\br{x_{i_3}}}{x_{i_3}-x_{i_2}}G\br{x_{-\set{i_1,i_2,i_3}}}H_{i_1,i_2}H_{i_2,i_3}H_{i_3,i_1}}
    \end{align}
     Providing an upper bound on this consists of several lemmas and the result is concluded by combing all of them via triangle inequalities. To keep the expressions short, we use the following notations to represent Eq.~\eqref{eq:restatecasevversion2}, which are clear in the context.
    \begin{align}
    \label{eq:mainequationwecare}
    \abs{2\sum_{\substack{i_1\neq i_2\neq i_3\\ x_{i_3}>x_{i_2}}}\frac{\frac{\braket{i_1}'\braket{i_3}-\braket{i_1}\braket{i_3}'}{[i_3-i_1]}\braket{i_2}-\frac{\braket{i_1}'\braket{i_2}-\braket{i_1}\braket{i_2}'}{[i_2-i_1]} \braket{i_3}}{[i_3-i_2]}},
    \end{align}
    where we implicitly hide the $G\br{x_{-\set{i_1,i_2,i_3}}}H_{i_1,i_2}H_{i_2,i_3}H_{i_3,i_1}$ term. We first give a sketch of how we are going to upper bound this inequality and break it into subsections.
    \begin{equation}\label{eqn:rk1}
    \eqref{eq:mainequationwecare}=\underbrace{\frac{\braket{i_1}\braket{i_3}'-\braket{i_1}'\braket{i_3}}{[i_3-i_1]}\cdot\frac{\braket{i_3}-\braket{i_2}}{[i_3-i_2]}}_{Section~\ref{sec:2.1},\hspace{1mm} Lemma~\ref{lem:3}}-\underbrace{\frac{\frac{\braket{i_1}\braket{i_3}'-\braket{i_1}'\braket{i_3}}{[i_3-i_1]}-\frac{\braket{i_1}\braket{i_2}'-\braket{i_1}'\braket{i_2}}{[i_2-i_1]}}{[i_3-i_2]}\braket{i_3}}_{(\star)}.    
    \end{equation}

    We now break up $(\star)$ into two cases
    \begin{align}
        \label{daggerstar}
    (\star)=\underbrace{(\star)\cdot  \mathbb{I}[\min \{x_{i_1},x_{i_3}\}>x_{i_2}]}_{(\dagger)}+ \underbrace{(\star) \cdot  \mathbb{I}[x_{i_1}<x_{i_2}<x_{i_3}]}_{(\dagger \dagger)}.
    \end{align}
    Note that there are the only two cases we need to handle since  by symmetry between $i_2$ and $i_3$, we can assume $x_{i_3}>x_{i_2}$, without loss of generality. Now we bound these two terms, separately.
    $$
    (\dagger)=\underbrace{\frac{\frac{\braket{i_3}'-\braket{i_1}'}{[i_3-i_1]}-\frac{\braket{i_2}'-\braket{i_1}'}{[i_2-i_1]}}{[i_3-i_2]}\braket{i_1}\braket{i_3}}_{Section~\ref{sec:2.2},\hspace{1mm} Lemma~\ref{lem:5}}-\underbrace{\frac{\frac{\braket{i_3}-\braket{i_1}}{[i_3-i_1]}-\frac{\braket{i_2}-\braket{i_1}}{[i_2-i_1]}}{[i_3-i_2]}\braket{i_1}'\braket{i_3}}_{ Section~\ref{sec:2.2},\hspace{1mm} Lemma~\ref{lem:333}}.
    $$
    and
    \begin{equation}\label{eqn:rk2}
        (\dagger\dagger )= \underbrace{\frac{\frac{\braket{i_1}\braket{i_3}'-\braket{i_3}\braket{i_3}'}{[i_3-i_1]}-\frac{\braket{i_1}\braket{i_2}'-\braket{i_2}\braket{i_2}'}{[i_2-i_1]}}{[i_3-i_2]}\braket{i_3}}_{Section~\ref{sec:2.3},\hspace{1mm} Lemma~\ref{lem:casev111}}+ \underbrace{\frac{\frac{\braket{i_3}'-\braket{i_1}'}{[i_3-i_1]}\braket{i_3}-\frac{\braket{i_2}'-\braket{i_1}'}{[i_2-i_1]}\braket{i_2}}{[i_3-i_2]}\braket{i_3}}_{(\P)\hspace{1mm} Section~\ref{sec:2.3}},
    \end{equation}
    
    and
    \begin{align}
        \label{eq:Pterm}
    (\P)= \underbrace{
    \frac{\braket{i_3}'-\braket{i_1}'}{[i_3-i_1]}\cdot\frac{\braket{i_3}-\braket{i_2}}{[i_3-i_2]}\cdot\braket{i_3}
    }_{Section~\ref{sec:2.4},\hspace{1mm} Lemma~\ref{lem:444}}+ \underbrace{\frac{\frac{\braket{i_3}'-\braket{i_1}'}{[i_3-i_1]}-\frac{\braket{i_2}'-\braket{i_1}'}{[i_2-i_1]}}{[i_3-i_2]}\braket{i_2}\braket{i_3}}_{Section~\ref{sec:2.4},\hspace{1mm}Lemma~\ref{lem:555}}
    \end{align}
    Finally in order to upper bound Eq.~\eqref{eq:mainequationwecare}, we simply bound each of these terms by $O\br{\Delta\log^3k\norm{H}^3}$ in the respective sections (as underbraced by the terms).

    \subsection{Upper bounding first term in Eq.~\eqref{eqn:rk1}}
    \label{sec:2.1}

    \begin{lemma}\label{lem:3}
    \[
    \abs{\sum_{i_1\neq i_2\neq i_3:\atop x_{i_3}>x_{i_2}}\frac{\braket{i_1}\braket{i_3}'-\braket{i_1}'\braket{i_3}}{[i_3-i_1]}\cdot\frac{\braket{i_3}-\braket{i_2}}{[i_3-i_2]}}\leq O\br{\Delta^2\cdot\log^2k\cdot \norm{H}^3}.
    \]
    \end{lemma}

    \begin{proof}[Proof of Lemma~\ref{lem:3}]
      We apply Claim~\ref{claim:gg'} to the first sum and obtain $O\br{\Delta \max\set{g'\br{x_{i_1}},g'\br{x_{i_3}}}}$ (note that we have $\max\{\cdot,\cdot\}$ to compensate for the fact that $x_{i_1}\geq x_{i_3}$ or $x_{i_3}\geq x_{i_1}$). Therefore, the left hand side in Lemma~\ref{lem:3} can be upper bounded by
      \begin{eqnarray}
        &&O\br{\sum_{i_1\neq i_2\neq i_3:\atop x_{i_3}>x_{i_2}}\Delta\abs{\max\set{
        g'\br{x_{i_1}},g'\br{x_{i_3}}}\cdot\frac{g\br{x_{i_3}}-g\br{x_{i_2}}}{x_{i_3}-x_{i_2}}\cdot G\br{x_{-\set{i_1,i_2,i_3}}}H_{i_1,i_2}H_{i_2,i_3}H_{i_3,i_1}}}\nonumber\\
        &\leq& O\br{\sum_{i_1\neq i_2\neq i_3:\atop x_{i_3}>x_{i_2}\geq 0,x_{i_1}\geq 0}\br{\cdots}+\sum_{i_1\neq i_2\neq i_3:\atop x_{i_3}>x_{i_2}\geq 0,x_{i_1}<0}\br{\cdots}+\sum_{i_1\neq i_2\neq i_3:\atop x_{i_3}>x_{i_2},x_{i_2}<0, x_{i_1}\geq 0}\br{\cdots}+\sum_{i_1\neq i_2\neq i_3:\atop x_{i_3}>x_{i_2}, x_{i_1}<0,x_{i_2}<0}\br{\cdots}}\nonumber\\
        &&\label{eqn:casev1}
      \end{eqnarray}

    \textbf{First term in Eq.~\eqref{eqn:casev1}.} Note that $g\br{x}\geq\frac{1}{2}$ if $x\geq 0$. Since $g'$ is monotone decreasing in the interval $[0,\infty)$, the first summation is upper bounded by
    \begin{eqnarray}
      &&O\br{\abs{\sum_{i_1\neq i_2\neq i_3:\atop x_{i_3}>x_{i_2}\geq 0,x_{i_1}\geq 0}\Delta\abs{\max\set{g'\br{x_{i_1}}g'\br{x_{i_2}}G\br{x_{-i_3}},g'\br{x_{i_3}}g'\br{x_{i_2}}G\br{x_{-i_1}}}H_{i_1,i_2}H_{i_2,i_3}H_{i_3,i_1}}}}\nonumber\\
        &\leq&O\br{\Delta\cdot\|G^{(2)}\|_1\cdot\max_{i_1,i_2}\sum_{i_3}\abs{H_{i_1,i_2}H_{i_2,i_3}H_{i_3,i_1}}}\nonumber\\
      &\leq&O\br{\Delta\cdot\log k\cdot\max_{i_1,i_2}\sum_{i_3}\abs{H_{i_1,i_2}H_{i_2,i_3}H_{i_3,i_1}}}\leq O\br{\Delta\cdot\log k\cdot\norm{H^3}}\label{eqn:casev++}
    \end{eqnarray}
    where the second inequality is from Fact~\ref{fac:benktus2} and the last inequality follows by Eq.~\eqref{eqn:h3}.

    \textbf{Second term in Eq.~\eqref{eqn:casev1}.} The second summation is upper bounded as follows. Again by the mean value theorem, we observe that
    \begin{eqnarray*}
      &&O\br{\sum_{i_1\neq i_2\neq i_3:\atop x_{i_3}>x_{i_2}\geq 0,x_{i_1}<0}\Delta\abs{\max\set{g'\br{x_{i_1}}g'\br{x_{i_2}},g'\br{x_{i_3}}g'\br{x_{i_2}}}G\br{x_{-i_1}}H_{i_1,i_2}H_{i_2,i_3}H_{i_3,i_1}}} \\
      &\leq&O\br{\Delta\cdot\sum_{i_1:x_{i_1}<0}\onenorm{G^{(1)}\br{x_{-i_1}}}\max_{i_2}\sum_{i_3}\abs{H_{i_1,i_2}H_{i_2,i_3}H_{i_3,i_1}}+\onenorm{G^{(2)}\br{x_{-i_1}}}\max_{i_2,i_3}\abs{H_{i_1,i_2}H_{i_2,i_3}H_{i_3,i_1}}}\\
      &\leq& O\br{\Delta\cdot\log^{1.5}k\cdot\norm{H}^3},
    \end{eqnarray*}
     where the last inequality is from Fact~\ref{fac:benktus2}, Eq.~\eqref{eqn:h3} and the assumption that $\abs{\set{i:x_i\leq 0}}\leq 3\log k$.

    \textbf{Third term in Eq.~\eqref{eqn:casev1}.} Using the fact that $g'(\cdot)$ is bounded by a constant, the third summation is upper bounded by
    \begin{eqnarray}
      &&O\br{\sum_{i_1\neq i_2\neq i_3:\atop x_{i_3}>x_{i_2},x_{i_2}<0, x_{i_1}\geq 0}\Delta\abs{\max\set{
        g'\br{x_{i_1}},g'\br{x_{i_3}}}\cdot G\br{x_{-\set{i_2,i_3}}}H_{i_1,i_2}H_{i_2,i_3}H_{i_3,i_1}}}\nonumber\\
      &=&O\br{\sum_{i_1\neq i_2\neq i_3:\atop x_{i_3}>x_{i_2}, x_{i_1}\geq 0,x_{i_2}<0,x_{i_3}\geq 0}\br{\cdots}+\sum_{i_1\neq i_2\neq i_3:\atop x_{i_3}>x_{i_2}, x_{i_1}\geq 0,x_{i_2}<0,x_{i_3}< 0}\br{\cdots}}.\label{eqn:3term}
    \end{eqnarray}
    For the first summation in Eq.~\eqref{eqn:3term}, using the fact that $g\br{x}\geq\frac{1}{2}$ when $x\geq 0$, it is upper bounded~by
    \begin{eqnarray*}
      &&O\br{\Delta\sum_{i_1\neq i_2\neq i_3:\atop x_{i_3}>x_{i_2}, x_{i_1}\geq 0,x_{i_2}<0,x_{i_3}\geq 0}\abs{\max\set{
        g'\br{x_{i_1}},g'\br{x_{i_3}}}\cdot G\br{x_{-\set{i_2}}}H_{i_1,i_2}H_{i_2,i_3}H_{i_3,i_1}}} \\
      &\leq&O\br{\Delta\sum_{i_2:x_{i_2}<0}\|G^{(1)}\|_1\max_{i_1}\sum_{i_3}\abs{H_{i_1,i_2}H_{i_2,i_3}H_{i_3,i_1}}}\\
      &\leq&O\br{\Delta\sum_{i_2:x_{i_2}<0}\sqrt{\log k}\max_{i_1}\sum_{i_3}\abs{H_{i_1,i_2}H_{i_2,i_3}H_{i_3,i_1}}}\\
      &\leq&O\br{\Delta\cdot\log^{1.5}k\cdot\norm{H}^3},
    \end{eqnarray*}
    where the second inequality is from Fact~\ref{fac:benktus2}, and the last inequality used Eq.~\eqref{eqn:h3} and the assumption that $\abs{\set{i:x_i\leq 0}}\leq 3\log k$.

     In order to upper bound the second summation in Eq.~\eqref{eqn:3term}, first observe that both $g\br{\cdot}$ and $G\br{\cdot}$ are positive and upper bounded by $1$. Thus, Eq.~\eqref{eqn:3term} can be bounded as
     \begin{eqnarray*}
       O\br{\Delta\sum_{i_2\neq i_3:\atop x_{i_2}<0,x_{i_3}<0}\sum_{i_1}\abs{H_{i_1,i_2}H_{i_2,i_3}H_{i_3,i_1}}}\leq O\br{\Delta\cdot\log^2k\cdot\norm{H}^3}.
     \end{eqnarray*}
     where we again use  Eq.~\eqref{eqn:h3} and the assumption that $\abs{\set{i:x_i\leq 0}}\leq 3\log k$.

    \textbf{Fourth term in Eq.~\eqref{eqn:casev1}.} The last summation is upper bounded by $O\br{\Delta\cdot\log^2k\cdot\norm{H}^3}$ using the same arguments to upper bound the second summation in Eq.~\eqref{eqn:3term}.
    \end{proof}

    \subsection{Upper bounding $(\dagger)$ first term in $(\star)$  in Eq.~\eqref{daggerstar}}
    \label{sec:2.2}
    We upper bound the quantity in $(\star)$ in two cases that $x_{i_1}>x_{i_2}$ and $x_{i_2}>x_{i_1}$. In order to prove this lemma we need the following lemmas and claims.

    \begin{claim}\label{claim:xhxh}
    	For integer $k\geq 1$, $X\in\sym{k}$ and $H\in\mat{k}$ it holds that
    	\[
    	\big\|\br{XH+HX}e^{-\frac{X^2}{2}}\big\|_2\leq2\norm{X}\cdot\twonorm{He^{-X^2/2}}
    	\]
    	and
    	\[
    	\big\|e^{-\frac{X^2}{2}}\br{XH+HX}\big\|_2\leq2\norm{X}\cdot\twonorm{e^{-X^2/2}H}\]
    \end{claim}
    \begin{proof}
    	As the Schattern norm is unitarily invariant, we assume that $X=\diag\br{x_1,\ldots,x_n}$ is diagonal without loss of generality. Then
    \[\twonorm{\br{XH+HX}e^{-X^2/2}}^2=\sum_{i,j}H_{i,j}^2\br{x_i+x_j}^2e^{-x_j^2}\leq 4\norm{X}^2\cdot\sum_{i,j}H_{i,j}^2e^{-x_j^2}=4\Delta^2\twonorm{He^{-X^2/2}}^2.\]
    The second inequality follows by the same argument.
    \end{proof}

    \begin{lemma}\label{lem:xhxh}
      Given an integer $k\geq 1$, $u_1,u_2,u_3\geq 0$ satisfying $u_1+u_2+u_3=1$ and $X\in\sym{k}, H_1, H_2,H_3\in\mat{k}$, if $u_1,u_3\leq\frac{1}{2}$, then it holds that
      \begin{eqnarray*}
      &&\abs{\Tr \Br{e^{-u_1X^2}H_1e^{-u_2X^2}H_2e^{-u_3X^2}H_3}}\\
      &\leq&\br{\twonorm{H_1e^{-\frac{1}{2}X^2}}+\twonorm{e^{-\frac{1}{2}X^2}H_1}}\cdot\br{\twonorm{H_2e^{-\frac{1}{2}X^2}}+\twonorm{e^{-\frac{1}{2}X^2}H_2}}\cdot\norm{H_3}.
      \end{eqnarray*}
    \end{lemma}
    \begin{proof}
    Using the inequality $\abs{\Tr ABC}\leq\twonorm{A}\cdot\twonorm{B}\cdot\norm{C}$ (where $\|\cdot \|_2$ is the standard Frobenius norm and $\|\cdot \|$ is the spectral norm), we have
      \begin{eqnarray*}
        &&\abs{\Tr \Br{e^{-u_1X^2}H_1e^{-u_2X^2}H_2e^{-u_3X^2}H_3}}\ \\
        &\leq&\twonorm{e^{-u_1X^2}H_1e^{\br{u_1-\frac{1}{2}}X^2}}\cdot\twonorm{e^{-u_3X^2}H_2e^{\br{u_3-\frac{1}{2}}X^2}}\cdot \norm{H_3}.
      \end{eqnarray*}
      We conclude the result by Lemma~\ref{lem:diagonal}.
    \end{proof}
    \begin{lemma}\label{lem:diagonal}
    Given diagonal matrices $A=\diag\br{a_1,\ldots, a_k}, B=\diag\br{b_1,\ldots, b_k}$ with $a_1\geq\cdots\geq a_k\geq 0$ and $b_1\geq\cdots\geq b_k\geq 0$ and an arbitrary matrix $H$, it holds that
    \[\twonorm{AHB}^2+\twonorm{BHA}^2\leq\twonorm{HAB}^2+\twonorm{ABH}^2.\]
    In particular,
    \[\twonorm{AHB}\leq\twonorm{HAB}+\twonorm{ABH}.\]
    \end{lemma}
    \begin{proof}
      Note that
      \begin{eqnarray*}
        &&\br{\twonorm{HAB}^2+\twonorm{ABH}^2}-\br{\twonorm{AHB}^2+\twonorm{BHA}^2}\\
        &=&\sum_{i,j}H_{i,j}^2\br{a_i^2b_i^2+a_j^2b_j^2-a_i^2b_j^2-a_j^2b_i^2}\\
        &=&\sum_{i,j}\br{a_i^2-a_j^2}\br{b_i^2-b_j^2}\geq 0,
      \end{eqnarray*}
      where the first equality is from the symmetry.
    \end{proof}

    \begin{lemma}\label{lem:trd}
      Given an integer $k\geq 1$, matrices $A, B, C\in\mat{k}$ and $X\in\sym{k}$ with $\norm{X}\leq\Delta$, it holds that
      \begin{eqnarray*}
        &&\abs{\Tr\Br{D^2\br{e^{-X^2/2}}\Br{A,B}C}} \\
        &\leq&4\Delta^2\cdot\max\begin{Bmatrix}\br{\twonorm{Ae^{-X^2/2}}+\twonorm{e^{-X^2/2}A}}\cdot\br{\twonorm{Be^{-X^2/2}}+\twonorm{e^{-X^2/2}B}}\cdot\norm{C},\\
        \br{\twonorm{Ae^{-X^2/2}}+\twonorm{e^{-X^2/2}A}}\cdot\br{\twonorm{Ce^{-X^2/2}}+\twonorm{e^{-X^2/2}C}}\cdot\norm{B},\\ \br{\twonorm{Be^{-X^2/2}}+\twonorm{e^{-X^2/2}B}}\cdot\br{\twonorm{Ce^{-X^2/2}}+\twonorm{e^{-X^2/2}C}}\cdot\norm{A}\end{Bmatrix}.
      \end{eqnarray*}
    \end{lemma}
    \begin{proof}
      Combining Lemma~\ref{lem:ex2derivative}, Lemma~\ref{lem:xhxh} and the inequality that
      \[\twonorm{\br{XA+AX}e^{-X^2/2}}+\twonorm{e^{-X^2/2}\br{XA+AX}}\leq2\Delta\br{\twonorm{Ae^{-X^2/2}}+\twonorm{e^{-X^2/2}A}}\]
      and
      \[\twonorm{\br{XB+BX}e^{-X^2/2}}+\twonorm{e^{-X^2/2}\br{XB+BX}}\leq2\Delta\br{\twonorm{Be^{-X^2/2}}+\twonorm{e^{-X^2/2}B}}\]
      and
      \[\norm{XA+AX}\leq2\Delta\norm{A},\hspace{4mm} \norm{XB+BX}\leq2\Delta\norm{B},\]
      we conclude the result.
    \end{proof}

    \begin{lemma}\label{lem:5}
      \[
      \abs{\sum_{i_1\neq i_2\neq i_3:\atop x_{i_1}>x_{i_2},x_{i_3}>x_{i_2}}\frac{\frac{\braket{i_3}'-\braket{i_1}'}{[i_3-i_1]}-\frac{\braket{i_2}'-\braket{i_1}'}{[i_2-i_1]}}{[i_3-i_2]}\braket{i_1}\braket{i_3}}\leq O\br{\Delta^2\cdot \log^{2.5}k \norm{H}^3}.
      \]
    \end{lemma}

    \begin{proof}[Proof of Lemma~\ref{lem:5}]
      We break the summation into two summations
      \begin{equation}\label{eqn:twosum}
       \abs{\sum_{i_1\neq i_2\neq i_3\atop x_{i_1}>x_{i_3}>x_{i_2}}\br{\cdots}}+\abs{\sum_{i_1\neq i_2\neq i_3\atop x_{i_3}>x_{i_1}>x_{i_2}}\br{\cdots}}
      \end{equation}
      For the first summation, we define
      \begin{equation*}
        A_{i,j}=\begin{cases}
                  H_{i,j}, & \mbox{if $x_{i}<x_{j}$}  \\
                  0, & \mbox{otherwise}.
                \end{cases}
      \end{equation*}
    and
    Then $\norm{A}\leq\log k\cdot\norm{H}$ by Fact~\ref{fac:bhatia} (without loss of generality, we may assume that $x_i$s are sorted in increasing order. Further notice that all the diagonal entries of $H$ are zeros. Thus $A$ is the upper triangle part of $H$). We first bound the  first term in Eq.~\eqref{eqn:twosum}. In this direction, we first rewrite it as
      \begin{eqnarray}
        &&\frac{1}{\sqrt{2\pi}}\sum_{i_2}G\br{x_{-i_2}}\br{\br{D^2\br{e^{-X^2/2}}[A,A^T]H}_{i_2,i_2}}=\frac{1}{\sqrt{2\pi}}\sum_{i_2{x_{i_2}< 0}}\br{\cdots}+\frac{1}{\sqrt{2\pi}}\sum_{i_2:x_{i_2}\geq 0}\br{\cdots}\nonumber\\
        &&\label{eqn:gi1}
        \end{eqnarray}
        where $X=\diag\br{x_1,\ldots,x_k}$ and we implicitly used that we are summing over terms with  $x_{i_2}< x_{i_3}$. Note that $A$ is obtained from $H$ by zeroing out part of entries. Thus 
        \begin{equation}\label{eqn:ahzero}
           \max\set{\twonorm{Ae^{-X^2/2}},\twonorm{A^Te^{-X^2/2}},\twonorm{e^{-X^2/2}A},\twonorm{e^{-X^2/2}A^T}}\leq\twonorm{He^{-X^2/2}}, 
        \end{equation}
        To upper bound first summation in Eq.~\eqref{eqn:gi1}, we apply Lemma~\ref{lem:trd}, Eq.~\eqref{eqn:ahzero} and  inequalities
        $\norm{A}\leq\log k\cdot\norm{H}$ and $\norm{HE_{i_2,i_2}}\leq\norm{H}$ and obtain
        \begin{equation}\label{eqn:trd2}
          \abs{\Tr\br{D^2\br{e^{-X^2/2}}[A,A^T]HE_{i_2,i_2}}}
          \leq16\Delta^2\log k\twonorm{He^{-X^2/2}}^2\cdot\norm{H}. 
        \end{equation}

       Thus, the first summation in Eq.~\eqref{eqn:gi1} is upper bounded by
        \begin{eqnarray}
          &&\frac{\Delta^2\cdot\log k}{\sqrt{2\pi}}\sum_{i_2:x_{i_2}<0}G\br{x_{-i_2}}\twonorm{He^{-X^2/2}}^2\cdot\norm{H}\nonumber\\
          &=&\br{\Delta^2 \log k\cdot \sum_{i_2:x_{i_2}<0}G\br{x_{-i_2}}\sum_{i_1,i_3}e^{-x_{i_3}^2}H_{i_1,i_3}^2\norm{H}}\nonumber\\
          &\leq&\br{\Delta^2 \log^2 k\cdot\max_{i_2}\sum_{i_1\neq i_3} g'\br{x_{i_3}}\cdot G\br{x_{-i_2}}\cdot H_{i_1,i_3}^2\norm{H}}\nonumber\\
          &\leq& O\br{\Delta^2\log^{2.5}k \norm{H}^3},\label{eqn:firstsummation}
        \end{eqnarray}
        where the first inequality is from the assumption that $\abs{\set{i:x_i<0}}\leq3\log k$ and the second inequality is from Fact~\ref{fac:benktus2}.

        For the second summation in Eq.~\eqref{eqn:gi1}, we define
        \begin{equation*}
          \tilde{H}_{i,j}=\begin{cases}
                                \frac{H_{i,j}}{g\br{x_j}}, & \mbox{if $x_j\geq 0$}  \\
                                0, & \mbox{otherwise}.
                              \end{cases}
        \end{equation*}
        Then $\norm{\tilde{H}}\leq2\norm{H}$ as $g\br{x_i}\geq\frac{1}{2}$ if $x_i\geq 0$.
        Again applying Eq.~\eqref{eqn:ahzero}, we can  verify that the second summation in Eq.~\eqref{eqn:gi1} is equal to
        \begin{eqnarray*}
         \abs{\frac{1}{\sqrt{2\pi}}G\br{x}\Tr~D^2\br{e^{-X^2/2}}[A,A^T]\tilde{H}}\leq  \frac{16 \Delta^2\log k}{\sqrt{2\pi}}G\br{x}\twonorm{He^{-X^2/2}}^2\norm{H}\leq O\br{\Delta^2\cdot\log^{1.5}k\norm{H}^3}.
        \end{eqnarray*}
        where the first inequality is from Lemma~\ref{lem:trd} and the second inequality is from Fact~\ref{fac:benktus2}.

    Finally, the second summation in Eq.~\eqref{eqn:twosum} can be upper bounded using the verbatim same arguments by $O\br{\Delta^2\cdot\log^{2.5} k\cdot\norm{H}^3}$. This proves the lemma statement.
    \end{proof}

    \begin{lemma}\label{lem:333}
      \[
      \abs{\sum_{i_1\neq i_2\neq i_3:\atop x_{i_3}>x_{i_2},x_{i_1}>x_{i_2}}\frac{\frac{\braket{i_3}-\braket{i_1}}{[i_3-i_1]}-\frac{\braket{i_2}-\braket{i_1}}{[i_2-i_1]}}{[i_3-i_2]}\braket{i_1}'\braket{i_3}}\leq O\br{\Delta\cdot\log^{1.5}k\cdot\norm{H}^3}
      \]
    \end{lemma}
    \begin{proof}
    \begin{align}
      &&\abs{\sum_{i_1\neq i_2\neq i_3:\atop x_{i_3}>x_{i_2},x_{i_1}>x_{i_2}}\frac{\frac{\braket{i_3}-\braket{i_1}}{[i_3-i_1]}-\frac{\braket{i_2}-\braket{i_1}}{[i_2-i_1]}}{[i_3-i_2]}\braket{i_1}'\braket{i_3}}=\abs{\sum_{i_1\neq i_2\neq i_3:\atop x_{i_3}>x_{i_2},x_{i_1}>x_{i_2}\geq 0}\br{\cdots}+\sum_{i_1\neq i_2\neq i_3:\atop x_{i_3}>x_{i_2},x_{i_1}>0, x_{i_2}< 0}\br{\cdots}}\label{eqn:333}
    \end{align}
    To upper bound the first summation in Eq.~\eqref{eqn:333}, we apply Fact~\ref{fac:mvtdd} and upper bound the first summation by
    \begin{eqnarray*}
      &&O\br{\sum_{i_1\neq i_2\neq i_3:\atop x_{i_3}>x_{i_2},x_{i_1}>x_{i_2}\geq 0}\abs{g''\br{\xi_{i_1,i_2,i_3}}g'\br{x_{i_1}}G\br{x_{-\set{i_1,i_2}}}H_{i_1,i_2}H_{i_2,i_3}H_{i_3,i_1}}}\\
      &\leq&O\br{\abs{\Delta\cdot\sum_{i_1\neq i_2\neq i_3:\atop x_{i_3}>x_{i_2},x_{i_1}>x_{i_2}\geq 0} g'\br{x_{i_2}}g'\br{x_{i_1}}G\br{x_{-\set{i_1,i_2}}}H_{i_1,i_2}H_{i_2,i_3}H_{i_3,i_1}}}\\
      &\leq&O\br{\twonorm{G^{(2)}\br{x}}\max_{i_1,i_2}\sum_{i_3}\abs{H_{i_1,i_2}H_{i_2,i_3}H_{i_3,i_1}}}\\
      &\leq&O\br{\Delta\cdot\log k\cdot\norm{H}^3}
      \end{eqnarray*}
      where the last inequality is from Fact~\ref{fac:benktus2} and Eq.~\eqref{eqn:h3}.   Note that $\abs{g''\br{\xi}}\leq \Delta$ for any $\xi\in[x_{i_2},\max\set{x_{i_1},x_{i_3}}]$ by Eq.~\eqref{eqn:g''}. Applying Fact~\ref{fac:mvtdd}, the second summation in Eq.~\eqref{eqn:333} is upper bounded by
      \begin{eqnarray*}
        &&O\br{\Delta\sum_{i_2:x_{i_2}<0}\sum_{i_1,i_3}g'\br{x_{i_1}}G\br{x_{-\set{i_1,i_2}}}\abs{H_{i_1,i_2}H_{i_2,i_3}H_{i_3,i_1}}} \\
        &\leq&O\br{\Delta\cdot \log k\cdot\max_{i_2}\onenorm{G^{(1)}\br{x_{-i_2}}}\cdot\max_{i_1}\sum_{i_3}\abs{H_{i_1,i_2}H_{i_2,i_3}H_{i_3,i_1}}}\\
        &\leq& O\br{\Delta\cdot\log^{1.5}k\cdot\norm{H}^3}
      \end{eqnarray*}
      where the first inequality is from the assumption that $\abs{\set{i:x_i<0}}\leq 3\log k$ and the second inequality is from Fact~\ref{fac:benktus2} and Eq.~\eqref{eqn:h3}.
    \end{proof}

    \subsection{Upper bounding first term in $(\dagger\dagger)$ in Eq.~\eqref{eqn:rk2}}
    \label{sec:2.3}
    We now bound the first term in Eq.~\eqref{eqn:rk2} when  $x_{i_3}>x_{i_2}>x_{i_1}$. Recall that the goal is to upper bound the following lemma.

    \begin{lemma}\label{lem:casev111}
      \[\abs{\sum_{i_1\neq i_2\neq i_3:\atop x_{i_3}>x_{i_2}>x_{i_1}}\frac{\frac{\braket{i_1}\braket{i_3}'-\braket{i_3}\braket{i_3}'}{[i_3-i_1]}-\frac{\braket{i_1}\braket{i_2}'-\braket{i_2}\braket{i_2}'}{[i_2-i_1]}}{[i_3-i_2]}\braket{i_3}}\leq O\br{\Delta\cdot\log^{1.5}k\cdot\norm{H}^3}\]
    \end{lemma}
    \begin{proof}[Proof of Lemma~\ref{lem:casev111}]
      \begin{eqnarray}
        &&\abs{\sum_{i_1\neq i_2\neq i_3:\atop x_{i_3}>x_{i_2}>x_{i_1}}\frac{\frac{\braket{i_1}\braket{i_3}'-\braket{i_3}\braket{i_3}'}{[i_3-i_1]}-\frac{\braket{i_1}\braket{i_2}'-\braket{i_2}\braket{i_2}'}{[i_2-i_1]}}{[i_3-i_2]}\braket{i_3}}\nonumber\\
        &\leq&\sum_{i_1\neq i_2\neq i_3:\atop x_{i_3}>x_{i_2}>x_{i_1}}\abs{\frac{\frac{\braket{i_1}\braket{i_2}'-\braket{i_3}\braket{i_2}'}{[i_3-i_1]}-\frac{\braket{i_1}\braket{i_2}'-\braket{i_2}\braket{i_2}'}{[i_2-i_1]}}{[i_3-i_2]}\braket{i_3}}+\sum_{i_1\neq i_2\neq i_3:\atop x_{i_3}>x_{i_2}>x_{i_1}}\abs{\frac{\frac{\braket{i_1}\braket{i_3}'-\braket{i_3}\braket{i_3}'}{[i_3-i_1]}-\frac{\braket{i_1}\braket{i_2}'-\braket{i_3}\braket{i_2}'}{[i_3-i_1]}}{[i_3-i_2]}\braket{i_3}}\nonumber\\
        &=&  \sum_{i_1\neq i_2\neq i_3:\atop x_{i_3}>x_{i_2}>x_{i_1}}\abs{\frac{\frac{\braket {i_1}-\braket{i_3}}{[i_3-i_1]}-\frac{\braket {i_1}-\braket{i_2}}{[i_2-i_1]}}{[i_3-i_2]}\braket{i_2}'\braket{i_3}}+\sum_{i_1\neq i_2\neq i_3:\atop x_{i_3}>x_{i_2}>x_{i_1}}\abs{\frac{\braket{i_1}-\braket{i_3}}{[i_3-i_1]}\cdot\frac{\braket{i_3}'-\braket{i_2}'}{[i_3-i_2]}\cdot\braket{i_3}}\label{eqn:1111}
      \end{eqnarray}
      The first term is upper bounded by $O\br{\Delta\cdot\log^{1.5}k\cdot\norm{H}^3}$ using the same argument in Lemma~\ref{lem:333}.  The second term can be rephrased as
      \begin{eqnarray}
        &&\abs{\sum_{i_1\neq i_2\neq i_3:\atop x_{i_3}>x_{i_2}>x_{i_1}}\frac{g\br{x_{i_1}}-g\br{x_{i_3}}}{x_{i_3}-x_{i_1}}\cdot\frac{g'\br{x_{i_3}}-g'\br{x_{i_2}}}{x_{i_3}-x_{i_2}}g\br{x_{i_3}}G\br{x_{-i_1}}H_{i_1,i_2}H_{i_2,i_3}H_{i_3,i_1}} \nonumber\\
        &\leq&\abs{\sum_{i_1\neq i_2\neq i_3:\atop x_{i_3}>x_{i_2}>x_{i_1},x_{i_1}\geq 0}\br{\cdots}}+\abs{\sum_{i_1\neq i_2\neq i_3:\atop x_{i_3}>x_{i_2}>x_{i_1},x_{i_1}< 0}\br{\cdots}}\label{eqn:123}
      \end{eqnarray}
      For the first summation in Eq.~\eqref{eqn:123}, we apply the mean value theorem for both $g$ and $g'$. From
      Eq.~\eqref{eqn:g''} it is upper bounded by
     \begin{eqnarray*}
       &&\sum_{i_1\neq i_2\neq i_3:\atop x_{i_3}>x_{i_2}>x_{i_1},x_{i_1}\geq 0}\Delta\abs{g'\br{x_{i_1}}g'\br{x_{i_2}}g\br{x_{i_3}}G\br{x_{-i_1}}H_{i_1,i_2}H_{i_2,i_3}H_{i_3,i_1}}
       \\
       &\leq&O\br{\Delta\cdot\onenorm{G^{(2)}\br{x}}\norm{H}^3}\\
       &\leq& O\br{\Delta\cdot\log k\cdot\norm{H}^3}.
     \end{eqnarray*}
     For the second term in Eq.~\eqref{eqn:123}, it is not hard to verify that
      \begin{equation}\label{eqn:mvtg'}
        \abs{\frac{g'\br{x_{i_3}}-g'\br{x_{i_2}}}{x_{i_3}-x_{i_2}}}\leq\Delta\max\set{g'\br{x_{i_3}},g\br{x_{i_2}}}
      \end{equation}
       Further notice that $\abs{g'\br{\cdot}}\leq 1$. Applying the mean value theorem to $g$, we upper bound the second summation in \ref{eqn:123} by
      \begin{eqnarray*}
        &&O\br{\Delta\sum_{i_1\neq i_2\neq i_3:\atop x_{i_3}>x_{i_2}>x_{i_1},x_{i_1}< 0}\max\set{g'\br{x_{i_3}},g'\br{x_{i_2}}}g\br{x_{i_3}}G\br{x_{-i_1}}\abs{H_{i_1,i_2}H_{i_2,i_3}H_{i_3,i_1}}}\\
        &\leq&O\br{\Delta\cdot\log k\cdot\max_{i_1}\cdot\onenorm{G^{(1)}\br{x_{-i_1}}}\cdot\max_{i_2}\sum_{i_3}\abs{H_{i_1,i_2}H_{i_2,i_3}H_{i_3,i_1}}}\\
        &\leq&O\br{\Delta\cdot\log^{1.5}k\cdot\norm{H}^3}
      \end{eqnarray*}
    where the first inequality is from the assumption that $\abs{\set{i:x_i<0}}\leq 3\log k$ and the second inequality is from Fact~\ref{fac:benktus2} Eq.~\eqref{eqn:h3}.
    \end{proof}

    \subsection{Upper bounding the second term ($\P$) in Eq.~\eqref{eqn:rk2}}
    \label{sec:2.4}

    Let us rewrite ($\P$) as a sum of two term as in Eq.~\eqref{eq:Pterm}. We first upper bound the first easy term.
    \begin{lemma}\label{lem:444}
    \[
    \abs{\sum_{i_1\neq i_2\neq i_3:\atop x_{i_3}>x_{i_2}>x_{i_1}}\frac{\braket{i_3}'-\braket{i_1}'}{[i_3-i_1]}\cdot\frac{\braket{i_3}-\braket{i_2}}{[i_3-i_2]}\cdot\braket{i_3}}\leq O\br{\Delta\cdot\log^2 k\cdot\norm{H}^3}
    \]
    \end{lemma}
    \begin{proof}[Proof of Lemma~\ref{lem:444}]
      We split the summation into two cases that $x_{i_2}\geq 0$ and $x_{i_2}<0$. For the case that $x_{i_2}\geq 0$, we apply the mean value theorem to $g\br{\cdot}$ and Eq.~\eqref{eqn:mvtg'}, it is upper bounded by $O\br{\Delta\cdot\log k\cdot\norm{H}^3}$. For the case that $x_{i_2}<0$, we have $x_{i_1}<0$. Note that $\abs{g'\br{\cdot}}\leq 1$. Thus
      it is upper bounded by
      \[O\br{\sum_{i_1\neq i_2:\atop x_{i_1}<x_{i_2}<0}\sum_{i_3}\abs{H_{i_1,i_2}H_{i_2,i_3}H_{i_3,i_1}}}\leq O\br{\log^2k\cdot\norm{H}^3}.\]
    \end{proof}
    Next, our goal is to prove an upper bound on the second term in Eq.~\eqref{eq:Pterm}.

    \begin{lemma}\label{lem:555}
      \[\abs{\sum_{i_1\neq i_2\neq i_3:\atop x_{i_3}>x_{i_2}>x_{i_1}}\frac{\frac{\braket{i_3}'-\braket{i_1}'}{[i_3-i_1]}-\frac{\braket{i_2}'-\braket{i_1}'}{[i_2-i_1]}}{[i_3-i_2]}\braket{i_2}\braket{i_3}}\leq O\br{\sqrt{\log k}\cdot\norm{H}^3}.\]
    \end{lemma}

    Before we prove this lemma, we first prove a ``simpler" proposition which will be crucial in upper bound the above.

    \begin{proposition}\label{claim:v5}
    \[\abs{\sum_{i_1\neq i_2\neq i_3}\frac{\frac{\braket{i_1}'-\braket{i_3}'}{[i_3-i_1]}-\frac{\braket{i_1}'-\braket{i_2}'}{[i_2-i_1]}}{[i_3-i_2]}}\leq O\br{\Delta^2\cdot\sqrt{\log k}\cdot\norm{H}^3}\]
    \end{proposition}
    \begin{proof}
    Using Fact~\ref{fac:fredivided},
    \begin{eqnarray*}
      \abs{\sum_{i_1\neq i_2\neq i_3}\frac{\frac{g'\br{x_{i_3}}-g'\br{x_{i_1}}}{x_{i_3}-x_{i_1}}-\frac{g'\br{x_{i_2}}-g'\br{x_{i_1}}}{x_{i_2}-x_{i_1}}}{x_{i_3}-x_{i_2}}G\br{x}H_{i_1,i_2}H_{i_2,i_3}H_{i_3,i_1}}=\frac{1}{\sqrt{2\pi}}\abs{\Tr\Br{ D^2\br{e^{-\frac{X^2}{2}}}[H,H]\cdot H}}G\br{x},
    \end{eqnarray*}
    where $X=\diag\br{x_1,\ldots, x_n}$.
    Using Lemma~\ref{lem:ex2derivative}, it suffices to upper bound
    \begin{equation}\label{eqn:ex2d1}
      G\br{x}\abs{\Tr \Br{ e^{-\frac{uX^2}{2}}\br{XH+HX}e^{-\frac{v\br{1-u}X^2}{2}}\br{XH+HX}e^{-\frac{\br{1-v}\br{1-u}X^2}{2}}H}}
    \end{equation}
    and
    \begin{equation}\label{eqn:ex2d2}
      G\br{x}\abs{\Tr\Br{e^{\frac{\br{u-1}X^2}{2}}H^2e^{-\frac{uX^2}{2}}H}}
    \end{equation}
    Note that $u+v\br{1-u}+\br{1-v}\br{1-u}$=1. At least two of these three quantities are at most $\frac{1}{2}$.  We upper bound Eq.~\eqref{eqn:ex2d1} in the following three cases.

    If $u\leq\frac{1}{2}$ and $\br{1-u}\br{1-v}\leq\frac{1}{2}$, using Claim~\ref{claim:xhxh}, Lemma~\ref{lem:xhxh} and the fact that $\twonorm{He^{-X^2/2}}=\twonorm{e^{-X^2/2}H}$ as $H$ is symmetric Eq.~\eqref{eqn:ex2d1} is upper bounded by
    \begin{eqnarray*}
      \big\|\br{XH+HX}e^{-\frac{X^2}{2}}\big\|_2^2\norm{H}\leq16 \Delta^2\twonorm{He^{-X^2/2}}^2\cdot\norm{H}.
    \end{eqnarray*}

    If $u\leq\frac{1}{2}$ and $v\br{1-u}\leq\frac{1}{2}$, then the Eq.~\eqref{eqn:ex2d1} is upper bounded by
    \begin{eqnarray*}
    	\twonorm{\br{XH+HX}e^{-\frac{X^2}{2}}}\cdot\twonorm{He^{-\frac{X^2}{2}}}\norm{XH+HX}&\leq& 2\Delta\twonorm{He^{-X^2/2}}^2\cdot \norm{XH+HX}\\
    	&\leq& 4\Delta^2\twonorm{He^{-X^2/2}}^2\cdot\norm{H}.
    \end{eqnarray*}
    where the second last inequality is by Claim~\ref{claim:xhxh}. The case that $u\br{1-v}\leq\frac{1}{2}$ and $v\br{1-u}\leq\frac{1}{2}$ follows similarly.    Also Eq.~\eqref{eqn:ex2d2} can be upper bounded with similar arguments. Thus
    \begin{equation}\label{eqn:normx2h2x2}
      G\br{x}\cdot \abs{\Tr~e^{\br{u-1}X^2/2}H^2e^{-uX^2/2}H}\leq 16\Delta^2 G\br{x}\norm{H}\cdot \twonorm{He^{-X^2/2}}^2.
    \end{equation}
    Therefore,
    \begin{eqnarray}
      &&\abs{\sum_{i_1\neq i_2\neq i_3}\frac{\frac{g'\br{x_{i_3}}-g'\br{x_{i_1}}}{x_{i_3}-x_{i_1}}-\frac{g'\br{x_{i_2}}-g'\br{x_{i_1}}}{x_{i_2}-x_{i_1}}}{x_{i_3}-x_{i_2}}G\br{x}H_{i_1,i_2}H_{i_2,i_3}H_{i_3,i_1}}\\
       \nonumber\\
       &\leq&\br{G(x)\cdot \Tr\Br{ D^2\br{e^{-\frac{X^2}{2}}}[H,H]\cdot H}}\nonumber\\
       &\leq&O\br{\Delta^2 G\br{x}\norm{H}\twonorm{He^{-X^2/2}}^2}\nonumber\\
      &=&O\br{\Delta^2\sum_{i_1,i_2}e^{-x_{i_2}^2}H_{i_1,i_2}^2G\br{x}\cdot\norm{H}}\nonumber\\
      &\leq &O\br{\Delta^2\sum_{i_1,i_2}g'\br{x_{i_2}}G\br{x}H_{i_1,i_2}^2\cdot\norm{H}}\quad\nonumber\\
      &\leq&O\br{\Delta^2\sum_{i_1,i_2}g'\br{x_{i_2}}G\br{x_{-i_2}}H_{i_1,i_2}^2\cdot\norm{H}}\nonumber\\
      &\leq&O\br{\Delta^2\onenorm{G^{(1)}\br{x}}\cdot\br{\max_{i_2}\sum_{i_1}H_{i_1,i_2}^2}\cdot\norm{H}}\nonumber\\
      &\leq&O\br{\Delta^2\onenorm{G^{(1)}\br{x}}\cdot\max_{i_2}\br{H^2}_{i_2,i_2}\cdot\norm{H}}\nonumber\\
      &\leq&O\br{\Delta^2\cdot\sqrt{\log k}\cdot\norm{H}^3},\label{eqn:3dgh}
    \end{eqnarray}
    where the second inequality used $e^{-x_i^2/2}\leq 1$, third inequality used $g(x)\in [0,1]$ and the last inequality is from Fact~\ref{fac:benktus2}.
    \end{proof}

    We are now ready to prove the main lemma.   Note that end of the day we need to bound the inequality in Lemma~\ref{lem:555} which can be written as 
      \begin{align}
      \label{eq:toprovelastcase}
      \abs{\sum_{i_1\neq i_2\neq i_3\atop i_1<i_2,i_1<i_3}\frac{\frac{g'\br{x_{i_1}}-g'\br{x_{i_3}}}{x_{i_3}-x_{i_1}}-\frac{g'\br{x_{i_1}}-g'\br{x_{i_2}}}{x_{i_2}-x_{i_1}}}{x_{i_3}-x_{i_2}}G\br{x_{-\set{i_1}}}H_{i_1,i_2}H_{i_2,i_3}H_{i_3,i_1}}\leq O\br{\Delta^2\cdot\log^{2.5}k\cdot\norm{H}^3}
      \end{align}

    Observe that in this section we are concerned with $x_{i_1}<x_{i_2}<x_{i_3}$ so the summation in this lemma and the equation above are over the same indices.

    \begin{proof}[Proof of Lemma~\ref{lem:555}] By the paragraph above, proving this lemma is equivalent to proving Eq.~\eqref{eq:toprovelastcase}.
        The left hand side of Eq.~\eqref{eq:toprovelastcase} can be expressed as
      \begin{eqnarray}
        &&\abs{\frac{1}{\sqrt{2\pi}}\sum_{i_1}G\br{x_{-i_1}}\br{\Tr~D^2\br{e^{-X^2/2}}\Br{A^{i_1},\br{A^{i_1}}^T}H}}\nonumber=\abs{\frac{1}{\sqrt{2\pi}}\sum_{i_1:x_{i_1}<0}\br{\cdots}+\frac{1}{\sqrt{2\pi}}\sum_{i_1:x_{i_1}\geq 0}\br{\cdots}}\nonumber\\
        &&\label{eqn:gi3}
        \end{eqnarray}
For the first summation above, let
        \begin{equation*}
          \br{A^{i_1}}_{i,j}=\begin{cases}
                                H_{i,i_1}, & \mbox{if $j=i_1$ and $i>i_1$}  \\
                                0, & \mbox{otherwise}.
                              \end{cases}
        \end{equation*}
        Note that $\twonorm{He^{-X^2/2}}=\twonorm{e^{-X^2/2}H}$ as $H$ is symmetric.  
        Using the same argument as Eq.~\eqref{eqn:ahzero}, we have
        \[\max\set{\twonorm{A^{i_1}e^{-X^2/2}},\twonorm{\br{A^{i_1}}^Te^{-X^2/2}},\twonorm{e^{-X^2/2}A^{i_1}},\twonorm{e^{-X^2/2}\br{A^{i_1}}^T}}\leq\twonorm{He^{-X^2/2}}.\]
        Further notice that  $\twonorm{A^{i_1}e^{-X^2/2}}\leq\twonorm{He^{-X^2/2}},\norm{A^{i_1}}\leq\norm{H},
    $. Following the same proof of Proposition~\ref{claim:v5},  we can upper bound the first summation in Eq.~\eqref{eqn:gi3}~by
    \begin{eqnarray*}
          &&\frac{1}{\sqrt{2\pi}}\sum_{i_1:x_{i_1}<0}G\br{x_{-i_1}}\twonorm{He^{-X^2/2}}^2\cdot\norm{H}\leq O\br{\Delta^2\cdot \log^{2.5}k \norm{H}^3},
        \end{eqnarray*}
        where the inequality follows from the argument in Eq.~\eqref{eqn:firstsummation}.    For the second summation, define
    \begin{equation*}
      B_{i_1,i_2}=\begin{cases}
                          \frac{H_{i_1,i_2}}{\sqrt{g\br{x_{i_1}}}}, & \mbox{if $x_{i_1}\geq0$} \\
                          0, & \mbox{otherwise}.
                        \end{cases}
    \end{equation*}
    Note that 
    $$\max\set{\twonorm{Be^{-X^2/2}},\twonorm{B^Te^{-X^2/2}},\twonorm{e^{-X^2/2}B},\twonorm{e^{-X^2/2}B^T}}\leq\sqrt{2}\twonorm{He^{-X^2/2}}$$ and $\norm{B}\leq\sqrt{2}\norm{H}$ (since $g(x)\geq 1/2$ for $x\geq 0$).
    Then the second summation in Eq.~\eqref{eqn:gi3} is equal to
    \begin{eqnarray*}
      \abs{G\br{x}\frac{1}{\sqrt{2\pi}}\Tr~D^2\br{e^{-X^2/2}}\Br{B,B^T}A}\leq\frac{1}{\sqrt{\pi}}G\br{x}\twonorm{He^{-X^2/2}}\norm{H}\leq O\br{\Delta^2\cdot\log^{1.5} k\cdot\norm{H}^3}
    \end{eqnarray*}
    where the inequality follows from the argument in Eq.~\eqref{eqn:firstsummation}.
    \end{proof}
    \end{document}